\documentclass[11pt]{article}
\usepackage{amsmath}
\usepackage{graphicx,psfrag,epsf}
\usepackage{enumerate}
\usepackage{natbib}
\usepackage{url} % not crucial - just used below for the URL 
\usepackage{mathtools}

\newtagform{brackets}{[}{]}
\usetagform{brackets}

\usepackage{setspace}
 \DeclareGraphicsExtensions{.eps, .ps}
\usepackage{soul,color}
\usepackage{amsfonts}
\usepackage{amssymb}
\usepackage{newfloat}
\usepackage{multirow}
\usepackage{arydshln}
\usepackage{comment}
\usepackage{enumitem}
\usepackage{amsthm}
\usepackage[title]{appendix}
\usepackage{lineno}
\usepackage{booktabs}
\usepackage{tikz}
\usepackage{comment}
\usepackage{xr}
\usetikzlibrary{positioning}
\usepackage{subcaption}
\usepackage{authblk}
\usepackage{adjustbox}
\usepackage{etoolbox}

\usepackage{xcolor}
\usepackage{hyperref}

\newcommand{\blind}{1}

\usepackage[margin=1in]{geometry}

\newtheorem{theorem}{Theorem}[section]
\newtheorem{example}{Example}[section] % [section] ensures 
\newtheorem{corollary}{Corollary}[theorem]
\newtheorem{lemma}[theorem]{Lemma}
\newtheorem{assumption}[theorem]{Assumption}
\newtheorem{remark}[theorem]{Remark}

\DeclareMathOperator{\Var}{\mathrm{Var}}
\DeclareMathOperator{\Cov}{\mathrm{Cov}}

\newcommand{\naive}{na\"{\i}ve }
\newcommand{\Naive}{Na\"{\i}ve }

\begin{document}

\def\spacingset#1{\renewcommand{\baselinestretch}%
{#1}\small\normalsize} \spacingset{1}

%%%%%%%%%%%%%%%%%%%%%%%%%%%%%%%%%%%%%%%%%%%%%%%%%%%%%%%%%%%%%%%%%%%%%%%%%%%%%%

\if1\blind
{
  \title{\bf Selective Inference in Growth Curve Models}
  \author[a]{Haochen Lei}
    \author[b]{Qian Zhang}
    \author[c]{Hongyuan Cao}
\affil[a]{Department of Statistics, Florida State University}
\affil[b]{Department of Educational Psychology and Learning Systems, Florida State University}
\affil[c]{Department of Biostatistics and Medical Informatics, University of Wisconsin-Madison}

  \maketitle
} \fi

\if0\blind
{
  \bigskip
  \bigskip
  \bigskip
  \begin{center}
    {\LARGE\bf Selective Inference in Growth Curve Models}
\end{center}
  \medskip
} \fi

\bigskip
\begin{abstract}
Growth curve models are widely used in psychological research, and variable
selection can help identify baseline characteristics associated with
longitudinal heterogeneity. However, conventional inference after
data-driven variable selection can be invalid because the same outcome data
are used for both selection and inference. We develop a data-fission
framework for post-selection inference in growth-curve models that separates
the information used for selection and inference while retaining all
participants in both stages through the addition and subtraction of Gaussian
noise. The framework accommodates flexible variable-selection procedures and
targets covariance-weighted linear projection parameters in the selected
working model. It provides exact inference when the covariance structure is
known, and we establish asymptotic validity under suitable regularity
conditions when the covariance structure is estimated. Simulations show that
the proposed method provides valid inference and improves efficiency over
subject-level data splitting, whereas \naive post-selection inference can be
biased. An application to the Longitudinal Study of American Youth
illustrates the proposed framework.
\end{abstract}

\noindent%
\vfill

\newpage
\spacingset{1.45} % DON'T change the spacing!
\section{Introduction}
\subsection{Variable Selection and Inference in Longitudinal Growth-Curve Models}

Growth curve models are widely used in psychological research to study
individual trajectories and heterogeneity in developmental change
\citep{bollen2006latent}. Modern studies may also include large numbers of
candidate predictors, such as genetic, neuroimaging, or environmental
measures \citep{golub1999molecular}, making variable selection an important
part of the modeling process.

Variable selection can help identify baseline characteristics associated
with different aspects of longitudinal heterogeneity, including baseline
levels and rates of change in child development \citep{mccoy2016early},
depressive-symptom trajectories \citep{xiang2022prediction}, conduct-problem
trajectories \citep{Odgers2008}, and cognitive decline in older adulthood
\citep{yaffe2004metabolic}. More generally, reducing the set of candidate
predictors can improve interpretability and focus subsequent inference on
the covariates most relevant to trajectory heterogeneity.

After identifying a subset of variables, a common next step is to fit a standard longitudinal model, such as a growth curve model, using the selected covariates with the same data and report the corresponding coefficient estimates and confidence intervals. Although this second-stage analysis appears straightforward, it raises two distinct questions. First, because the variables were selected using the same outcome data, can conventional inference still be trusted? Second, when the selection procedure is flexible and the relationships
between the baseline covariates and the growth parameters are nonlinear,
what do the coefficients in the subsequently fitted growth curve model
actually represent?

\subsection{Post-Selection Inference and the Remaining Gap}
Unfortunately, the answer to the first question is generally no. When the
same data are used for both variable selection and subsequent inference,
conventional inferential procedures that ignore the selection step cannot
remain valid. The resulting estimates and confidence intervals may reflect not only the
underlying associations but also the data-dependent selection process,
leading to distorted effect estimates and confidence intervals with coverage
below their nominal level in both standard linear models
\citep{berk2013valid,kuchibhotla2022post} and longitudinal or linear
mixed-effects models
\citep{rugamer2022selective,d2026methods}.

Although post-selection inference has been studied extensively for standard
linear models \citep{berk2013valid,lee2016exact,tian2018selective, rasines2023splitting}, comparatively fewer methods have been developed
specifically for linear mixed-effects models. For a recent review and
empirical comparison of the existing methods, see \citet{d2026methods}. Several different strategies have been proposed for post-selection
inference in linear mixed-effects models.
For example, \citet{kramlinger2023uniformly}
developed Lasso-based confidence regions, while \citet{claeskens2021post}
studied selection based on conditional AIC and constructed post-selection
regions using Monte Carlo methods.

However, these approaches rely on specific properties of the selection
procedure and are therefore difficult to extend to flexible machine-learning
selectors, such as mixed-effects trees or sparse neural networks, whose
selection events may be analytically intractable.
\citet{rugamer2022selective} provided a more general Monte Carlo strategy by
repeatedly applying the selection procedure to simulated responses, but this
can become computationally expensive when the observed selection event has
low probability.

An alternative method is subject-level data splitting, in which one group of
participants is used for selection and an independent group is reserved for
inference. This approach places essentially no restrictions on the selection
procedure, but uses only a fraction of the available participants at each
stage and can therefore %lose selection power and 
reduce inferential precision.

Existing approaches therefore involve a trade-off between restrictions or
computational costs associated with the selection procedure and the loss of
efficiency caused by sample splitting. This motivates a method that can accommodate a flexible selector without repeatedly characterizing or reconstructing its selection event, while retaining as much information as possible from all participants.

\subsection{Post-Selection Inferential Targets}
Another important question is what we are actually estimating after a set of
features has been selected. When the relationships between the baseline
covariates and the growth parameters, such as the intercept and rate of
change, are nonlinear, there may not exist a unique ``true'' sparse linear
coefficient vector for the selection procedure to recover. Instead, we treat
the selected linear mixed-effects model as a working approximation and
interpret its coefficients as linear projection parameters. This perspective
is consistent with the broader view of regression models as useful
approximations even when the parametric model is misspecified
\citep{buja2019modelsI, buja2019modelsII}, and is closely related to the
projection-based targets considered in post-selection inference with
randomized splitting \citep{rasines2023splitting}.

Depending on the inferential context, we consider two related projection
targets. First, conditioning on the covariates and observation times actually
observed yields a finite-sample, design-conditional projection target. Second,
when the covariates and observation schedule are regarded as random, averaging
over their distribution yields a population-level projection target. The
design-conditional target is particularly natural for post-selection inference
because it describes the best covariance-weighted linear approximation for the
realized sample and selected working model, without requiring the relationships between the baseline covariates and growth parameters to be linear.

Once the inferential target has been defined, the remaining challenge is how
to conduct a valid inference for these projection parameters after
data-dependent variable selection.
\subsection{Proposed Framework}

Data fission provides another way to separate the information used for
selection and inference without assigning different participants exclusively
to the two stages \citep{leiner2025data}. The general idea is to introduce
additional randomization and construct two related responses from the
observed data: one for model selection and the other for subsequent inference.

Extending data fission to longitudinal data introduces an additional
challenge. Repeated observations from the same participant are correlated,
so the randomization must account for the within-subject covariance
structure. %Moreover, 
This covariance structure is typically unknown and
must be estimated from the observed data. 
%A post-selection procedure for
%growth curve models must therefore accommodate both the dependence among repeated observations and
%the uncertainty associated with estimating their covariance structure.

In this paper, we develop a post-selection inference framework for growth curve models represented through linear mixed-effects models. Our goal is
to provide valid inference for the projection parameters in a selected
working model while allowing the variable-selection procedure to remain
flexible. In particular, the proposed framework avoids the need to explicitly
characterize or repeatedly reconstruct the selection event and, unlike
subject-level data splitting, retains all participants in both the selection
and inference stages.

%When the covariance matrix
%$\Sigma_0$ of the longitudinal outcome is known, we generate an auxiliary
%Gaussian random vector
% \(
%W\sim N(0,\Sigma_0),
%W\perp Y,
%\)
%and, for a fixed information-allocation parameter $\tau>0$, construct $Y^{\mathrm{sel}}$ used for variable selection and $Y^{\mathrm{inf}}$ reserved for post-selective inference as
%\(
%Y^{\mathrm{sel}} = Y + \tau W\) and \(Y^{\mathrm{inf}} = Y - \tau^{-1}W,
%\)
%where $Y$ is the observed longitudinal response and $W$ is the auxiliary noise variable.
%The covariance of the added noise and the relative scaling of the two
%perturbations are chosen so that
%\(
%\operatorname{Cov}
%\left(
%Y^{\mathrm{sel}},
%Y^{\mathrm{inf}}
%\right)
%=0.
%\)
%Under the Gaussian mixed-effects model, $Y^{\mathrm{sel}}$ and $Y^{\mathrm{inf}}$ are therefore
%independent. $Y^{\mathrm{sel}}$ can be used with an arbitrary variable-selection
%procedure, while $Y^{\mathrm{inf}}$ is reserved exclusively for estimation and
%inference after selection.

When the covariance matrix $\Sigma_0$ is known, we generate
$W\sim N(0,\Sigma_0)$ independently of $Y$ and construct
\[
Y^{\mathrm{sel}}=Y+\tau W,
\qquad
Y^{\mathrm{inf}}=Y-\tau^{-1}W.
\]
%\textcolor{red}{what is Their?}Their cross-
The covariance between $Y^{\mathrm{sel}}$ and $Y^{\mathrm{inf}}$ is zero, so under the Gaussian model the two randomized
responses are independent. The selection response may therefore be used with
a flexible variable-selection procedure, while the inference response is
reserved for estimation after selection.

We further consider the practically important setting in which the covariance
matrix is unknown. In this case, the covariance components are estimated from
the observed longitudinal data, and the estimated covariance matrix is used to
construct the randomized selection and inference responses. Because the
covariance estimate itself depends on the observed outcome, these empirical
responses are no longer independent in finite samples. Under suitable
conditions on covariance estimation, we show that the conditional distribution
of the post-selection estimator based on the estimated covariance is
asymptotically equivalent to that of the oracle procedure based on the true
covariance matrix.
%We show that, under suitable
% conditions %on covariance estimation, the selected design, and the probability of the selection event,
% the conditional distribution of estimated target using the estimated covariance converge to those using the true covariance matrix.

Our framework places few restrictions on the variable-selection algorithm
itself. To illustrate this flexibility, we consider two substantially
different working selectors: a mixed-effects model tree and a sparse neural
network. Both methods allow the selection stage to capture relationships
that may be more complex than those represented in the linear
working model.
%
%The remainder of the paper is organized as follows. Section~2 introduces the
%growth curve model, the general selection mechanism, and the
%design-conditional and population projection targets. Section~3 develops the
%oracle data fission procedure when the covariance structure is known.
%Section~4 considers the unknown-covariance setting and establishes the
%asymptotic validity of the empirical procedure. Section~5 describes the
%tree-based and neural-network working selectors. Section~6 presents the
%simulation study. Section~7 presents the empirical application, and
%Section~8 concludes with a discussion. Technical proofs are provided in the Appendix.

\section{Model Framework and Inferential Targets}
\subsection{Growth Curve Model}
In standard longitudinal data analysis, the linear mixed-effects model \citep{laird1982random} is routinely employed to characterize developmental trajectories. However, when the number of candidate covariates %such as intensive behavioral measures, 
 is large, specifying a parametric linear relationship a priori is often restrictive. We therefore start from a semiparametric growth curve representation that accommodates arbitrary nonlinear main effects and interactions.

Consider longitudinal data collected from $N$ independent participants (subjects) indexed by $i = 1, \ldots, N$. For participant $i$, data are collected across $T_i$ distinct time points. Let $t_i = (t_{i1}, t_{i2}, \ldots, t_{iT_i})^\top \in \mathbb{R}^{T_i}$ denote the vector of time points for participant $i$. Incorporating subject-specific total observation counts $T_i$ and timing vectors $t_i$ accommodates unbalanced longitudinal designs where participants are observed at varying total frequencies or irregular time intervals. Let $x_i = (x_{i1}, \ldots, x_{ip})^\top \in \mathbb{R}^p$ represent a $p$-dimensional vector of time-invariant, baseline covariates for participant $i$. Following the latent growth curve framework \citep{meredith1990latent, bollen2006latent}, we characterize the outcome trajectory of participant $i$ at continuous time $t$ as a linear developmental process:
\begin{equation*}
	Y_i(t) = \eta_{1i} + \eta_{2i} t + \epsilon_i(t),
\end{equation*}
where $\eta_{1i} \in \mathbb{R}$ is the latent intercept representing the subject-specific baseline level at $t = 0$, $\eta_{2i} \in \mathbb{R}$ is the latent slope representing the subject-specific rate of change over time, and $\epsilon_i(t) \in \mathbb{R}$ is the measurement error at time $t$. We collect individual trajectory parameters into the vector $\eta_i = (\eta_{1i}, \eta_{2i})^\top \in \mathbb{R}^2$. To avoid imposing restrictive linear constraints on how baseline covariates relate to trajectory dynamics, we model the relationship nonparametrically:
\begin{equation}
\label{eq:eta}
\eta_i = m(x_i) + b_i, \qquad b_i \sim N(0, D),
\end{equation}
where $m(x_i) = \{m_1(x_i), m_2(x_i)\}^\top \in \mathbb{R}^2$ represents unknown, flexible mean functions. Here, $m_1(x_i): \mathbb{R}^p \to \mathbb{R}$ captures how baseline covariates impact the overall average baseline level (intercept), while $m_2(x_i): \mathbb{R}^p \to \mathbb{R}$ captures how baseline covariates impact the average rate of change over time (slope). The vector $b_i = (b_{1i}, b_{2i})^\top \in \mathbb{R}^2$ contains unobserved, subject-specific random effects accounting for individual heterogeneity in baseline status ($b_{1i}$) and growth rate ($b_{2i}$) unexplained by $x_i$. The matrix $D \in \mathbb{R}^{2 \times 2}$ is the positive-definite covariance matrix of random trajectory deviations:
{
	\renewcommand{\arraystretch}{0.8}
\begin{equation}
\label{eq:D}
D = \begin{pmatrix} d_{11} & d_{12} \\ d_{12} & d_{22} \end{pmatrix} = \begin{pmatrix} \operatorname{Var}(b_{1i}) & \operatorname{Cov}(b_{1i}, b_{2i}) \\ \operatorname{Cov}(b_{1i}, b_{2i}) & \operatorname{Var}(b_{2i}) \end{pmatrix},
\end{equation}
}
where $d_{11}$ denotes residual intercept variance, $d_{22}$ denotes residual slope variance, and $d_{12}$ reflects covariance between an individual's initial status and rate of growth. We assume $b_i$ is independent of covariates $x_i$ and independent across subjects ($b_i \perp b_j$ for $i \neq j$). Let $Y_i = (Y_i(t_{i1}), \ldots, Y_i(t_{iT_i}))^\top \in \mathbb{R}^{T_i}$ denote the vector of repeated outcome measurements for subject $i$. We construct the subject-specific design matrix $Z_i \in \mathbb{R}^{T_i \times 2}$ as:
%
%Z_i = \begin{pmatrix} 1 & t_{i1} \\ 1 & t_{i2} \\ \vdots & \vdots \\ 1 & t_{iT_i} 
%\end{equation}
\begin{equation}
	\label{eq:Z_i}
Z_i=\begin{pmatrix}1_{T_i},t_i\end{pmatrix},
\end{equation}
where the first column maps the latent intercept to each time point and the second column maps observed measurement times $t_i$ to the latent slope. The observed outcome vector $Y_i$ is expressed in matrix notation as:
\begin{equation}
Y_i = Z_i \eta_i + \epsilon_i = Z_i m(x_i) + Z_i b_i + \epsilon_i,
\end{equation}
where $\epsilon_i = (\epsilon_i(t_{i1}), \ldots, \epsilon_i(t_{iT_i}))^\top \in \mathbb{R}^{T_i}$ represents measurement errors for subject $i$, assuming $\epsilon_i \sim N(0, R_i)$ independent of $b_i$. 
The matrix $R_i \in \mathbb{R}^{T_i \times T_i}$ models the within-person error covariance structure. 
%Under conditionally independent measurement errors, $R_i = \sigma^2 I_{T_i}$, where $I_{T_i}$ denotes the $T_i \times T_i$ identity matrix, assuming error variances are constant ($\sigma^2$) across time and uncorrelated within individuals. Under continuous-time autocorrelation structures, $(R_i)_{jk} = \sigma^2 \exp(-\rho |t_{ij} - t_{ik}|)$, where $(R_i)_{jk}$ represents the entry in the $j$-th row and $k$-th column of $R_i$. 
For ease of exposition, we assume $R_i = \sigma^2 I_{T_i}$ hereafter. Marginalizing over random trajectory deviations $b_i$, the conditional distribution of $Y_i$ given covariates $x_i$ and measurement times $t_i$ is multivariate normal:
 \begin{equation}
\label{eq:Y_i}
Y_i \mid x_i, t_i \sim N \Big( \mu_i, \; \Sigma_i \Big),
\end{equation} where the conditional mean $\mu_i=Z_im(x_i)$, and the conditional variance-covariance matrix $\Sigma_i \in \mathbb{R}^{T_i \times T_i}$ decomposes into between-person structural variation and within-person error variance:
\begin{equation}
\label{eq:Sigma_i}
\Sigma_i = Z_i D Z_i^\top + R_i.
\end{equation}
% where the total conditional variance-covariance matrix $\Sigma_i \in \mathbb{R}^{T_i \times T_i}$ decomposes into between-person structural variation and within-person error variance:
% \begin{equation}
% \Sigma_i = Z_i D Z_i^\top + R_i.
% \end{equation}
While all subjects share global variance parameters $D$ and $\sigma^2$, subject-specific differences in measurement schedules $t_i$ and observation counts $T_i$ yield distinct matrix dimensions and structures $\Sigma_i$.  Let $n=\sum_{i=1}^N T_i$ be the total number of observations.
Stacking across $N$ participants, let
\[
Y=(Y_1^\top,\ldots,Y_N^\top)^\top,\qquad
\mu=(\mu_1^\top,\ldots,\mu_N^\top)^\top,
\]
and
\[
\Sigma=\operatorname{blockdiag}(\Sigma_1,\ldots,\Sigma_N),
\]
where $\operatorname{blockdiag}(\Sigma_1,\ldots,\Sigma_N)$ denotes the
block-diagonal matrix with subject-specific covariance matrices
$\Sigma_1,\ldots,\Sigma_N$ on the diagonal blocks and zero on the off-diagonal blocks.
%Stacking all $N$ participants yields an overall sample response vector $Y = (Y_1^\top, Y_2^\top, \ldots, Y_N^\top)^\top \in \mathbb{R}^{n}$, where $n = \sum_{i=1}^N T_i$ is the total number of observations in the study. Similarly, we define the stacked sample mean vector $\mu \in \mathbb{R}^n$:
%% \begin{equation}
%% \mu = \Big( (Z_1 m(x_1))^\top, (Z_2 m(x_2))^\top, \ldots, (Z_N m(x_N))^\top \Big)^\top,
%% \end{equation}
%{{blue}
%\begin{equation}
%\mu = \Big( (Z_1 m(x_1))^\top, (Z_2 m(x_2))^\top, \ldots, (Z_N m(x_N))^\top \Big)^\top=\Big(\mu_1^\top,\ldots, \mu_N^\top \Big)^\top,
%\end{equation}
%}
%and the block-diagonal overall covariance matrix $\Sigma \in \mathbb{R}^{n \times n}$:
%\begin{equation}
%\Sigma = \operatorname{blockdiag}(\Sigma_1, \Sigma_2, \ldots, \Sigma_N),
%\end{equation}
%{{blue}
%where $\operatorname{blockdiag}(\Sigma_1,\ldots,\Sigma_N)$ denotes the
%block-diagonal matrix with subject-specific covariance matrices
%$\Sigma_1,\ldots,\Sigma_N$ on the diagonal blocks and zero on the off-diagonal blocks.
%}

Because subjects are assumed independent, off-diagonal block structures equal zero matrices. The joint conditional distribution across the entire sample is:
\begin{equation}
Y \mid \{x_i\}_{i=1}^N, \{t_i\}_{i=1}^N \sim N(\mu, \Sigma).
\end{equation}
We parameterize the overall covariance structure as $\Sigma = \Sigma(\theta_0)$, where:
\begin{equation}
	\label{eq:theta0}
	\theta_0 = (\sigma^2, D_{11},D_{12},D_{21},D_{22})\in \mathbb{R}^5
\end{equation}
collects variance-covariance components governing measurement noise ($\sigma^2$) and latent growth heterogeneity ($D$).

\begin{remark}
	The Gaussian assumption is strictly imposed on the marginal error distribution $N(0, \Sigma)$. This parametric noise structure is mathematically necessary to construct the independent data fission perturbations \citep{tian2018selective,leiner2025data} that guarantee valid post-selection inference. However, we do not assume that the true data-generating process follows a sparse linear model, and leave the mean functions $m_1$ and $m_2$ completely unspecified. 
\end{remark}

%This nonparametric mean formulation naturally raises two sequential methodological challenges that motivate the remainder of this section. First, we must mathematically formalize a generic data-driven selection algorithm used to identify a parsimonious set of important covariates from a high-dimensional pool (Section 2.2).  Second, we must rigorously define the parameter of interest when fitting a linear mixed-effects working model to this selected subset, given that the underlying reality $m(x)$ is nonlinear (Section 2.3).

\subsection{Data-Driven Variable Selection}
\label{sec:data_driven_vs}
In high-dimensional longitudinal studies, exploratory analyses often perform data-driven variable selection to identify subsets of predictors associated with the developmental trajectory. Each candidate predictor $x_j (j=1, \ldots, p)$ may influence the intercept function $m_1(x)$, the slope function $m_2(x)$, or both. We index the candidate feature pool by $\{1, \ldots, p\}$ for the intercept and $\{p+1, \ldots, 2 p\}$ for the slope.

%Let $\mathcal {S} : R ^ \rightarrow 2^{\{1, \ldots, 2 p\}}$ denote a generic selection algorithm that maps the stacked response vector $Y$ to an active subset $M \subseteq\{1, \ldots, 2 p\}$. We can decompose $M$ into $M=$ $M_I \cup\left\{p+j: j \in M_S\right\}$, where $M_I\subseteq \{1, \ldots, p\}$ and $ M_S \subseteq\{1, \ldots, p\}$ represent the selected covariates for the intercept and slope, respectively.

Let $\mathcal{S} : \mathbb{R}^n \to 2^{\{1, \ldots, 2p\}}$ denote a generic selection algorithm that maps the overall observed response vector $Y \in \mathbb{R}^n$ to an active subset of selected features $M \subseteq \{1, \ldots, 2p\}$. We can decompose any selected subset into $M = M_I \cup \{p + j : j \in M_S\}$, where $M_I \subseteq \{1, \ldots, p\}$ and $M_S \subseteq \{1, \ldots, p\}$ denote the specific covariates selected for the intercept and slope, respectively.

A critical advantage of our subsequent inferential framework is its agnostic nature regarding the specific algorithmic form of $\mathcal{S}(\cdot)$. We do not impose restrictive structural assumptions—such as linearity or polyhedral selection constraints—on the selection procedure. We only require that for any candidate subset $M$, the selection event is measurable:
\begin{equation}
\label{eq:A_M}
\mathcal{A}_M \triangleq\left\{y \in R ^n: \mathcal {S} (y)=M\right\} .
\end{equation}

The collection $\left\{ \mathcal{A}_M: M \subseteq\{1, \ldots, 2 p\}\right\}$ forms a disjoint, measurable partition of the response space.  This flexible formulation seamlessly accommodates a wide array of modern feature screening tools utilized in longitudinal data analysis, including sparsity-inducing penalization methods (e.g., Group Lasso, \cite{yuan2006model}), hierarchical neural selectors (e.g., LassoNet \cite{lemhadri2021lassonet}), and tree-based mixed-effects algorithms (e.g., LMERtree \cite{fokkema2018detecting}). We fix internal algorithmic randomness (e.g., the random seed) when the algorithms involve stochastic components.

%For selection algorithms that involve stochastic components, such as random weight initialization in neural networks or random feature splitting in tree construction, we assume that any internal algorithmic randomness (e.g., the random seed) is fixed a priori. Conditioning on this fixed algorithmic state ensures that $\mathcal{S}(\cdot)$ remains a strictly deterministic and measurable mapping from the data space to the model space.

\subsection{Targets of Post-Selection Inference}
\label{sec:target}
Having defined the true model and the selection mechanism, we now rigorously specify the inferential target. Because the trajectory functions \(m_1(\cdot)\) and \(m_2(\cdot)\) are left unrestricted, the data-generating mechanism need not contain a “true” sparse linear coefficient vector.  This raises a fundamental question: after a model has been selected, what parameter should subsequent inference target?

In general, we use a linear mixed-effects model as a working approximation to the potentially nonlinear mean trajectory, with the resulting coefficients interpreted as linear projection parameters. Depending on how the covariates and observation times are treated, two related projection targets arise. First, we may condition on the observed covariates and observation times and treat the realized design as fixed. This leads to a finite-sample, design-conditional projection target that describes the best linear approximation to the conditional mean trajectory for the observed sample. Alternatively, we may regard the covariates and observation schedule as random and define the projection by averaging over their population distribution, leading to a population-level projection target. We consider both targets because they provide complementary interpretations of the selected working model.

\subsubsection{Design-Conditional Projection Target}
	Let
	\begin{equation}
		\label{eq:D_N}
		\mathcal{D}_N = \{x_i,t_i : i=1,\ldots,N\}
	\end{equation}
	denote the realized design, consisting of the baseline covariates and
	observation times for the $N$ participants. Consider a fixed candidate model
	\(
	M\subseteq\{1,\ldots,2p\},
	\)
	with corresponding intercept and slope index sets $M_I$ and $M_S$ defined
	in Section \ref{sec:data_driven_vs}. Throughout this subsection, $M$ is treated as fixed. The
	randomness induced by data-driven model selection is introduced subsequently.
	
	For subject $i$, let $x_{i,M_I}$ and $x_{i,M_S}$ denote the subvectors of
	baseline covariates indexed by $M_I$ and $M_S$, respectively. Define
	{
		\renewcommand{\arraystretch}{0.8}
	\begin{equation}
		\label{eq:C_M}
		C_M(x_i)
		=
		\begin{pmatrix}
			1 & x_{i,M_I}^\top & 0 & 0^\top\\
			0 & 0^\top & 1 & x_{i,M_S}^\top
		\end{pmatrix}.
	\end{equation}
}
	The corresponding longitudinal working design matrix is
	\[
	X_{i,M}=Z_iC_M(x_i)
	\in
	\mathbb{R}^{T_i\times(|M_I|+|M_S|+2)}.
	\]
%For subject $i$, let $x_{i, M_I} \in \mathbb{R}^{|M_I|}$ and $x_{i, M_S} \in \mathbb{R}^{|M_S|}$ denote the subvectors of the selected baseline predictors. The corresponding longitudinal design matrix (including global intercept and slope) for subject $i$ is defined as:
%\[ X_{i, M} = Z_i \begin{pmatrix} 1 & x_{i, M_I}^\top& 0  & 0^\top \\ 0 & 0^\top & 1 & x_{i, M_S}^\top \end{pmatrix} \in \mathbb{R}^{T_i \times (|M_I| + |M_S| + 2)}. \]
The columns correspond, in order, to the global intercept, the selected intercept predictors, the global time slope, and the selected covariate-by-time interaction terms.
Stacking the subject-specific matrices yields the full working design matrix:
\[ X_M = \begin{pmatrix} X_{1, M}^\top & \cdots & X_{N, M}^\top \end{pmatrix}^\top. \]
We restrict attention to candidate models for which $X_M$ has full
column rank.

Accounting for the true within-subject covariance structure $\Sigma$, we define the finite-sample post-selection inferential target $\beta_{M}^D = (\alpha_I^D, (\beta_{M_I}^D)^\top,\alpha_S^D, (\beta_{M_S}^D)^\top)$ as the best linear approximation of the true conditional mean $\mu$ within the selected model:
\begin{equation}
\label{eq:beta_D}
     \beta_M^D \triangleq \arg\min_{\beta \in \mathbb{R}^{|M_I|+|M_S|+2}} \mathbb{E}\left[ (Y - X_M \beta)^\top \Sigma^{-1} (Y - X_M \beta) \mid \mathcal{D}_N \right] = (X_M^\top \Sigma^{-1} X_M)^{-1} X_M^\top \Sigma^{-1} \mu.
\end{equation}
%where $\alpha_I^D$ denotes the overall trajectory level at $t=0$, $\alpha_S^D$ denotes the overall time slope, and $\beta_{M_I}^D$ and $\beta_{M_S}^D$ are the coefficients associated with the selected intercept and slope predictors.

For a fixed candidate model \(M\) and realized design \(D_N\), \(\beta_M^D\) is fixed. Once the model is selected data-adaptively, however, the corresponding post-selection target is random because its identity and dimension depend on the selected model. The following examples illustrate the relationship between the true underlying data-generating process and the selected working model.

\begin{example}[Correct Variable Selection] 
\label{examp:true}
To explicitly illustrate the construction of the finite-sample inferential target and the stacked design matrix, we consider a simplified longitudinal design with four participants ($N=4$). Each participant is measured at two identical time points, $t = (0, 1)^\top$. For simplicity, we assume that the marginal covariance matrix is $\Sigma = \sigma^2 I_8$.
Assume that the baseline covariates $x_i = (x_{i1}, x_{i2})^\top$ for the four participants are strictly mean-centered:
$x_1 = (-1, -2)^\top$, $x_2 = (-1, 0)^\top$, $x_3 = (0, 1)^\top$, and $x_4 = (2, 1)^\top$.
Suppose the true data-generating mechanism is governed by $m_1(x_1, x_2) = 1 + x_1$ and $m_2(x_1, x_2) = 2 + x_2$, yielding the outcome $Y_i(t) = (1+x_{i1}) + (2+x_{i2})t+ \epsilon(t)$. Evaluating this true mean trajectory for all participants across their respective time points yields the $8 \times 1$ stacked mean vector $\mu$:
\[ \mu_1 = (y_{10}, y_{11}, y_{20}, y_{21}, y_{30}, y_{31}, y_{40}, y_{41})^\top = (0, 0, 0, 2, 1, 4, 3, 6)^\top. \]

We first illustrate the case where the selection procedure correctly identifies $x_1$ for the intercept and $x_2$ for the slope. The selected model $M$ has $M_I = \{1\}$ and $M_S = \{2\}$. The working design matrix $X_M$ is constructed by stacking the subject-specific matrices. For instance, the sub-matrix for the first subject is:
{
	\renewcommand{\arraystretch}{0.8}
\[ X_{M1} = \begin{pmatrix} 1 & 0 \\ 1 & 1 \end{pmatrix} \begin{pmatrix} 1 & -1&0  & 0 \\ 0 & 0 & 1 & -2 \end{pmatrix} = \begin{pmatrix} 1 & -1 & 0 & 0 \\ 1 & -1 & 1 & -2 \end{pmatrix}. \]
}
Deriving $X_{M2}$, $X_{M3}$, and $X_{M4}$ analogously yields the full design matrix:
{
	\renewcommand{\arraystretch}{0.8}
	\[ X_M = \begin{pmatrix} 1 & -1 & 0 & 0 \\ 1 & -1 & 1 & -2 \\ 1 & -1 & 0 & 0 \\ 1 & -1 & 1 & 0 \\ 1 & 0 & 0 & 0 \\ 1 & 0 & 1 & 1 \\ 1 & 2 & 0 & 0 \\ 1 & 2 & 1 & 1 \end{pmatrix}. \]
}
When $\Sigma = \sigma^2 I_8$, the covariance terms cancel out, and the finite-sample inferential target is exactly computed as:
\[ \beta_{M1}^D = (X_M^\top X_M)^{-1} X_M^\top \mu_1 = (1, 1, 2, 1)^\top. \]
Because the selected working model correctly and fully spans the true mean function, the projection parameter $\beta_M^D$ exactly recovers the true coefficients: a global intercept of 1, a global slope of 2, an $x_1$ effect of 1, and a time-varying $x_2$ effect of 1. 
\end{example}

\begin{example}[Incorrect Variable Selection]
Now we consider a realistic scenario where the selection algorithm can be erroneous. Suppose it incorrectly selects $x_2$ for the intercept and $x_1$ for the slope. The new selected model $M'$ has  $M_I' = \{2\}$ and $M_S' = \{1\}$. 
%For the first subject, the design matrix is:
%\[ X_{M'1} = \begin{pmatrix} 1 & 0 \\ 1 & 1 \end{pmatrix} \begin{pmatrix} 1 & -2 & 0 & 0 \\ 0 & 0 & 1 & -1 \end{pmatrix} = \begin{pmatrix} 1 & -2 & 0 & 0 \\ 1 & -2 & 1 & -1 \end{pmatrix}. \]
%The resulting stacked design matrix is: $X_{M'}=(X_{M'1}^\top,X_{M'2}^\top,X_{M'3}^\top,X_{M'4}^\top)^\top$
%\[ X_{M'} = \begin{pmatrix} 1 & -2 & 0 & 0 \\ 1 & -2 & 1 & -1 \\ 1 & 0 & 0 & 0 \\ 1 & 0 & 1 & -1 \\ 1 & 1 & 0 & 0 \\ 1 & 1 & 1 & 0 \\ 1 & 1 & 0 & 0 \\ 1 & 1 & 1 & 2 \end{pmatrix}. \]
The corresponding design-conditional projection is therefore:
\[ \beta_{M'1}^D = (X_{M'}^\top X_{M'})^{-1} X_{M'}^\top \mu = \begin{pmatrix} 1  & 11/14 & 2 & 8/7 \end{pmatrix}^\top \approx \begin{pmatrix} 1.00  & 0.79 & 2.00 & 1.14 \end{pmatrix}^\top. \]
These coefficients should not be interpreted as biased estimates of the original true parameters. Rather, they answer a different question: given the incorrectly selected working model, what linear combination of the available columns best approximates the true mean trajectory? 
For example, the true intercept function depends on $x_1$, but the selected intercept model contains only $x_2$. Because $x_1$ and $x_2$ are not orthogonal in the realized four-subject design ($\sum_{i=1}^4 x_{i1} x_{i2} \neq 0$), the selected $x_2$ column can partially approximate the non-selected $x_1$ contribution. 

This highlights an essential feature of the proposed framework. Valid post-selection inference does not require the selection procedure to identify the correct model. Conditional on the selected model, the inferential target is simply redefined as the best linear projection onto the realized subspace. Even when the selection algorithm can be erroneous, the data fission procedure guarantees that the constructed $95\%$ confidence intervals will maintain exact nominal coverage for these specific projection coefficients $(1, 11/14,2, 8/7)^\top$.
\end{example}

\begin{example}[Nonlinear Mean Function]
The projection framework is particularly indispensable when the true mean trajectory is nonlinear and no exact linear coefficient vector exists. Using the setting as in example \ref{examp:true} with four subjects and true model $M_I = \{1\}$ and $M_S = \{2\}$.
However, suppose the true generative functions are nonlinear: $m_1(x_1, x_2) = x_1^2$ and $m_2(x_1, x_2) = x_2$. The true mean trajectory is therefore $Y_i(t) = x_{i1}^2 + t x_{i2} + \epsilon_i(t)$. Evaluating this yields the corresponding stacked conditional mean:
\[ \mu_2 = (1, -1, 1, 1, 0, 1, 4, 5)^\top. \]
Although the selected covariates correspond to those that truly enter the two trajectory components, the intercept function $x_1^2$ is nonlinear and cannot be represented exactly by the linear working term $x_1$. We obtain the new target:
\[ \beta^D_{M2} = (X_M^\top X_M)^{-1} X_M^\top \mu_2 = \begin{pmatrix} 3/2 & 8/7 & 0  & 4/7 \end{pmatrix}^\top. \]
This example illustrates two profound points. First, the projection remains perfectly well-defined even though the true intercept function is nonlinear. The coefficient $8/7$ is simply the linear coefficient that, together with the remaining working-model terms, best approximates the structural impact of $x_1^2$ over the four observed covariate values. Second, it demonstrates that in the finite-sample setting, structural effects can leak across terms due to non-orthogonality. Even though the true slope function is exactly $m_2(x) = x_2$, the projected slope coefficient for $x_2$ is $4/7$ rather than $1$, as the linear working model jointly optimizes all parameters to best fit the entire misspecified trajectory.
\end{example}

\subsubsection{Population-Level Projection Target}
For the population-level analysis, we regard the baseline covariates as random and use \(X_i\) to denote the random covariate vector whose observed realization is \(x_i\).
To facilitate the interpretation of the population projection, we assume that
the observation process is non-informative with respect to the baseline
covariates. 

%Recall the semiparametric growth-curve model
%\[
%Y_i
%=
%Z_i m(X_i) + u_i,
%\qquad
%m(X_i)
%=
%\begin{pmatrix}
%m_1(X_i)\\
%m_2(X_i)
%\end{pmatrix},
%\]
%where $u_i$ is random with
%\[
%\mathbb{E}(u_i\mid X_i,Z_i)=0,
%\qquad
%\operatorname{Var}(u_i\mid X_i,Z_i)=\Sigma_i.
%\]
For a fixed candidate model
$M\subseteq\{1,\ldots,2p\}$, with corresponding index sets
$M_I$ and $M_S$,
%define
%\begin{equation}
%\label{eq:C_M}
%C_M(X_i)
%=
%\begin{pmatrix}
%1 & X_{i,M_I}^{\top}  & 0& 0^\top\\
%0 & 0^\top & 1 & X_{i,M_S}^{\top}
%\end{pmatrix},
%\end{equation}
%and 
let
\(
\beta_M^P
=
\begin{pmatrix}
\alpha_I^P&
\beta_I^P&
\alpha_S^P&
\beta_S^P
\end{pmatrix}^\top,
\)
where
\(
\beta_I^P\in\mathbb{R}^{|M_I|},
\) and \(
\beta_S^P\in\mathbb{R}^{|M_S|}.
\)
%The corresponding working approximation to the two trajectory %components is
%is
%\[
%C_M(X_i)\beta_M^P
%=
%\begin{pmatrix}
%\alpha_I^P + X_{i,M_I}^{\top}\beta_I^P\\
%\alpha_S^P + X_{i,M_S}^{\top}\beta_S^P
%\end{pmatrix}.
%\]
%When model is selected, w
We define the population projection target as
\begin{equation}
\label{eq:beta_pop}
\beta_M^P
=
\arg\min_{\beta}
\mathbb{E}
\left[
\{Y_i-Z_iC_M(X_i)\beta\}^{\top}
\Sigma_i^{-1}
\{Y_i-Z_iC_M(X_i)\beta\}
\right],
\end{equation}
where $C_M$ is defined in \eqref{eq:C_M}, and $Z_i$ is defined in \eqref{eq:Z_i}.

\begin{theorem}
\label{thm:beta_est}Let \(X\) denote a generic random covariate vector having the same distribution as \(X_i\) with 
\(
\mathbb{E}(X_i)=0.
\)
Consider a fixed candidate model
\(M\subseteq\{1,\ldots,2p\}\), with corresponding
intercept and slope index sets \(M_I\) and \(M_S\).
Let
\(
\Gamma_{II}
=
\operatorname{Var}(X_{M_I})\), 
\(
\Gamma_{SS}
=
\operatorname{Var}(X_{M_S}),\)
\(\Gamma_{IS}
=
\operatorname{Cov}(X_{M_I},X_{M_S}),\) and
\(
\Gamma_{SI}
=
\Gamma_{IS}^{\top}.
\)
Let
\(
A_i
=
Z_i^\top \Sigma_i^{-1} Z_i,
\)
and suppose that
{
	\renewcommand{\arraystretch}{0.8}
\[
\mathbb{E}(A_i\mid X_i)
=
A_0
=
\begin{pmatrix}
a & b\\
b & c
\end{pmatrix}
\]
}
is positive definite.
The $\beta_M^P$ defined in (\ref{eq:beta_pop}) can be calculated as
\[
\alpha_I^P
=
\mathbb{E}\{m_1(X)\},
\qquad
\alpha_S^P
=
\mathbb{E}\{m_2(X)\},
\]
and
{
	\renewcommand{\arraystretch}{0.8}
\[
\begin{pmatrix}
\beta_I^P\\
\beta_S^P
\end{pmatrix}
=
\begin{pmatrix}
a\Gamma_{II} & b\Gamma_{IS}\\
b\Gamma_{SI} & c\Gamma_{SS}
\end{pmatrix}^{-1}
\begin{pmatrix}
a\,\operatorname{Cov}\{X_{M_I},m_1(X)\}
+
b\,\operatorname{Cov}\{X_{M_I},m_2(X)\}
\\[2mm]
b\,\operatorname{Cov}\{X_{M_S},m_1(X)\}
+
c\,\operatorname{Cov}\{X_{M_S},m_2(X)\}
\end{pmatrix}.
\]
}
Furthermore, if we select the same covariates $J\subset\{1,\ldots,p\}$ in both intercept term and slope term,
\(
M_I=M_S=J,
\)
then the joint longitudinal projection separates into the ordinary population
linear projections
\(
\beta_I^P
=
\Gamma_{JJ}^{-1}
\operatorname{Cov}\{X_J,m_1(X)\},
\text{ and }
\beta_S^P
=
\Gamma_{JJ}^{-1}
\operatorname{Cov}\{X_J,m_2(X)\}.
\)
\end{theorem}

Theorem \ref{thm:beta_est} highlights an important distinction between \(M_I=M_S\) and \(M_I\neq M_S\). When \(M_I=M_S=J\), the joint longitudinal projection separates into two ordinary population linear projections. %: \(\beta_I^P\) depends only on \(\operatorname{Cov}\{X_J,m_1(X)\}\), while \(\beta_S^P\) depends only on \(\operatorname{Cov}\{X_J,m_2(X)\}\). 
In contrast, when \(M_I\neq M_S\), this separation generally no longer holds, and each coefficient may depend on covariance terms involving both \(m_1(X)\) and \(m_2(X)\).

	\begin{remark}[Average derivative interpretation]
		Suppose that the same covariate set is used for the intercept and slope
		components, $M_I=M_S=J$. Then, for $k=1,2$,
		\[
		\beta_{k,J}^P
		=
		\Gamma_{JJ}^{-1}
		\operatorname{Cov}\{X_J,m_k(X)\}.
		\]
		Define the reduced conditional mean function
		\(
		m_{k,J}(x_J)
		=
		\mathbb{E}\{m_k(X)\mid X_J=x_J\}.
		\)
		By the law of iterated expectations,
		\(
		\operatorname{Cov}\{X_J,m_k(X)\}
		=
		\operatorname{Cov}\{X_J,m_{k,J}(X_J)\}.
		\)
		If $X_J$ follows a multivariate Gaussian distribution and $m_{k,J}$ is
		sufficiently smooth, Stein's identity \citep{stein1981estimation} gives
		\[
		\beta_{k,J}^P
		=
		\mathbb{E}\{\nabla m_{k,J}(X_J)\}.
		\]
		Thus, $\beta_{k,J}^P$ can be interpreted as the average local change in the
		$k$th growth parameter with respect to the selected covariates, after
		averaging over the non-selected covariates conditional on $X_J$
		\citep{chernozhukov2021automatic}.
		
%		When all covariates are selected, $m_{k,J}=m_k$, and the result reduces to
%		\(
%		\beta_k^P
%		=
%		\mathbb{E}\{\nabla m_k(X)\}.
%		\)
%		If the selected and non-selected covariates are independent, the same
%		average-derivative interpretation applies to the selected coordinates after
%		averaging over the non-selected covariates. When they are correlated,
%		however, changes in $X_J$ also change the conditional distribution of the
%		non-selected covariates, so the projection coefficient reflects both direct
%		and dependence-induced changes.
	\end{remark}

\begin{example}[Correct Variable Selection]
Suppose
\(
m_1(x_1,x_2)=1+x_1,
\) and \(
m_2(x_1,x_2)=2+x_2.
\)
Assume that $(X_1,X_2)^\top$ follows a standard bivariate normal distribution
with correlation $\rho$.
%\[
%\begin{pmatrix}
%X_1\\
%X_2
%\end{pmatrix}
%\sim
%N\left(
%\begin{pmatrix}
%0\\
%0
%\end{pmatrix},
%\begin{pmatrix}
%1 & \rho\\
%\rho & 1
%\end{pmatrix}
%\right).
%\]
Suppose that the selected model correctly assigns
\(
M_I=\{1\},
\) and \(M_S=\{2\}.\)
By Theorem \ref{thm:beta_est},
\(
\alpha_{I1}^P=1,\) and \(\alpha_{S1}^P=2.\) Moreover,
{
	\renewcommand{\arraystretch}{0.8}
\[
\begin{pmatrix}
\beta_{I1}^P\\
\beta_{S1}^P
\end{pmatrix}
=
\begin{pmatrix}
a & b\rho\\
b\rho & c
\end{pmatrix}^{-1}
\begin{pmatrix}
a+b\rho\\
b\rho+c
\end{pmatrix}=\begin{pmatrix}
1\\
1
\end{pmatrix}.
\]
}
Therefore,
\(
\beta_{M1}^P
=
(1,1,2,1)^\top.
\)
In this example, the selected working model contains the true conditional mean
exactly. Hence the population projection recovers the structural coefficients,
regardless of the correlation $\rho$ and particular values of $a$, $b$, and $c$.
\end{example}

\begin{example}[Incorrect Variable Selection]
Consider the same data-generating mechanism, but suppose that the selected
variables are assigned to the wrong trajectory components:
\(
M_I'=\{2\},
\)
and
\(
M_S'=\{1\}.
\)
Because the covariates are centered,
\(
\alpha_{I'1}^P=1,
\) and \(
\alpha_{S'1}^P=2.
\)
The remaining coefficients 
%satisfy
%\[
%\begin{pmatrix}
%\beta_{I'1}^P\\
%\beta_{S'1}^P
%\end{pmatrix}
%=
%\begin{pmatrix}
%a & b\rho\\
%b\rho & c
%\end{pmatrix}^{-1}
%\begin{pmatrix}
%a\rho+b\\
%b+c\rho
%\end{pmatrix},
%\]
%which 
are
\[
\beta_{I'1}^P
=
\frac{
\rho(ac-b^2)+bc(1-\rho^2)
}{
ac-b^2\rho^2
},\quad\text{and}\quad
\beta_{S'1}^P
=
\frac{
\rho(ac-b^2)+ab(1-\rho^2)
}{
ac-b^2\rho^2
}.
\]
If \(b=0,\)
we have
\(
\beta_{I'1}^P=\beta_{S'1}^P=\rho.
\)
In this case, each coefficient reduces to the ordinary population projection
induced by the correlation between $X_1$ and $X_2$.
In contrast, if \( \rho=0, \)
\(
\beta_{I'1}^P=\frac{b}{a},\) and \(
\beta_{S'1}^P=\frac{b}{c}.
\)
Thus, even when $X_1$ and $X_2$ are independent, the projection coefficients
are not zero when $b\neq0$. This is because in this example, $M_I'\ne M_{S}'$. $\beta_{I'1}^P$, the coefficient of the intercept, depends on the covariance terms involving $m_2(X)$. $\beta_{S'1}^P$, the coefficient of slope, depends on the covariance terms involving $m_1(X)$. 
This example illustrates that misspecification in a longitudinal working model
can propagate through two distinct mechanisms: correlation among baseline
covariates and coupling between the intercept and slope components.
\end{example}

\begin{example}[Nonlinear Mean Function]

Finally, suppose
\(
m_1(x_1,x_2)=x_1^2,
\) and \(
m_2(x_1,x_2)=x_2,
\)
where $(X_1,X_2)^\top$ follows a standard bivariate normal distribution
with correlation $\rho$.
Suppose
\(
M_I=\{1\},
\)
and
\(
M_S=\{2\}.
\)
The global trajectory components are
\(
\alpha_{I2}^P
=
\mathbb{E}(X_1^2)
=
1,\)
and
\(
\alpha_{S2}^P
=
\mathbb{E}(X_2)
=
0.
\)
As we have
\(
\operatorname{Cov}(X_1,X_1^2)
=
\mathbb{E}(X_1^3)
=
0,
\)
and
\(
\operatorname{Cov}(X_2,X_1^2)
=
\mathbb{E}(X_2X_1^2)
=
0.
\)
%Therefore,
%\begin{align*}
%\begin{pmatrix}
%\beta_{I2}^P\\
%\beta_{S2}^P
%\end{pmatrix}
%&=
%\begin{pmatrix}
%a & b\rho\\
%b\rho & c
%\end{pmatrix}^{-1}
%\begin{pmatrix}
%b\rho\\
%c
%\end{pmatrix}\
%=
%\begin{pmatrix}
%0\\
%1
%\end{pmatrix}.
%\end{align*}
Therefore, Theorem~\ref{thm:beta_est} gives
$\beta_I^P=0$ and $\beta_S^P=1$.
This example emphasizes that variable selection and projection inference answer
different questions. A variable may be important to a nonlinear trajectory and
therefore be correctly selected, while its subsequent linear projection
coefficient is zero.
\end{example}

\section{Selective Inference via Data Fission with A Known Covariance Matrix}
\label{sec:known_cov}
\subsection{The Challenge of Post-Selection Inference}
%In applied research, a commonly adopted workflow is to apply a variable selection method (e.g., the LASSO) to the full dataset $Y$ to identify an active set $M$, and then, without further adjustment, estimate parameters and construct confidence intervals as if the selected model had been specified in advance. This %\text{red}{french word add dotdot}
%\naive approach ignores a fundamental statistical issue: both the selection and inference stages are conducted on the same data and are therefore statistically dependent \citep{berk2013valid}.

Classical inference typically treats the working model as fixed conditional on the realized design, rather than selected using the response \(Y\). For example, the covariates may be selected by domain experts based only on their substantive knowledge, without using the observed response \(Y\). However, when the selected model $M=\mathcal{S}(Y)$ is determined by the
response, conditioning on the selection event
\(
Y\in\mathcal{A}_M
\)
generally changes the sampling distribution of the subsequent estimator.
Therefore, the effect of the data-driven selection event
$\mathcal{A}_M$ cannot be ignored in subsequent inference.

% Because the selected model $M = \mathcal{S}(Y)$ is a random outcome driven by the response, the chosen design matrix $X_M$ and the response $Y$ are intrinsically correlated. \text{red}{not clear} Classical inference relies on the unconditioned response space, which yields a symmetric Gaussian sampling distribution for the generalized least squares (GLS) estimator. \text{red}{do you mean classic inference incorporating seletion or not? where does the generalized least square come from?} However, conditioning on the data-driven selection event $\mathcal{A}_M$ severely restricts the response space. 
For instance, in LASSO-based selection, the event $\{Y \in \mathcal{A}_M\}$ analytically corresponds to $Y$ falling into a specific polyhedral region \citep{lee2016exact}. Consequently, the conditional sampling distribution of the \naive estimator is no longer normal. %For example, under Gaussian errors, conditioning on a LASSO selection event induces a truncated Gaussian distribution for relevant linear contrasts.

Ignoring this truncation and dependence leads to severe inferential bias. As a result, effect sizes are typically overestimated, and nominal confidence intervals tend to be invalid, leading to misleading scientific conclusions \citep{berk2013valid}. This breakdown of classical normality directly highlights why decoupling the selection and inference stages is strictly necessary. %To restore valid inference, one must evaluate the selected model using a \text{red}{why do you need a continuous response vector? can it be discrete?} continuous response vector that is statistically independent of the selection process itself.
To restore valid inference, one can evaluate the selected model using a response vector that is statistically independent of the selection process itself.

\subsection{The Data Fission Construction}

To circumvent the intractable conditional distribution discussed above, one must evaluate the selected model using a response vector that is statistically independent of the selection process. The traditional remedy is data splitting \citep{cox1975note}, where the subjects are randomly partitioned into a selection set and an inference set. While statistically valid, this approach drastically reduces the effective sample size, reducing valuable statistical power that is critical in high-dimensional longitudinal studies.

Data fission \citep{leiner2025data} offers an elegant alternative: rather than splitting the subjects (the rows of the data), we continuously split the information within the response vector itself. % From an information-theoretic perspective, data fission overcomes the fundamental inefficiency of sample splitting by achieving an exact decomposition of Fisher information without partitioning the \(N\) participants into separate selection and inference samples. 
In this section, we assume that the true structural covariance parameter $\theta_0$ is known, meaning the marginal covariance matrix $\Sigma = \Sigma(\theta_0)$ is fully specified. This serves as the theoretical foundation for the empirical procedure with estimated covariance developed in Section 4.

To create independent data splits, we introduce an auxiliary random noise vector $W$, generated independently of the observed data, such that $W \sim N(0,\Sigma)$ \citep{leiner2025data}. Crucially, the external noise is engineered to mirror the exact within-subject correlation structure of the original data. Given a fixed scalar $\tau > 0$ that governs the signal-to-noise allocation for the selection step and the inference step, we construct two perturbed response vectors:\[Y^{sel} = Y + \tau W, \qquad Y^{inf} = Y - \tau^{-1}W.\]

This additive perturbation disentangles the selection and inference stages while preserving the full cohort structure. Because $Y$ and $W$ are independent multivariate normal vectors, any linear combination of them remains jointly Gaussian. A direct calculation of their cross-covariance yields:
%\begin{align*}
%	\operatorname{Cov}(Y^{sel}, Y^{inf}) &= \operatorname{Cov}(Y + \tau W, Y - \tau^{-1}W) \nonumber \\&= \operatorname{Cov}(Y, Y) - \tau^{-1}\operatorname{Cov}(Y, W) + \tau \operatorname{Cov}(W, Y) - \tau \tau^{-1} \operatorname{Cov}(W, W) \nonumber \\&= \Sigma - 0 + 0 - \Sigma = 0.
%\end{align*}
\[
\operatorname{Cov}(Y^{sel}, Y^{inf}) = \operatorname{Cov}(Y + \tau W, Y - \tau^{-1}W)=0
\]
In the multivariate Gaussian family, zero covariance directly implies exact statistical independence. Therefore, $Y^{sel}$ and $Y^{inf}$ provide two independent randomized responses of the longitudinal outcome. The tuning parameter $\tau$ serves as an information-theoretic trade-off mechanism between the two stages: a smaller $\tau$ injects less noise into $Y^{sel}$, thereby allocating more Fisher information to the variable selection procedure $\mathcal{S}(Y^{sel})$, but consequently inflates the variance of $Y^{inf}$, yielding wider confidence intervals during inference. Conversely, a larger $\tau$ sacrifices selection power to achieve tighter post-selection inferential precision. Regardless of the choice of $\tau$, conditional on the realized design, the two randomized responses are %exactly 
independent.

\subsection{Valid Post-Selection Inference}
We utilize the first perturbed response, $Y^{sel}$, exclusively to execute the variable selection procedure. The selected model is thus strictly a function of $Y^{sel}$, given by $M = \mathcal{S}(Y^{sel})$.
For a fixed candidate model $M$, we condition on the event that $M$ is selected, and construct the target design matrix $X_M$ as defined in Section \ref{sec:target}. All subsequent inference regarding the design-conditional projection target \(\beta_M^D\) is conducted using the independent inference response \(Y^{inf}\). We define the corresponding post-selection estimator as \(\beta_M(Y^{inf})\), where
\begin{equation}
	\label{eq:beta_hat}
	{\beta}_M(Y^{inf}) = (X_M^\top \Sigma^{-1} X_M)^{-1} X_M^\top \Sigma^{-1} Y^{inf}.
\end{equation}
Because the estimator ${\beta}_M(Y^{inf})$ depends on the original data only through $Y^{inf}$, and the selection event depends on the data only through $Y^{sel}$, the estimator and the selection event are independent. Consequently, conditioning on the selection event $\mathcal{A}_M$ does not alter the sampling distribution of the estimator:
\begin{equation*}
	{\beta}_M(Y^{inf}) \mid \{Y^{sel} \in \mathcal{A}_M\} \stackrel{d}{=} {\beta}_M(Y^{inf}).
\end{equation*}
Since $Y^{inf} \sim N(\mu, (1+\tau^{-2})\Sigma)$, the exact distribution of our estimator is straightforward to establish. For any dynamically selected model $M$, the conditional distribution of the estimator centers exactly on the design-conditional projection target \(\beta_M^D\):
\begin{equation}
    \label{eq:inf_known}
	{\beta}_M(Y^{inf}) \mid \{Y^{sel} \in \mathcal{A}_M\} \sim N\left(\beta_M^D, (1+\tau^{-2})(X_M^\top \Sigma^{-1} X_M)^{-1}\right).
\end{equation}
This tractability enables the immediate construction of %exact 
hypothesis tests and confidence intervals. For the $j$-th coefficient in the selected working model $M$, let $\nu_j^\top$ be the $j$-th row of $(X_M^\top \Sigma^{-1} X_M)^{-1} X_M^\top \Sigma^{-1}$. The post-selection $(1-\alpha)$ confidence interval for the design-conditional projection coefficient $\beta_{M, j}^D$ is naturally given by:
\begin{equation*}
	{\beta}_{M, j}(Y^{inf}) \pm z_{1-\alpha/2} \sqrt{(1+\tau^{-2}) \nu_j^\top \Sigma \nu_j},
\end{equation*}
where $z_{1-\alpha/2}$ is the $(1-\alpha/2)$ quantile of the standard normal distribution. This guarantees exact nominal coverage conditionally on the chosen model, effectively eliminating the selection bias.
\section{Selective Inference with Estimated Covariance}

In practical applications, the true variance components $\theta_0$ are rarely known and must be estimated from the observed data. In this section, we extend the selective inference framework developed in Section 3 to accommodate an unknown covariance structure and establish the asymptotic validity of replacing the true covariance matrix
\(
\Sigma_0=\Sigma(\theta_0)
\)
with its empirical estimate
\(
\hat\Sigma=\Sigma(\hat\theta).
\)
For notational simplicity, we set $\tau=1$ throughout this section and the corresponding proofs. The same argument can be extended to any fixed $\tau>0$.

\subsection{Empirical Data Fission Procedure}

Let $\hat\theta$ be a consistent estimator of the variance components $\theta_0$, yielding the estimated marginal covariance matrix $\hat\Sigma$. Because $\hat\Sigma$ is estimated from the observed response $Y$, we write $\hat\Sigma(Y)$ when it is useful to emphasize this dependence.
Mirroring the oracle procedure, conditional on $Y$, we generate an auxiliary
Gaussian noise vector
\(
\widehat W \mid Y
\sim
N\{0,\widehat\Sigma(Y)\}.
\)
The empirical perturbed responses are then constructed as
\begin{equation}
\label{eq:sel_inf}
\widehat Y^{\mathrm{sel}}
=
Y+\widehat W,
\qquad
\widehat Y^{\mathrm{inf}}
=
Y-\widehat W.
\end{equation}
%We use $\widehat Y^{\mathrm{sel}}$ to perform variable selection. To formulate
%the subsequent selective inference model-wise, fix a candidate model
%\(M\subseteq\{1,\ldots,2p\}\) and define its selection region as
%\[
%\mathcal A_M
%=
%\left\{
%y\in\mathbb R^n:
%\mathcal S(y)=M
%\right\}.
%\]
%We study inference conditional on the event that the empirical selection
%procedure selects this fixed model:
%\[
%\left\{
%\mathcal S(\widehat Y^{\mathrm{sel}})=M
%\right\}
%=
%\left\{
%\widehat Y^{\mathrm{sel}}\in\mathcal A_M
%\right\}.
%\]
%On this event, subsequent inference is conducted using
%$\widehat Y^{\mathrm{inf}}$. The corresponding feasible post-selection
%estimator is
%\[
%\widehat\beta_M
%=
%\left\{
%X_M^\top\widehat\Sigma(Y)^{-1}X_M
%\right\}^{-1}
%X_M^\top\widehat\Sigma(Y)^{-1}
%\widehat Y^{\mathrm{inf}}.
%\]
%When there is no ambiguity, we suppress the dependence of
%$\widehat\Sigma(Y)$ on $Y$ and write it simply as $\widehat\Sigma$.
%% Throughout this section, $M$ denotes the fixed candidate model on which we
%% condition, rather than the random output of the selection procedure. \text{red}{This is confusing. What do you try to communicate here?}
Variable selection is performed using $\widehat Y^{\mathrm{sel}}$, yielding
$M=\mathcal S(\widehat Y^{\mathrm{sel}})$, and subsequent inference is based
on $\widehat Y^{\mathrm{inf}}$. For a fixed candidate model $M$, the feasible
post-selection estimator is
\[
\widehat\beta_M
=
\left(
X_M^\top\widehat\Sigma^{-1}X_M
\right)^{-1}
X_M^\top\widehat\Sigma^{-1}
\widehat Y^{\mathrm{inf}}.
\]
When there is no ambiguity, we suppress the dependence of
$\widehat\Sigma(Y)$ on $Y$.
We allow the number of candidate covariates $p=p_N$ to increase with the
sample size, while requiring the selection procedure to select only a
bounded number of covariates. Specifically, there exists a constant
$\bar d<\infty$, independent of $N$, such that every model that can be
selected satisfies
\begin{equation}
\label{eq:d_M}
d_M=|M_I|+|M_S|+2\leq \bar d.
\end{equation}
Throughout the asymptotic analysis below, we condition on one particular
selected model $M\subseteq\{1,\ldots,2p\}$ and treat $M$, $M_I$, and $M_S$
as fixed.
% Throughout the asymptotic analysis below, we condition on a particular
% candidate model $M\in\mathcal M_N$ and treat $M$, together with its index
% sets $M_I$ and $M_S$, as fixed. %and bounded 
%as $N\to\infty$.
Specifically, we study the conditional
distribution given the selection event
\(
\widehat Y^{\mathrm{sel}}\in\mathcal A_M,
\)
that is, given that the data-driven selection procedure selects $M$.
Thus, although the selected model is random before conditioning, it is
treated as fixed in the subsequent asymptotic analysis.

Unlike the oracle construction in Section \ref{sec:known_cov}, the empirical selection and inference responses are not exactly independent in finite samples. This distinction motivates the asymptotic analysis below.

\subsection{Asymptotic Equivalence and Validity}
%A fundamental statistical challenge in post-selection inference with an estimated covariance matrix is the loss of exact independence between $\hat Y^{\mathrm{sel}}$ and $\hat Y^{\mathrm{inf}}$. In the oracle procedure, the auxiliary noise $W$ is generated using the true covariance matrix $\Sigma_0$, which guarantees that  $Y^{\mathrm{sel}}$ and $Y^{\mathrm{inf}}$ are exactly independent. In the empirical procedure, however, $\hat W$ is generated using $\hat\Sigma(Y)$, which is itself estimated from the full observed response $Y$. Consequently, $\hat Y^{\mathrm{sel}}$ and $\hat Y^{\mathrm{inf}}$ share information through the estimated covariance matrix.

To provide theoretical results, we only use part of the data to estimate the covariance. Specifically,
let $\mathcal P\subset\{1,\ldots,N\}$ be a subject-level pilot set,
selected independently of the observed outcomes, and let
\(
N_s=|\mathcal P|,
\) and \(
\mathcal P^c=\{1,\ldots,N\}\setminus\mathcal P.
\)
Only the subjects in
$\mathcal P$ are used to estimate the covariance parameters. However, all
$N$ subjects, including the pilot subjects, are retained in the subsequent
selection and inference stages.
We impose the following assumptions. 
\begin{assumption}
	The subject-level tuples \((X_i,T_i,Z_i,Y_i),i=1,\ldots, N\) are independent and identically drawn from a population distribution $\mathcal O$,
    with the number of observational time $E(T_i)<\infty.$ Consequently, $n=\sum_{i=1}^N T_i\asymp N$.
    We further assume that  
    \[
    Y_i\mid X_i,T_i,Z_i\sim N(\mu_i,\Sigma_i(\theta_0)),
   \qquad
	\Sigma_i(\theta)
	=Z_iD(\theta)Z_i^\top+\sigma^2(\theta)I_{T_i}.
	\]
   At the true parameter value $\theta_0$, the random-effects covariance $D_0:=D(\theta_0)$ is positive definite and
$\sigma_0^2:=\sigma^2(\theta_0)>0$.
\end{assumption}

\begin{assumption}
	The number of pilot participants $N_s$ that are used to estimate $\theta_0$ satisfies:
	\(
	N_s\to\infty\) and \(N_s/N\to 0.\)
	The pilot estimator $\widehat \theta$ satisfies
	\[
	\|\widehat\theta-\theta_0\|_2
	=
	O_p(N_s^{-1/2}).
	\]
\end{assumption}

\begin{assumption}
	Every candidate model under consideration has common support,
	\(
	M_I=M_S.
	\)
	Let $\Theta_0$ denote a fixed neighborhood of $\theta_0$ contained in the covariance-parameter space.
    Moreover, for every $\theta\in\Theta_0$,
	\[
	B_0(\theta):=\mathbb E\!\left\{
	Z_i^\top\Sigma_i(\theta)^{-1}Z_i
	\mid X_i
	\right\}
	\]
	is a deterministic positive-definite $2\times2$
	matrix, regardless of the value of $X_i$, $i=1,\ldots,N$.
\end{assumption}
% \begin{assumption}[Design, dimension, and moment regularity]
% 	The selected-model dimension is uniformly bounded,
% 	$d_M\le \bar d<\infty$, where $d_M=|M_I|+|M_S|+2$. There exist constants
% 	$0<c_1\le c_2<\infty$ such that
% 	\[
% 	\mathbb P\left(
% 	c_1I_{d_M}
% 	\preceq
% 	\frac{1}{N}
% 	X_M^\top\Sigma_0^{-1}X_M
% 	\preceq
% 	c_2I_{d_M}
% 	\right)
% 	\longrightarrow1.
% 	\]
%     Here $A\preceq B$ means that $(B-A)$ is semidefinite.
    
% 	Moreover, for the population projection residual
% 	\[
% 	r_{Mi}^{P}
% 	=
% 	\mu_i-X_{i,M}\beta_M^P,
% 	\]
% 	there exists a constant $C_h<\infty$ such that
% 	\[
% 	\mathbb E\left[
% 	\left\|
% 	X_{i,M}^{\top}\Sigma_{0i}^{-1/2}
% 	\right\|_s^2
% 	\left\|
% 	\Sigma_{0i}^{-1/2}r_{Mi}^{P}
% 	\right\|_2^2
% 	\right]
% 	\le C_h.
% 	\]
% \end{assumption}
\begin{assumption}
	Let $X_{i,\mathrm{full}},i=1,\ldots,N$ denote the working design matrix containing
	all candidate intercept and slope terms. There exist constants
	$0<c_1\le c_2<\infty$ such that
	\[
	c_1 I
	\preceq
	\mathbb E\left[
	X_{i,\mathrm{full}}^\top
	\Sigma_i(\theta_0)^{-1}
	X_{i,\mathrm{full}}
	\right]
	\preceq
	c_2 I.
	\]
	For each candidate model $M$ under consideration, define %the population projection residual
	\(
	r_{Mi}^{P}
	=
	\mu_i-X_{i,M}\beta_M^P.
	\)
	There exists a constant $C_h<\infty$ such that
	\[
	\sup_M
	\mathbb E\left[
	\left\|
	X_{i,M}^{\top}
	\Sigma_i(\theta_0)^{-1}
	r_{Mi}^{P}
	\right\|_2^2
	\right]
	\le C_h.\]
\end{assumption}

\begin{assumption}
	For the fixed candidate model $M$ with $M_I=M_S$, there exists a
	constant $c_4>0$ such that,
	\[
	\mathbb P\left(
	\widehat Y^{\mathrm{sel}}\in\mathcal A_M
	\mid\mathcal D_N
	\right)
	\ge c_4.
	\]
\end{assumption}

Assumption 1 is standard in longitudinal studies. Assumption 2 requires the parameters to be estimated sufficiently
accurately from the pilot sample. Assumption 3
requires common support for the intercept and slope components and a
non-informative observation process. Assumptions 4 and 5 impose standard regularity conditions required for the
asymptotic analysis. Further discussion of these assumptions is
provided in Appendix \ref{sec:assum_discussion}.

\begin{theorem}[Conditional validity of pilot-assisted empirical data fission]
	\label{thm:pilot_assisted_validity}
    Let $\hat V_M$ be the estimated covariance matrix of $\widehat \beta_M$ with 
    \begin{equation}
    \label{eq:hat_vm}
       \hat V_M=2(X_M^T\widehat \Sigma^{-1}X_M)^{-1} 
    \end{equation}
	Under Assumptions 1--5, 
	\begin{equation}
    \label{eq: inference}
	\sup_{z\in\mathbb R^{d_M}}
	\left|
	\mathbb P\left\{
	\widehat V_M^{-1/2}
	(\widehat\beta_M-\beta_{M}^{D})
	\le z
	\,\middle|\,
	\widehat Y^{\mathrm{sel}}\in\mathcal A_M,
	\mathcal D_N
	\right\}
	-
	\Phi_{d_M}(z)
	\right|
	\xrightarrow{p}0,
	\end{equation}
	where vector inequalities are interpreted componentwise and $\Phi_{d_M}$ is
	the distribution function of $N_{d_M}(0,I_{d_M})$.
\end{theorem}

Theorem \ref{thm:pilot_assisted_validity} establishes that the additional dependence introduced by estimating the covariance matrix has an asymptotically negligible effect on the conditional distribution of the post-selection estimator. Thus, the inference response $\hat Y^{\mathrm{inf}}$, together with $\hat\Sigma$, can be used to approximate the oracle inference procedure developed in Section \ref{sec:known_cov}.

Since the oracle estimator is centered at the design-conditional projection target $\beta_M^D$, Theorem \ref{thm:pilot_assisted_validity} provides the theoretical basis for constructing feasible post-selection confidence intervals using the estimated covariance matrix. The formal proof, based on controlling the total variation distance between the oracle and empirical constructions, is provided in Appendix \ref{sec:proof}.

\section{Working Variable-Selection Procedures}
\label{sec:working_methods}

The proposed data fission framework places minimal restrictions on the variable-selection procedure. %In particular, the selection algorithm \(\mathcal{S}(\cdot)\) may be viewed as a deterministic black box that maps the selection response, together with the observed design variables, to a subset of baseline covariates. This flexibility is useful in longitudinal applications, where important predictors may affect developmental trajectories through nonlinearities, interactions, or subgroup-specific patterns that are difficult to specify in advance.
To illustrate this generality, we consider two qualitatively different working selectors based on established methods: a mixed-effects model tree and a LassoNet-style neural network. % We do not regard either procedure as universally optimal. Rather, they are chosen to represent two common forms of heterogeneity in longitudinal behavioral data. 
The tree-based procedure is particularly suited to threshold-defined and subgroup-specific trajectory differences, whereas the neural-network procedure is intended to capture smoother nonlinear effects and interactions among multiple baseline predictors.
%These procedures are used only to identify the baseline variables that enter the subsequent inferential working model. The internal parameters of the tree or neural network are not themselves targets of inference. %Consequently, the validity of the oracle data fission procedure does not require either working selector to correctly specify the true longitudinal mean function. Instead, 
Once a subset of variables has been selected, inference is conducted for the corresponding linear projection parameters defined in Section~\ref{sec:target}.

\subsection{Mixed-Effects Model Tree}
\label{sec:tree_selector}

%Tree-based methods provide a flexible approach for identifying heterogeneous subgroups by recursively partitioning the predictor space \citep{breiman1984classification}. Unlike conventional regression models, which require the functional form of predictor effects to be specified in advance, regression trees can adaptively identify threshold effects and interactions among predictors.
%
%Model-based recursive partitioning extends this idea by fitting a parametric model within each subgroup and using parameter instability to determine whether further partitioning is needed \citep{zeileis2008model}. Starting from the full sample, the procedure fits the specified model and evaluates whether its parameters remain stable across candidate partitioning variables. If evidence of parameter instability is detected, the sample is split according to the variable showing the strongest instability, and the procedure is recursively repeated within the resulting subgroups.

%For longitudinal and clustered data, 
We use the mixed-effects model tree of \citet{fokkema2018detecting}, %which incorporates random effects into model-based recursive partitioning.
 which extends model-based recursive partitioning \citep{zeileis2008model} to
clustered and longitudinal data.
In our setting, the longitudinal
selection response $\widehat Y^{\mathrm{sel}}$ is the outcome, while baseline
covariates are used as candidate partitioning variables. Within each node,
the model contains a fixed intercept and linear time slope together with
subject-specific random intercepts and slopes.
%so that splits capture
% heterogeneity in the intercept, slope, or both. 
Estimation proceeds
iteratively by alternating between tree fitting and updating the mixed-effects
 model until convergence.
% the working model within each node contains a fixed intercept and linear time slope, together with subject-specific random intercepts and slopes. Baseline covariates are used as candidate partitioning variables, and a split may reflect heterogeneity in the trajectory intercept, slope, or both.

We use the fitted tree as a variable-screening device and define the selected covariate set as
\[
\widehat J_{\mathrm{tree}}
=
\left\{
j:
X_j \text{ appears in at least one internal split of the fitted tree}
\right\}.
\]
A variable appearing in a split is therefore identified as useful for describing heterogeneity in the longitudinal trajectory, but the tree does not uniquely assign that variable to either the intercept or slope component. We therefore use a common variable-level support in the subsequent inferential model:
\(
M_I=M_S=\widehat J_{\mathrm{tree}}.
\)
Thus, once $X_j$ is selected by the tree, both $X_j$ and $tX_j$ are included in the linear growth-curve working model. The corresponding intercept and slope projection coefficients are estimated separately at the inference stage, and neither coefficient is assumed to be nonzero.

\subsection{LassoNet-Style Neural-Network Selection}
\label{sec:lassonet_selector}

%Neural networks provide a flexible approach for modeling nonlinear relationships and interactions among predictors \citep{goodfellow2016deep}. A standard feed-forward neural network repeatedly applies linear transformations and nonlinear activation functions. This flexibility is useful in longitudinal settings where baseline characteristics may influence developmental trajectories through nonlinear or interacting mechanisms.
%
%However, conventional feed-forward neural networks are not naturally suited to variable selection. Simply adding an $\ell_1$ penalty to individual network weights does not generally produce a sparse set of input variables, because a single predictor may influence the fitted outcome through multiple network paths. Thus, sparsity in network parameters does not directly imply sparsity in the original predictors.
%
%In contrast, the Lasso performs variable selection directly by placing an $\ell_1$ penalty on a single coefficient associated with each predictor \citep{tibshirani1996regression}. If the coefficient of a predictor is shrunk to zero, that predictor is removed from the fitted linear model. This provides a clear variable-level notion of sparsity, but the resulting model is limited to linear effects unless additional nonlinear terms and interactions are specified in advance.

	Neural networks provide a flexible approach for modeling nonlinear
	relationships and interactions among predictors \citep{goodfellow2016deep}.
	However, conventional feed-forward networks are not naturally suited to
	variable selection: sparsity in individual network weights does not generally
	translate into sparsity at the input-variable level because a predictor may
	contribute through multiple network paths. %Simply adding an $\ell_1$ penalty to individual network weights does not generally produce a sparse set of input variables, because a single predictor may influence the fitted outcome through multiple network paths. %Thus, sparsity in network parameters does not directly imply sparsity in the original predictors.
	
	In contrast, the Lasso provides direct variable-level sparsity by penalizing
	one coefficient per predictor \citep{tibshirani1996regression}, but is limited
	to linear effects unless nonlinear terms and interactions are specified in
	advance. LassoNet \citep{lemhadri2021lassonet} combines these two features by
	coupling a sparse linear component with a nonlinear neural network,
%LassoNet \citep{lemhadri2021lassonet} combines this idea with a nonlinear neural network by adding a sparse linear {\{blue} component},
\[
g(X)=\gamma^\top X+h(X),
\]
where $\gamma$ contains the linear coefficients and $h(X)$ denotes the nonlinear component. LassoNet imposes the hierarchical constraint
\[
\|\Gamma_j^{(1)}\|_\infty
\le C_{\mathrm{Net}}|\gamma_j|,
\]
where $\Gamma_j^{(1)}$ denotes the first-layer weights associated with
predictor $X_j$. Hence, if $\gamma_j=0$, then
$\Gamma_j^{(1)}=0$, excluding $X_j$ from both the linear and nonlinear
components.
%where $\gamma$ contains the linear coefficients and $h(X)$ denotes the nonlinear component. Let $\Gamma_j^{(1)}$ denote the vector of first-layer weights associated with predictor $X_j$. LassoNet imposes the hierarchical constraint
%\(
%\|\Gamma_j^{(1)}\|_\infty
%\le C_{Net}|\gamma_j|,
%\)
%where $ C_{Net}>0$ controls the strength of the hierarchy. Consequently, if $\gamma_j=0$, then $W_j^{(1)}=0$, so predictor $X_j$ cannot contribute through either the linear or nonlinear part of the model. This provides an explicit variable-selection rule while retaining the ability to capture nonlinear effects and interactions.

We adapt this construction to the growth-curve setting by introducing
separate LassoNet branches for the intercept and slope functions:
\[
g_I(X)=\gamma_I^\top X+h_I(X),\qquad
g_S(X)=\gamma_S^\top X+h_S(X),
\]
where $\gamma_I$ and $\gamma_S$ are the corresponding linear coefficients,
and $h_I(\cdot)$ and $h_S(\cdot)$ denote the nonlinear components.
The network takes the baseline covariates $X_i$ as input and is fitted to the
repeated selection response $\widehat Y_i^{\mathrm{sel}}$. Within-participant
dependence is incorporated through the working covariance matrix
\[
\Sigma_i^{\mathrm{sel}}(\theta)
=
(1+\tau^2)
\left\{
Z_iD(\theta)Z_i^\top+\sigma^2(\theta)I_{T_i}
\right\},
\]
with the covariance parameter $\theta$ estimated jointly during network
fitting. We set $M_I=M_S$, and covariate $j$ is selected if at least one of
$\gamma_{I,j}$ and $\gamma_{S,j}$ is nonzero.

\section{Simulation Study}
\label{sec:simulation}

\subsection{Objectives and Overview}
\label{sec:obj_over}

We conducted simulation studies to evaluate the finite-sample performance of
the proposed data fission framework after complex, data-adaptive variable
selection. Our primary objective was not to identify an optimal
variable-selection algorithm. Rather, we examined whether valid inference for
the selected projection targets could be recovered after applying selection
procedures whose selection events are difficult to characterize analytically.

We considered the two working selectors described in Section~5: a
mixed-effects model tree and a LassoNet-style neural network. For each
selector, we evaluated post-selection inference under both known and unknown
covariance structures.

	When the covariance matrix $\Sigma_0$ was treated as known, we compared two
	oracle procedures. \textit{Oracle data fission} generated
	$W\mid X,t\sim N(0,\Sigma_0)$ and constructed
	$Y^{\mathrm{sel}}=Y+\tau W$ for variable selection and
	$Y^{\mathrm{inf}}=Y-\tau^{-1}W$ for GLS inference based on true covariance. \textit{Oracle data splitting} instead randomly divided
	subjects equally into selection and inference samples, with selection
	performed on the former and GLS inference using the true covariance on the
	latter.
	
	When the covariance structure was unknown, we compared four procedures.
	\textit{Empirical data fission with plug-in GLS} used
	$\widehat{\Sigma}$ to construct $\widehat Y^{\mathrm{sel}}$ and
	$\widehat Y^{\mathrm{inf}}$ as in (\ref{eq:sel_inf}), with selection based on
	$\widehat Y^{\mathrm{sel}}$ and plug-in GLS inference based on
	$\widehat Y^{\mathrm{inf}}$. \textit{Empirical data fission with REML} used
	the same selection response but refitted the selected working model by REML
	using $\widehat Y^{\mathrm{inf}}$. \textit{Data splitting with REML} divided
	subjects equally into selection and inference samples and refitted the
	selected working model by REML using only the inference sample.
	\textit{\Naive REML inference} used the same observed outcome for both
	selection and subsequent REML inference, without adjustment for selection.

In the simulation study, the covariance parameters used to construct \(\widehat\Sigma\) were estimated from the full sample. We selected the same covariates for both intercept and slope. 
%The two empirical data fission procedures used the same selection response and
%therefore produced the same selected covariates, and they differed only in the
%subsequent inference procedure. 
All data fission procedures used the balanced
allocation $\tau=1$. Each simulation condition was evaluated using
$R=500$ Monte Carlo replications.

\subsection{Data Generating Process}
Data were generated according to the semiparametric growth-curve model
introduced in \eqref{eq:eta},\eqref{eq:D}. Specifically, the subject-specific
latent intercept and slope were
\begin{equation}
	\eta_{1i}=m_1(X_i)+b_{1i},
	\qquad
	\eta_{2i}=m_2(X_i)+b_{2i}.
	\label{eq:simulation_latent}
\end{equation}
The random effects follow
{
	\renewcommand{\arraystretch}{0.8}
\[
b_i=(b_{1i},b_{2i})^\top
\sim
N\left\{
0,
\sigma_u^2
\begin{pmatrix}
	1 & 0.3\\
	0.3 & 1
\end{pmatrix}
\right\},
\]
}
and the independent measurement errors satisfied
$e_{ij}\sim N(0,9)$.

Baseline covariates followed an equicorrelated Gaussian distribution,
$X_i\sim N_p(0,\Sigma_\rho)$, where
$\Sigma_\rho=(1-\rho)I_p+\rho\mathbf{1}_p\mathbf{1}_p^\top$.
To allow unbalanced longitudinal designs, $T_i$, the number of measurements
for subject $i$, ranged from 3 to 8, with each count equally likely.
Each participant had a baseline measurement at $t_{i1}=0$, while the
remaining $T_i-1$ measurement times were sampled independently from
$\operatorname{Uniform}(0.1,1)$. 

\subsection{Mean Structure and Signal Calibration}
\label{subsec:mean_stru}

We considered both linear and nonlinear specifications of the mean
trajectory functions. For each case, we first define unscaled structural
functions $f_1$ and $f_2$ for the intercept and slope components,
respectively, and set
\(
m_\ell(X)
=
\beta_v f_\ell(X),
\ell=1,2,
\)
where $\beta_v>0$ is a common scaling parameter controlling the magnitude
of the covariate effects.

We first consider the linear mean trajectory model. The structural functions were defined as
\begin{equation}
f_1(X)
=
\sum_{k=1}^{10}X_k,
\qquad
f_2(X)
=
\sum_{k=6}^{15}X_k.
\label{eq:linear_dgp}
\end{equation}
Thus, variables $X_1,\ldots,X_{10}$ were structurally active for the
intercept component, whereas variables $X_6,\ldots,X_{15}$ were structurally
active for the slope component.
For the nonlinear mean trajectory model, %to evaluate performance under nonlinear mean structures
 we introduced cubic
transformations and pairwise interactions:
\begin{align}
f_1(X)
&=
\sum_{k=1}^{4}\frac{X_k^3}{\sqrt{15}}
+
\sum_{k=5}^{8}X_k
+
X_5X_9
+
X_6X_{10},
\label{eq:nonlinear_intercept}
\\
f_2(X)
&=
\sum_{k=6}^{9}\frac{X_k^3}{\sqrt{15}}
+
\sum_{k=10}^{13}X_k
+
X_{10}X_{14}
+
X_{11}X_{15}.
\label{eq:nonlinear_slope}
\end{align}
The factor $1/\sqrt{15}$ standardizes each cubic component to have unit
marginal variance. The linear components also have
unit marginal variance.
% For both mechanisms, the union of variables that were structurally active
% in either the intercept or slope component was
% \[
% S_0=\{1,\ldots,15\}.
% \]
Here, 
$f_1$ and $f_2$ have the same distribution up to a relabeling of the
covariates, and hence
\(
\Var\{f_1(X)\}
=
\Var\{f_2(X)\},
\)
and we define \(
V_m(\rho)
=
\Var\{f_1(X)\}
=
\Var\{f_2(X)\}
\). The exact expressions for $V_m(\rho)$ under the linear and nonlinear
trajectories are derived in Appendix~\ref{sec:signal_calibration}.

Following \citet{zhang2026variable}, we define a component-level
signal-to-noise ratio (SNR) for the Level-2 model as
\begin{equation}
s
=
\frac{\beta_v^2}{\sigma_u^2},
\label{eq:simulation_snr}
\end{equation}
where $\beta_v^2$ represents the common scale of an individual structural
component and $\sigma_u^2$ represents the unexplained Level-2 variation in
the latent growth parameters. Thus, $s$ controls the strength of an
individual structural component relative to the subject-specific
random-effect variation.

To make the overall variability of the latent intercept and slope comparable
across simulation settings, we fixed
\(
\Var(\eta_{1i})
=
\Var(\eta_{2i})
=
10.
\)
% Define
% \[
% V_m(\rho)
% =
% \Var\{f_1(X)\}
% =
% \Var\{f_2(X)\}.
% \]
Because $m_\ell(X)=\beta_v f_\ell(X)$ and
$\Var(b_{\ell i})=\sigma_u^2, \ell=1,2$, we have
\(
\Var(\eta_{\ell i})
=
\beta_v^2 V_m(\rho)
+
\sigma_u^2
\), and 
%Combining (\ref{eq:simulation_snr}) and (\ref{eq:var_dec}) gives
\begin{equation}
\sigma_u^2
=
\frac{10}{1+sV_m(\rho)},
\qquad
\beta_v
=
\sqrt{s\sigma_u^2}
=
\left\{
\frac{10s}{1+sV_m(\rho)}
\right\}^{1/2}.
\label{eq:simulation_calibration}
\end{equation}
Thus, for a fixed value of $s$, the scaling parameter $\beta_v$ generally
varies with $\rho$ because predictor correlation changes
$V_m(\rho)$. This calibration keeps the marginal variances of the latent
intercept and slope fixed across simulation conditions. %The exact expressions for $V_m(\rho)$ under the linear and nonlinear mechanisms are derived in Appendix~\ref{sec:signal_calibration}.

%\subsection{Experimental Design}
%
%The same simulation grid was used for the mixed-effects tree and
%LassoNet-style selectors. We varied the number of subjects over
%\[
%N\in\{200,500,1000\},
%\]
%while fixing the number of candidate baseline covariates at
%\[
%p\in\{100,200,500\}.
%\]
%The component-level signal-to-random-effect ratio was varied over
%\[
%s\in\{0.01,0.03,0.07,0.10\},
%\]
%and the equicorrelation among the baseline covariates was varied over
%\[
%\rho\in\{0,0.1,0.3,0.5\}.
%\]
%Both the linear and nonlinear mean structures were considered under every
%combination of $N$, $s$, and $\rho$, with
%\(
%\tau=1
%\text{, and }
%\sigma^2=9.
%\)
%
%The range of $N$ and $s$ allowed us to examine performance from small-sample,
%weak-signal settings to settings in which the active variables could be
%identified more reliably. Varying $\rho$ allowed us to assess the effect of
%predictor dependence on both variable selection and inference for the selected
%projection target. 

\subsection{Experimental Design}

The same simulation settings were used for the mixed-effects-tree and
LassoNet selectors. We considered the full factorial grid over
sample sizes $N\in\{200,500,1000\}$, numbers of candidate baseline covariates
$p\in\{100,200,500\}$, component-level signal-to-random-effect ratios
$s\in\{0.01,0.03,0.07,0.10\}$, and equicorrelations
$\rho\in\{0,0.1,0.3,0.5\}$. Both linear and nonlinear mean structures were
included in every combination, with $\tau=1$, $\sigma^2=9$, and 500 Monte
Carlo replications per setting.

This design allowed us to assess performance across varying sample sizes,
dimensions, signal strengths, and degrees of predictor dependence.

\subsection{Inferential Targets and Evaluation Criteria}

For each replication and each selected model, inference was evaluated against
the corresponding realized-design projection target
\begin{equation}
\beta_{\widehat M}^{D}
=
\left(
X_{\widehat M}^\top
\Sigma_0^{-1}
X_{\widehat M}
\right)^{-1}
X_{\widehat M}^\top
\Sigma_0^{-1}\mu.
\label{eq:simulation_conditional_target}
\end{equation}
For the \naive and data fission procedures, the target was calculated using
the full cohort. For data splitting, it was calculated using the realized
design of the inference subsample. Thus, each procedure was evaluated against
the design-conditional projection target corresponding to the observations
actually used for inference.

Throughout this subsection, the superscript $(r)$ denotes the quantity
corresponding to the $r$th Monte Carlo replication. Let
$\widehat S^{(r)}$ denote the variable-level selected set and
$\widehat M^{(r)}$ the corresponding selected coefficient-level model.
The realized-design projection target in replication $r$ is denoted by
$\beta_{\widehat M^{(r)}}^{D,(r)}$. We define $S_0$ as the active set.

Selection accuracy is evaluated at the variable level using the F1 score,
%For replication $r$, define true positive selection (TP), false positive selection (FP) and false negative selection (FN) by
%\[
%\operatorname{TP}^{(r)}
%=
%|\widehat S^{(r)}\cap S_0|,
%\qquad
%\operatorname{FP}^{(r)}
%=
%|\widehat S^{(r)}\setminus S_0|,
%\qquad
%\operatorname{FN}^{(r)}
%=
%|S_0\setminus\widehat S^{(r)}|.
%\]
which is the harmonic mean of the precision and recall, and is defined as
\[
F_1^{(r)}
=
\frac{
	2|\widehat S^{(r)}\cap S_0|
}{
	|\widehat S^{(r)}|+|S_0|
}.
\]
%\[
%F_1^{(r)}
%=
%\frac{
%2\operatorname{TP}^{(r)}
%}{
%2\operatorname{TP}^{(r)}
%+
%\operatorname{FP}^{(r)}
%+
%\operatorname{FN}^{(r)}
%}
%=
%\frac{
%2|\widehat S^{(r)}\cap S_0|
%}{
%|\widehat S^{(r)}|+|S_0|
%}.
%\]
We report the average F1 score across Monte Carlo replications.
%\begin{equation}
%\operatorname{F1}
%=
%\frac{1}{R}
%\sum_{r=1}^{R}
%F_1^{(r)}.
%\label{eq:selection_f1}
%\end{equation}
A larger F1 score indicates better recovery of the structurally active
variable set, with $F_1=1$ corresponding to exact recovery.

For inferential evaluation, we distinguish variable-level activity from
component-specific structural activity. Under both data-generating
processes, the structurally active coefficient-level set is
\(
M_0
=
\{1,\ldots,10\}
\cup
\{p+6,\ldots,p+15\},
\)
where the first set corresponds to covariates active in the intercept
component and the second to covariates active in the slope component.
Among selected structurally active coefficient terms, we report the
overestimation rate and the bias
\begin{align}
\operatorname{Overestimation}
&=
\frac{
\sum_{r=1}^{R}
\sum_{j\in\widehat M^{(r)}\cap M_0}
\mathbf 1
\left\{
\widehat\beta_j^{(r)}
>
\beta_{\widehat M^{(r)},j}^{D,(r)}
\right\}
}{
\sum_{r=1}^{R}
|\widehat M^{(r)}\cap M_0|
},
\label{eq:overestimation}\\
\operatorname{Bias}
&=
\frac{
\sum_{r=1}^{R}
\sum_{j\in\widehat M^{(r)}\cap M_0}
\left(
\widehat\beta_j^{(r)}
-
\beta_{\widehat M^{(r)},j}^{D,(r)}
\right)
}{
\sum_{r=1}^{R}
|\widehat M^{(r)}\cap M_0|
}.
\label{eq:simulation_bias}
\end{align}
Coverage and average confidence-interval length were evaluated over all
selected coefficient terms:
\begin{align}
\operatorname{Coverage}
&=
\frac{
\sum_{r=1}^{R}
\sum_{j\in\widehat M^{(r)}}
\mathbf 1
\left\{
\beta_{\widehat M^{(r)},j}^{D,(r)}
\in \mathrm{CI}_j^{(r)}
\right\}
}{
\sum_{r=1}^{R}
|\widehat M^{(r)}|
},
\label{eq:simulation_coverage}\\
\operatorname{Length}
&=
\frac{
\sum_{r=1}^{R}
\sum_{j\in\widehat M^{(r)}}
|\mathrm{CI}_j^{(r)}|
}{
\sum_{r=1}^{R}
|\widehat M^{(r)}|
},
\label{eq:simulation_length}
\end{align}
where $\mathrm{CI}^{r}_j$ represents the confidence interval for coefficient $j$ in replication $r$.
The nominal confidence level was $95\%$.

\subsection{Results}
\label{sec:simulation-results}

We present results for the neural-network selector under the nonlinear mean
trajectory model. The linear mean trajectory model and the mixed-effects-tree
selector showed similar qualitative patterns, and their results are reported
in the Appendix. We first examine the setting with $N=200$ and
$p=500$, including both the oracle and practical procedures described in
Section~\ref{sec:obj_over}. We then examine the practical procedures by
varying $p$ at $N=500$ and varying $N$ at $p=500$, with $\rho=0.3$ fixed
in both comparisons.

Throughout the figures, coverage and overestimation rates are reported as
deviations from their reference values of $0.95$ and $0.50$, respectively,
while bias is reported on its original scale. Values near zero therefore
indicate agreement with the corresponding reference values. Error bars
represent $95\%$ bootstrap Monte Carlo confidence intervals based on the
500 simulation replications.

\subsubsection{Oracle Post-Selection Inference}
Figure~\ref{fig:nn-inference-validity_oracle} summarizes the oracle results
following neural-network selection. Both oracle data fission and oracle data
splitting were well calibrated, with coverage close to $0.95$, overestimation
rates close to $0.50$, and negligible bias. Oracle data fission generally
achieved higher F1 scores and shorter confidence intervals than data splitting,
indicating greater efficiency while maintaining valid post-selection inference.
\begin{figure}[ht!]
	\centering
	\includegraphics[width=\textwidth]{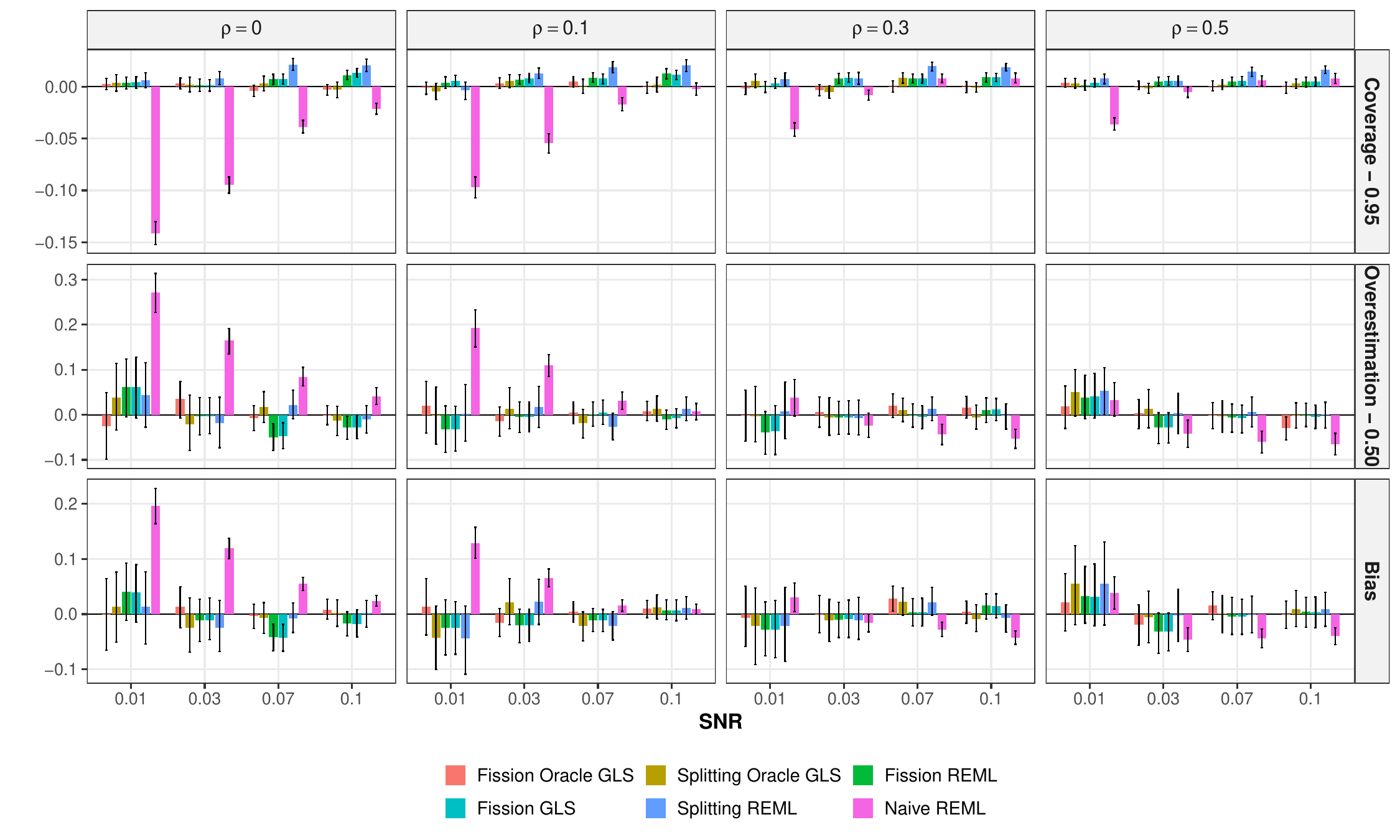}
    \vspace{-0.5em}
	\includegraphics[width=\textwidth]{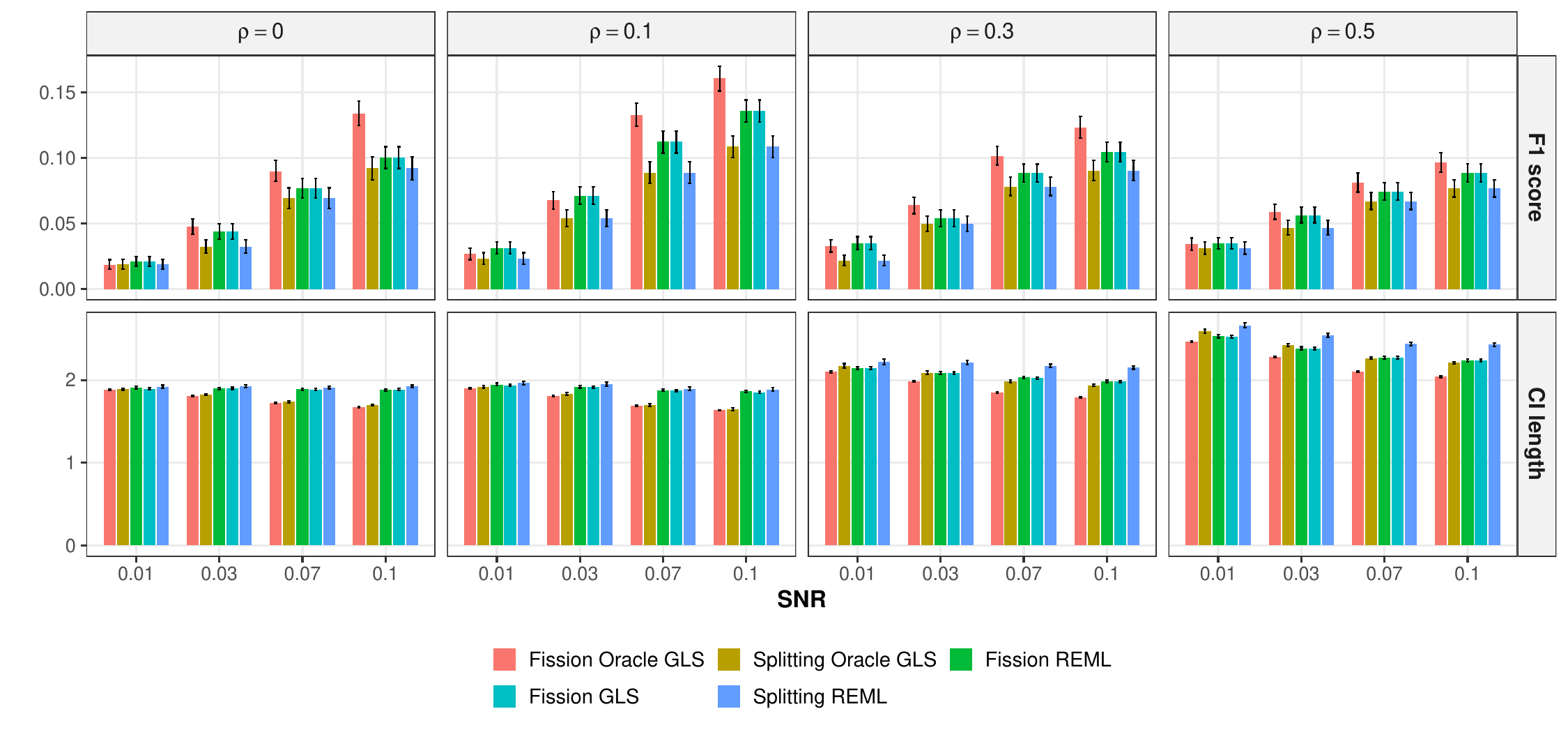}
	\caption{Inferential performance and selection accuracy of neural network selector with nonlinear mean trajectory with $N=200$ and $p=500$.
	}
	\label{fig:nn-inference-validity_oracle}
\end{figure}

\subsubsection{Post-Selection Inference with Estimated Covariance}
\label{sec:simulation-fixed}
%\begin{figure}[ht!]
%	\centering
%	\includegraphics[width=0.8\textwidth]{./graph/nn_inference_deviation_p500_nonlinear_N200.pdf}
%	\vspace{-0.5em}
%    \includegraphics[width=0.8\textwidth]{./graph/nn_f1_ci_length_p500_nonlinear_N200.pdf}
%	\caption{Validity and effect of neural network selector with nonlinear mean trajectory with $N=200$ and $p=500$.
%	}
%	\label{fig:nn-inference-validity}
%\end{figure}
The upper panel of Figure~\ref{fig:nn-inference-validity_oracle} summarizes the
inferential performance under the estimated covariance setting with
$N=200$ and $p=500$. Empirical fission with plug-in GLS and empirical
fission with REML both maintained coverage close to the nominal level $0.95$,
overestimation rates close to $0.50$, and small signed bias across most
settings. The two fission implementations produced very similar results.
Subject-level data splitting with REML was also generally well calibrated,
although it was slightly conservative in some settings.

In contrast, \naive REML showed departures from the reference values,
particularly when the covariate correlation was small. Under weak signals,
its coverage was substantially below $0.95$, accompanied by positive
deviations in the overestimation rate and signed bias. These results
illustrate the inferential distortion caused by reusing the same response
for both variable selection and subsequent inference without accounting
for the selection step.

The difference between \naive and adjusted inference became smaller in
several settings as $\rho$ increased. One possible explanation is that
strongly correlated predictors provide partially interchangeable
representations of the same signal. For illustration, suppose that
standardized predictors $X_1$ and $X_2$ have correlation $\rho$ and
that the true signal is $\beta X_1$. The best linear projection of this
signal onto $X_2$ is $\beta\rho X_2$, with mean squared approximation
error
\(
\beta^2(1-\rho^2).
\)
The loss from selecting $X_2$ instead of $X_1$ therefore decreases as
the correlation increases. In addition, stronger covariate correlation
leads to longer confidence intervals, which can further reduce the
observed difference in coverage. This pattern should not be interpreted
as showing that \naive post-selection inference becomes valid under
strong correlation. The \naive procedure still lacks an adjustment for
selection and continued to exhibit bias and undercoverage in a number of
settings.
% \subsubsection{Selection Accuracy and Interval Length}
% \label{sec:simulation-efficiency}
% The lower panel of Figure~\ref{fig:nn-inference-validity_oracle} compares
% the F1 score and length of confidence interval under the estimated covariance setting with
% $N=200$ and $p=500$. Because fission plug-in GLS
% and fission REML use the same selection response and therefore produce the
% same selected model, their F1 scores are identical by construction.

The F1 score generally increased with the signal-to-noise ratio. The two
data fission procedures, using REML and plug-in GLS for inference,
produced nearly identical F1 scores because they shared the same selection
response and therefore selected the same models. Empirical fission generally
achieved higher F1 scores than subject-level data splitting, particularly at
moderate and strong signal levels. F1 scores tended to decrease as the
covariate correlation increased, reflecting the greater difficulty of
identifying the structurally active variables when multiple predictors
contained similar information.

Mean confidence-interval length generally increased with the covariate
correlation, consistent with greater collinearity among the selected
predictors. The two data fission procedures also produced very similar
interval lengths across settings, and both yielded shorter intervals than
subject-level data splitting.

\Naive REML is not included in the interval-length comparison because its
inference was not well calibrated. Consequently, shorter intervals from the
\naive\ procedure would not represent a valid gain in inferential precision.

\subsubsection{Varying the Number of Subjects at $p=500$}
\label{sec:simulation-vary-N}
\begin{figure}[ht]
	\centering
	\includegraphics[width=\textwidth]{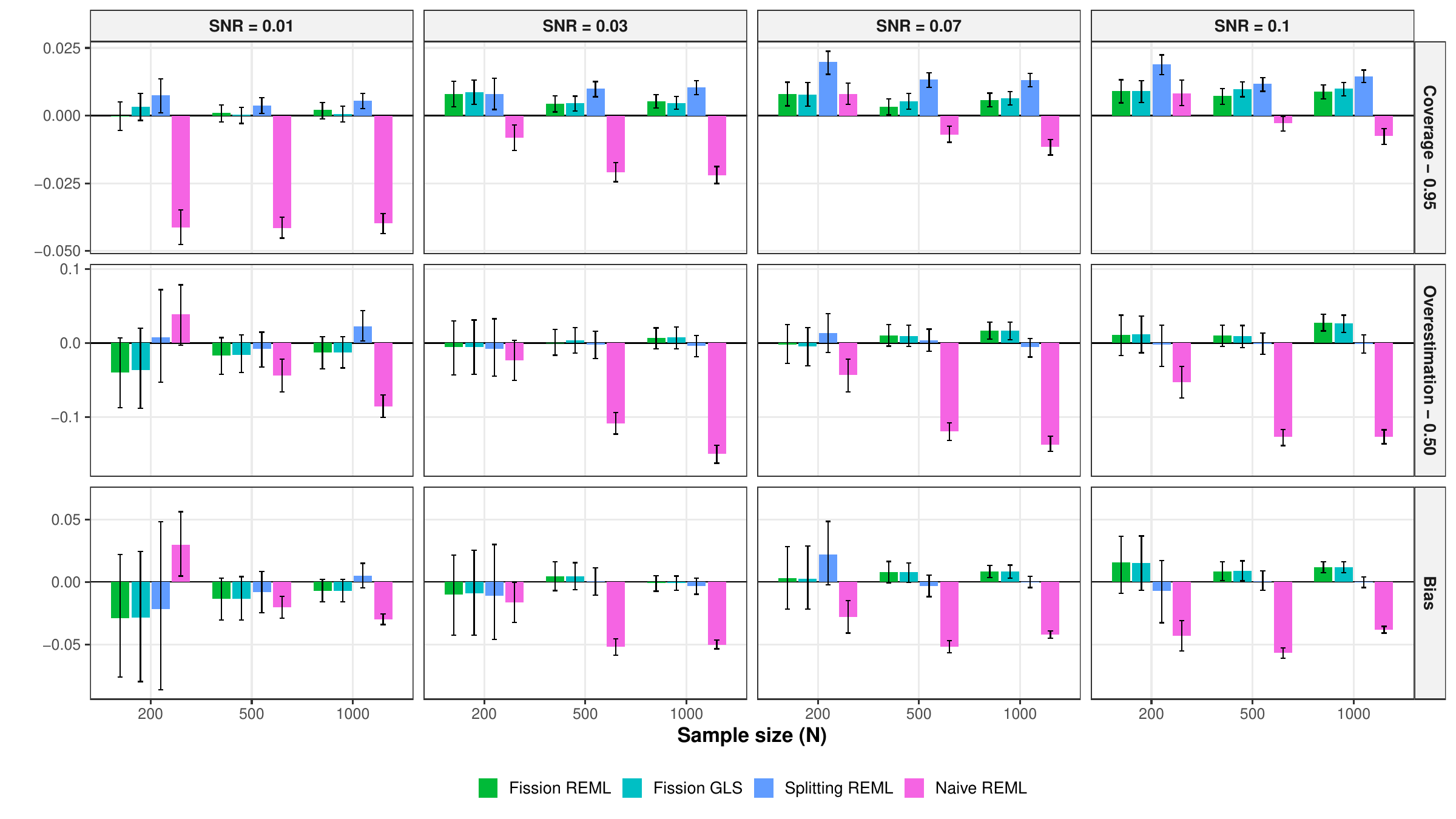}
    \vspace{-0.5em}
	\includegraphics[width=\textwidth]{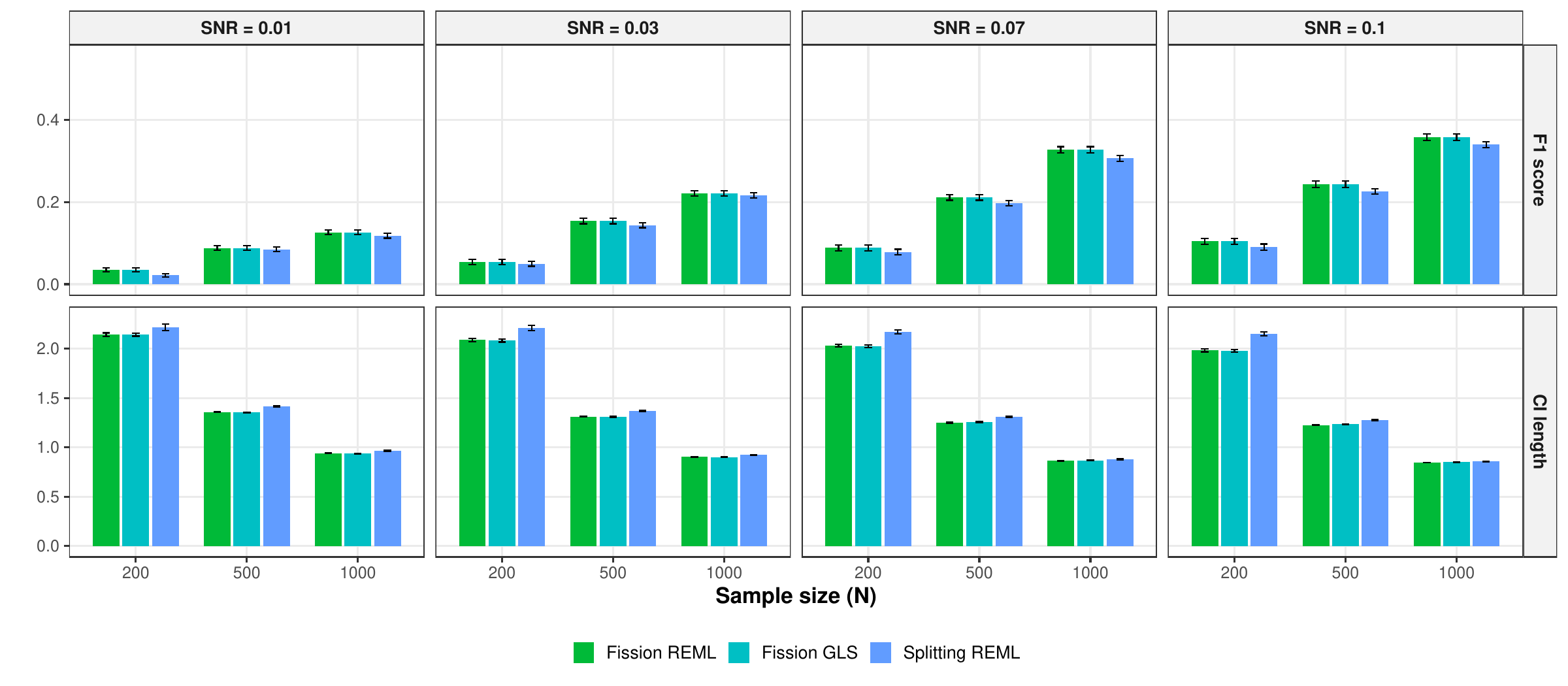}
	\caption{Inferential performance and selection accuracy of the neural network selector with nonlinear mean trajectory with $p=500$ and $\rho=0.3$.
	}
	\label{fig:nn-N}
\end{figure}
Figure~\ref{fig:nn-N} compares the data splitting and fission for
$N\in\{200,500,1000\}$ at $p=500$ and $\rho=0.3$. Increasing the
sample size generally improved F1 scores and substantially reduced
confidence-interval lengths. Increasing $N$, however, did not eliminate the inferential distortions
of \naive REML. Departures from the reference values persisted, although the variable-selection accuracy improved. Thus, increasing the sample
size and improving variable recovery alone were insufficient to ensure
valid post-selection inference when the same response was reused for
both selection and inference.

\subsubsection{Varying the Number of Covariates at $N=500$}
\label{sec:simulation-vary-P}
\begin{figure}[ht]
	\centering
	\includegraphics[width=\textwidth]{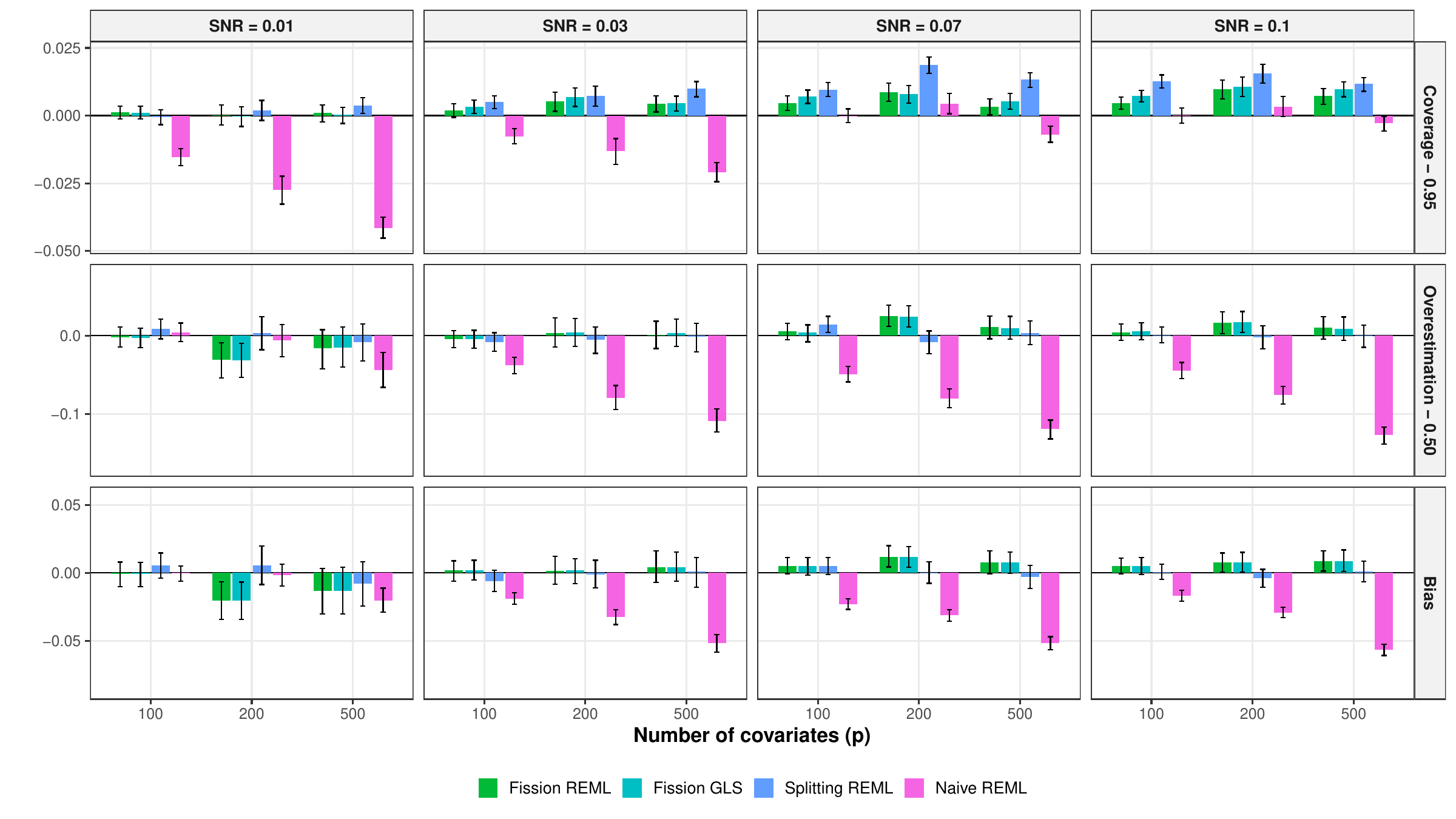}
    \vspace{-0.5em}
	\includegraphics[width=\textwidth]{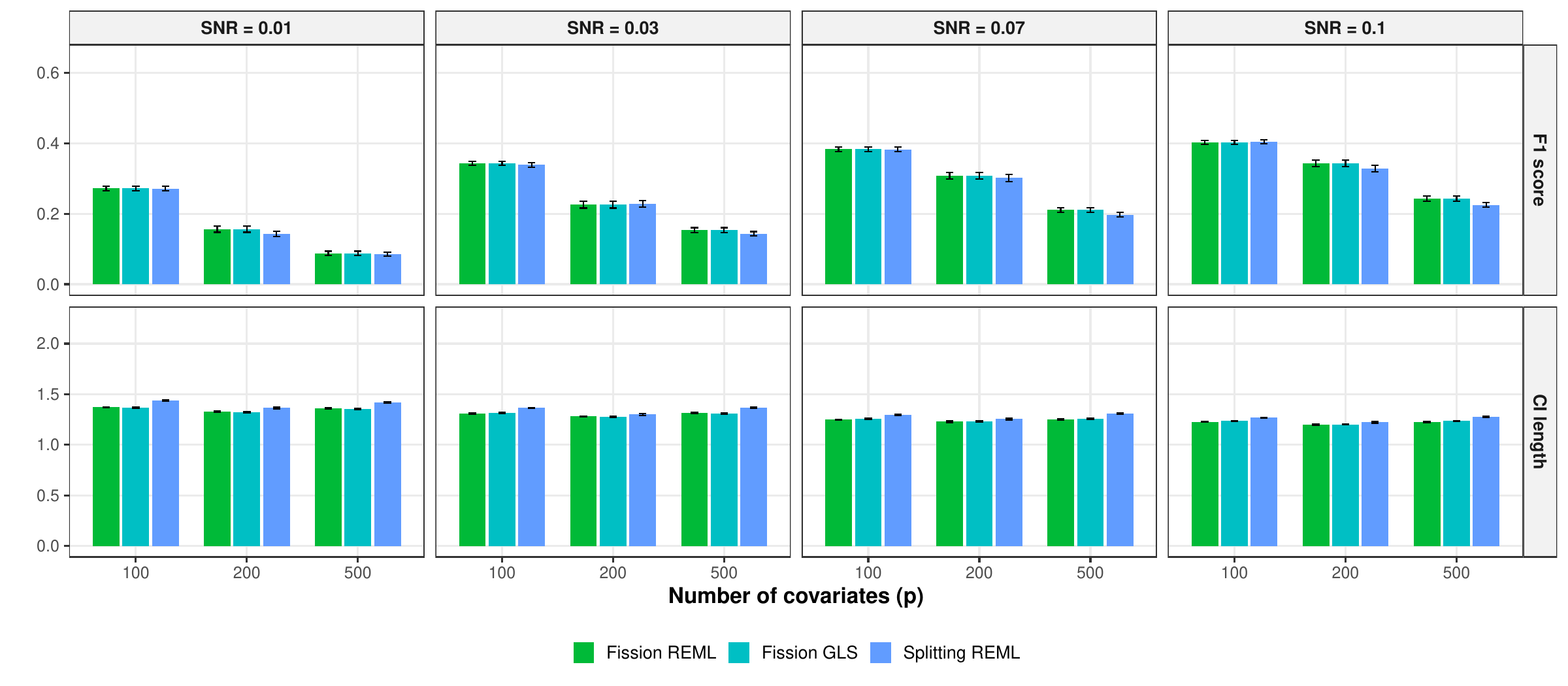}
	\caption{Inferential performance and selection accuracy of the neural network selector with nonlinear mean trajectory with $N=500$ and $\rho=0.3$.
	}
	\label{fig:nn-p}
\end{figure}
Figure~\ref{fig:nn-p} compares the data splitting and fission for $p\in\{100,200,500\}$ at $N=500$ and $\rho=0.3$ .
As $p$ increased, the F1 score generally decreased, reflecting the greater
difficulty of identifying the active variables from a larger
candidate set, while confidence interval length changed relatively little across different numbers of covariates $p$. The inferential distortions of \naive REML became
more pronounced as the number of candidate covariates
increased, particularly under weak signals, while the other procedures remained valid.
\section{Real Data Analysis}
\label{sec:real-data}

We illustrate the proposed selective-inference framework using data from
the Longitudinal Study of American Youth (LSAY)
\citep{millerLSAY2021}. 
%LSAY is a longitudinal study supported by the
%National Science Foundation that was designed to investigate the
%development of students' mathematics and science achievement and
%attitudes, together with student, family, school, teacher, and
%educational characteristics.
We focus on the younger cohort, which was followed annually from the
seventh through the twelfth grade. The outcome of interest is the
longitudinal mathematics achievement score. Measurement time was coded
according to study year as $t=0,\ldots,5$, with the first measurement
occasion coded as $t=0$.

We revisit the LSAY data analyzed by
\citet{zhang2026variable}, using the same general
covariate-prescreening strategy but focusing here on post-selection
inference. Variables collected after high school were excluded, and
baseline measurements were retained for covariates measured repeatedly
over time. When multiple categorical variables represented similar
constructs, the variable with fewer response categories was retained.
Nominal variables were represented by dummy indicators, ordinal variables
retained their original ordering, and continuous covariates were
standardized before variable selection.

We first removed observations with missing values in either the
mathematics outcome or any of the retained covariates. We then retained
participants with a baseline observation and at least one follow-up
observation. The resulting analytic sample contained 1,309 students,
6,087 student-wave observations, and 51 candidate baseline covariates.

We applied the mixed-effects model tree and LassoNet-style neural-network
selectors described in Section~\ref{sec:working_methods}. For each
selector, we considered four inferential procedures: \naive REML,
subject-level data splitting with REML, empirical fission with plug-in
generalized least squares (GLS), and empirical fission with REML. For
both fission procedures, selection was performed on the same selection
response, so that they produced the same selected working model and differed
only in the inference procedure.

Figure~\ref{fig:real-selection} summarizes the variables selected by the
two working selectors. For the neural-network selector, empirical fission,
data splitting, and \naive selection selected four, two, and four variables,
respectively. For the mixed-effects tree, seven variables were selected
under each procedure, although the selected sets were not identical
across the three selection schemes.

For comparison, in our previous analysis of the LSAY data
\citep{zhang2026variable}, BMLEV3 (math track level) was among the most
consistently identified covariates, with evidence of association with both the
baseline level and the growth rate of mathematics achievement. MTHV9\_7
(math course content level at Grade 7) was also identified as a potentially
important predictor of the baseline level, while BMSCHTRK (ability grouping
in math classes) was consistently identified for the baseline component.
In the present analysis, BMLEV3 was again selected under empirical fission
by both working selectors, while MTHV9\_7 was selected by the mixed-effects
tree. BMSCHTRK was not selected by either empirical-fission selector.

\begin{figure}[!htbp]
	\centering
	\includegraphics[width=\textwidth]{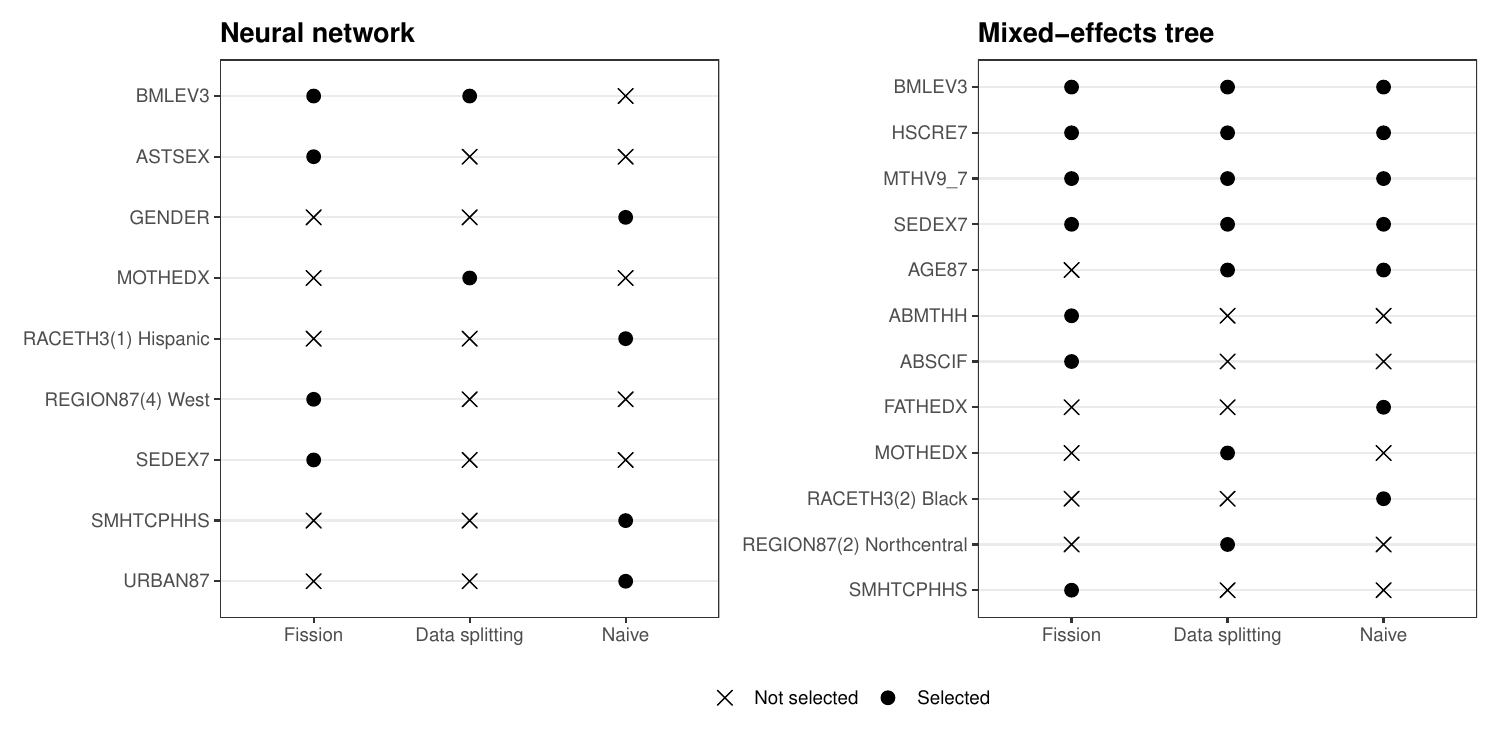}
	\caption{Covariates selected by the neural-network and mixed-effects-tree
		working selectors under empirical fission, subject-level data splitting,
		and \naive selection.}
	\label{fig:real-selection}
\end{figure}

Figure~\ref{fig:real-fission} compares plug-in GLS and REML inference
for the fission-selected models. For both the neural-network and
mixed-effects-tree selectors, the two procedures produced very similar
point estimates and confidence intervals across the selected covariates.
This empirical agreement suggests that, in this application, refitting
the selected model by REML leads to conclusions similar to those from GLS.

\begin{figure}[!htbp]
	\centering
	\includegraphics[width=\textwidth]{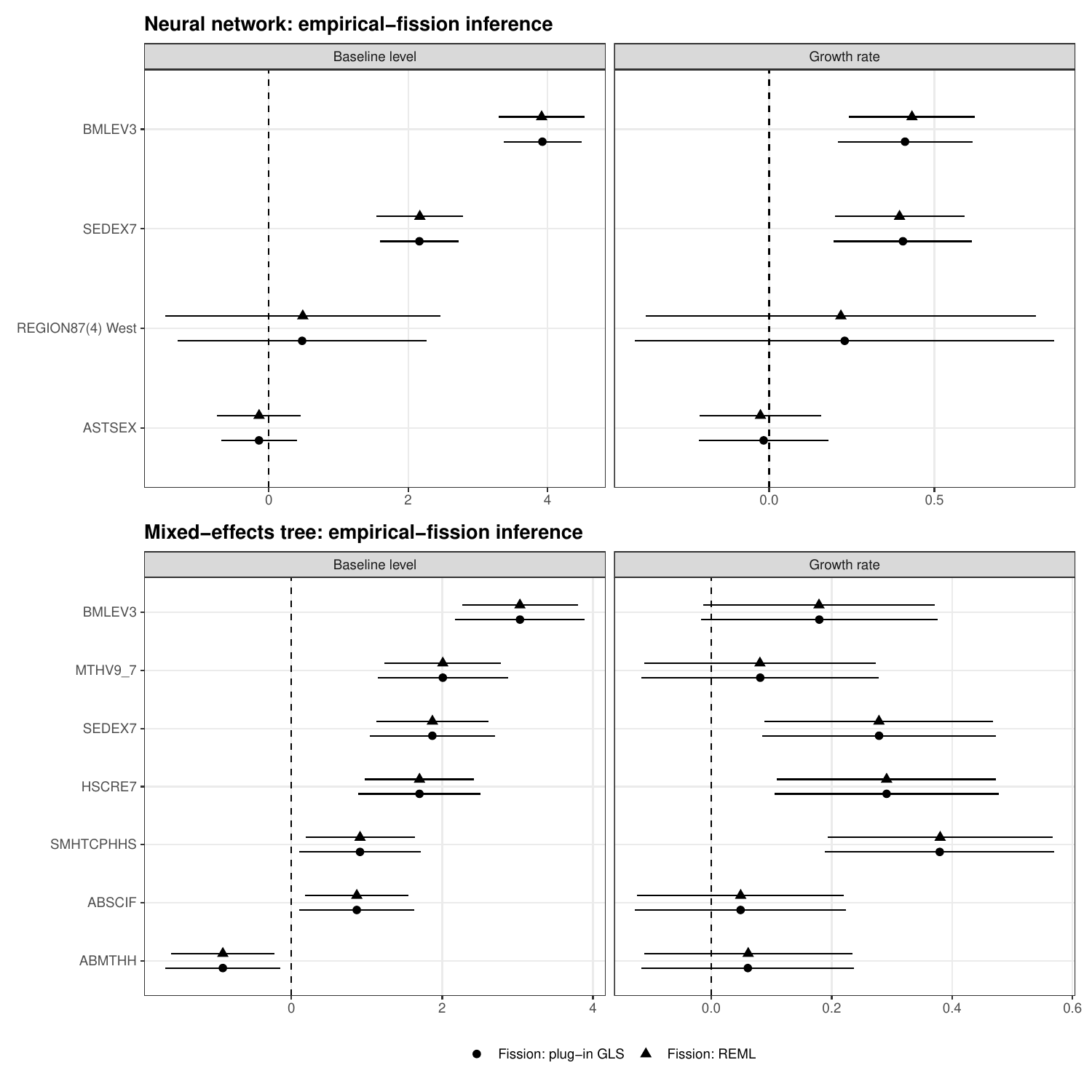}
	\caption{Post-selection projection-coefficient estimates and pointwise
		95\% confidence intervals for the empirical-fission-selected models.}
	\label{fig:real-fission}
\end{figure}
\section{Discussion}
\label{sec:discussion}

In this paper, we develop a post-selection inference framework for
longitudinal growth-curve models based on data fission. The proposed procedure provides exact
inference when the covariance structure is known and asymptotically valid
inference when it is estimated.
%A key challenge in longitudinal data is that repeated measurements from the same participant are correlated, so any valid inferential procedure must account for the within-subject covariance structure. 
%When the outcome covariance structure is known, we construct separate response vectors for selection and inference by adding
%and subtracting Gaussian noise in a way that ensures the two randomized
%vectors of responses are independent conditional on the observed design. When the
%covariance structure is unknown, we extend the construction by replacing the
%true covariance matrix with an estimated one. We establish asymptotic validity
%for the resulting plug-in procedure under suitable conditions, and our
%simulations show that the practical empirical-fission implementation perform well in finite samples. 
%When the
%covariance matrix is estimated sufficiently accurately, the residual
%dependence between the two responses becomes small. After selection, either plug-in GLS or REML can be
%used for inference on the selected working model. Plug-in GLS is directly
%supported by our theoretical analysis, whereas REML follows the more familiar
%mixed-effects modeling workflow. In our simulations and real-data analysis,
%the two approaches produce very similar estimates and confidence intervals.
An important advantage of the proposed framework is its flexibility with
respect to the choice of variable-selection procedure. The mixed-effects
tree and neural-network-based procedures considered in
Section~\ref{sec:working_methods} are only illustrative. Other
machine-learning or variable-selection methods, as well as different tuning
or hyperparameter choices, can be incorporated into the same framework
without changing the subsequent inferential procedure.
%This flexibility is particularly useful for modern selection methods with
%complicated selection events, allowing the selection procedure to be treated
%largely as a black box for subsequent inference.
%This flexibility is particularly useful because many modern
%variable-selection methods induce complicated selection events that are
%difficult or impractical to characterize analytically. Data fission sidesteps this difficulty by
%separating the information used for selection from that used for inference.
%In practice, once the model has been selected, standard estimation and
%inference can be carried out using the inference response for the selected
%working model. 
Thus, the selection procedure can be treated essentially as a
black box, without explicitly characterizing its selection event or deriving
a selector-specific conditional distribution.

A natural alternative to data fission is subject-level data splitting.
Both approaches separate the information used for variable selection from
that used for subsequent inference, but they allocate the available
information differently. Data splitting assigns different participants to
the selection and inference stages, so only a subset of the sample contributes
to each stage. Data fission retains all participants
	at both stages, but relies on stronger assumptions about Gaussianity and the
	covariance structure.
%In contrast, data fission retains all participants in both
%stages and separates the information used for selection and inference by
%adding appropriately structured Gaussian random noise, which has the same covariance structure as the outcome. When
%the covariance structure is reasonably well specified, this construction
%can use the available information more efficiently. Consistent with this
%intuition, our simulations show that data fission generally achieved better
%variable-selection accuracy and shorter confidence intervals than
%subject-level data splitting while maintaining similar inferential
%calibration.
%These efficiency gains, however, come with stronger modeling assumptions.
%In particular, the construction studied here relies on the Gaussian
%assumption and the covariance structure induced by the growth-curve model
%considered in this paper. Different longitudinal models may imply different
%forms of the marginal covariance matrix and require different procedures for
%estimating its variance components. Consequently, the proposed fission
%construction may need to be adapted before being applied to other
%longitudinal models. 
When these 
assumptions are questionable, or when the
covariance structure cannot be reliably specified or estimated, subject-level
data splitting may provide a simpler and more robust alternative. Although the
general data fission framework of \citet{leiner2025data} also provides
constructions beyond the Gaussian setting, these constructions are not
directly tailored to the correlated longitudinal mixed-effects setting
considered here. %Extending data fission to longitudinal outcomes under weaker distributional assumptions is therefore an important direction for future work.

In this paper, we focus on inference for linear projection parameters in the
selected working model. This makes the definition and interpretation of the
inferential target particularly important. When the relationships between the covariates and the growth parameters are
nonlinear,
there may be no uniquely defined linear coefficient that directly represents
the association of a particular covariate with the trajectory. Even under a
linear mean structure, when covariates are correlated, the coefficient of a
given covariate can depend on which other covariates are included in the
working model. We therefore define the inferential target as the
covariance-weighted linear projection under the selected working model. This
target summarizes how the selected covariates are linearly associated with
the underlying mean trajectory within the selected model. It is also important
to distinguish whether a variable is useful for describing the trajectory
from whether its corresponding projection coefficient is nonzero. A variable
may be important for capturing a nonlinear effect or interaction while its
linear projection coefficient is small or even zero. Thus, variable selection
and subsequent projection inference address related but distinct questions:
the former identifies variables useful for describing trajectory
heterogeneity, whereas the latter quantifies their best linear representation
within the selected working model.

Several extensions are of interest for future work. First, extending the
framework to non-Gaussian longitudinal outcomes would broaden its range of
applications. Second, more general random-effects and within-participant
dependence structures could be accommodated beyond the growth-curve
covariance model considered here. Third, extending the framework to
time-varying covariates would allow selection and inference to incorporate
predictors that evolve over time.
%Several extensions are of interest for future work. First, the current
%construction relies on the Gaussian assumption for the longitudinal outcome.
%Extending data fission to discrete outcomes and non-Gaussian continuous outcomes would
%broaden the range of potential applications. Second, the covariance structure
%considered in this paper is induced by the specific growth-curve model used
%throughout the analysis. Other longitudinal models may involve different
%random-effects or within-subject dependence structures, and developing a more
%general framework that accommodates such structures would be worthwhile. A
%further extension is to allow time-varying covariates. Many longitudinal
%studies include covariates that evolve over time. Extending both variable
%selection and subsequent projection inference to this setting would broaden
%the practical scope of the proposed framework. %Finally, developing theoretical guarantees for full-sample covariance estimation, which is more commonly used in practice than the pilot-assisted construction considered in our theory, is another important direction for future work.\textcolor{red}{not quite understanding the last sentence}
\section*{Acknowledgment}
The research of Cao is partially supported by the University of Wisconsin-Madison Office of the Vice Chancellor for Research (OVCR), NSF 2553817, and NSF 2311249. The research of Zhang is partially supported by the Institute of Education Sciences R305B240029. 
\newpage
\appendix
\begin{center}
	\Large\bfseries Appendix
\end{center}
\vspace{1em}
\begingroup
\pretocmd{\section}{\clearpage}{}{}

\section{Discussion of assumptions}
\label{sec:assum_discussion}
Assumption~1 adopts the usual i.i.d. participant-level sampling framework.
The positivity conditions are standard in growth-curve models \citep{bollen2006latent}. 

Assumption~2 is inspired by \citet{rasines2023splitting}, who estimate the
variance using only a subset of the observations in order to limit the
dependence between variance estimation and subsequent inference. In their
theoretical analysis, the variance estimator is based on the first
$\lfloor n/2\rfloor$ observations and is assumed to satisfy root-$n$ rate. Their result
is developed for a linear-model setting and imposes additional structure on
the selection event. In contrast, our framework allows a general measurable
selection procedure and requires the pilot sample to satisfy
$N_s/N\to 0$. Consequently, the covariance estimator is only required to
converge at the corresponding root-pilot-sample-size rate,
$O_p(N_s^{-1/2})$.

 The root-$N_s$ covariance-estimation rate in Assumption~2 can also be
motivated by a standard residual-based covariance decomposition.  For illustration, we consider a simplified setting in which each participant
has the same fixed number of repeated observations. Suppose that a conditional mean function $\mu(\cdot)$ is first
estimated by $\hat \mu(\cdot)$ and that the covariance is subsequently estimated from the
resulting residuals $Y_i-\hat \mu(X_i)$. Let
\[
\Delta_i=\widehat{\mu}(X_i)-\mu(X_i),
\qquad
\varepsilon_i=Y_i-\mu(X_i).
\]
For a residual-based covariance estimator, its difference from the true
covariance can be decomposed schematically as
\[
\widehat{\Sigma}-\Sigma_0
=
\frac{1}{N_s}\sum_{i=1}^{N_s}
\{\varepsilon_i\varepsilon_i^\top-\Sigma_0\}
-
\frac{1}{N_s}\sum_{i=1}^{N_s}
\{\varepsilon_i\Delta_i^\top+\Delta_i\varepsilon_i^\top\}
+
\frac{1}{N_s}\sum_{i=1}^{N_s}
\Delta_i\Delta_i^\top.
\]
The
first term has the usual order
$O_p(N_s^{-1/2})$. If
\(
r_{N_s}
=
\|\widehat{\mu}-\mu\|^2_{L_2(P)}
\)
, the cross term is typically of
order $O_p(r_{N_s}^{1/2}N_s^{-1/2})$, while the quadratic remainder is of order
$O_p(r_{N_s})$. Consequently,
$r_{N_s}=O_p(N_s^{-1/2})$ is sufficient for the residual-based covariance
estimator to retain the root-$N_s$ rate. 

Rates of the required order are attainable under specific regularity
conditions for both tree-based and neural-network estimators. For example,
\citet{mourtada2020minimax} establish minimax prediction rates for Mondrian
trees and forests under sufficiently smooth, low-dimensional settings. These rates
make $r_{N_s}$ of order $N_s^{-1/2}$ or smaller. Analogous sufficiently
fast rates are available for deep ReLU networks under smoothness and
structural assumptions \citep{schmidthieber2020nonparametric}. These results
concern mean-function estimation rather than the specific mixed-effects
covariance estimator used here, but they illustrate settings in which the
root-$N_s$ rate assumed in Assumption~2 is attainable.

Assumption~3 is motivated by Theorem~\ref{thm:beta_est}. When
$M_I=M_S$, the population projection target is invariant to the covariance
parameter $\theta$. This invariance eliminates changes in the projection
target caused solely by replacing the true covariance parameter with an
estimated one. The condition defining $B_0(\theta)$ ensures that the
covariance weighting does not depend on the baseline covariates at the
population level. A sufficient condition is that the number and timing of
repeated measurements are independent of the baseline covariates and that
the covariance structure depends on the observation schedule only through
the common parameter $\theta$.

Assumption~4 imposes standard population-level design and moment regularity
conditions. The eigenvalue condition is imposed on the full working design
matrix and ensures that the population design is nondegenerate. Similar
conditions are commonly used in selective-inference theory
\citep[e.g.,][]{zhao2022selective}. Since $X_{i,M}$ is obtained by selecting
columns from the full design matrix $X_{i,\mathrm{full}}$, the same lower and upper
eigenvalue bounds automatically carry over to every candidate submodel $M$;
see Lemma~\ref{lem:full_to_submodel} in appendix. For the residual term, a sufficient set
of primitive conditions is a finite fourth moment of the covariance-weighted
full design,
\[
\mathbb E\left[
\left\|
\Sigma_i(\theta_0)^{-1/2}X_{i,\mathrm{full}}
\right\|_s^4
\right]<\infty,
\]
together with
\[
\mathbb E\left[
\left\|
X_{i,\mathrm{full}}^\top\Sigma_i(\theta_0)^{-1}\mu_i
\right\|_2^2
\right]<\infty.
\]
Related fourth-moment conditions on the design and moment conditions on
design-response products are used in regression with submodel
\citep{kuchibhotla2023uniform}.

Assumption~5 requires the probability of the conditioning event to be
bounded away from zero. Such non-vanishing selection probability
conditions are common in asymptotic selective-inference theory
\citep[e.g.,][]{rasines2023splitting,tian2017asymptotics}.

\section{Proofs}
\label{sec:proof}
\subsection{Notations}
\label{sec:estimated_covariance_notation}

We briefly collect the notations used repeatedly below. There are $N$
independent participants, with total number of observations
$n=\sum_{i=1}^N T_i\asymp N$. Conditional on the realized design
$\mathcal D_N$ in \eqref{eq:D_N}, from \eqref{eq:Y_i}, we have
\[
Y_i\sim N_{T_i}(\mu_i,\Sigma_{0i}),
\quad\mu_i=Z_im(X_i),\quad 
\Sigma_{0i}=Z_iD_0Z_i^\top+\sigma_0^2I_{T_i},\quad i=1,\ldots, N.
\]
Define
\[
\Sigma_0=\operatorname{blockdiag}
(\Sigma_{01},\ldots,\Sigma_{0N}).
\]

For a fixed candidate model $M$, let $X_M$ denote its working design. The
realized-design and population projection targets are denoted by
\begin{equation}
\label{eq:beta_MN}
\beta_{M,N}^{D}
=
(X_M^\top\Sigma_0^{-1}X_M)^{-1}
X_M^\top\Sigma_0^{-1}\mu
\end{equation}
and 
\[\beta_M^P=\{\mathbb E(X_{i,M}^\top\Sigma_{0i}^{-1}X_{i,M})\}^{-1}
\mathbb E\{X_{i,M}^\top\Sigma_{0i}^{-1}\mu_i\}.\]
Here $\beta^D_{M,N}$ has the same meaning as $\beta_{M}^D$ defined in \eqref{eq:beta_D}, and $N$ is used to emphasize the sample size. Their residuals are
\[
r_M=\mu-X_M\beta_{M,N}^{D},
\qquad
r_{Mi}^P=\mu_i-X_{i,M}\beta_M^P.
\]
In particular,
\[
X_M^\top\Sigma_0^{-1}r_M=0.
\]
The covariance is parameterized by
\[
\theta
=
(\sigma^2,d_{11},d_{12},d_{21},d_{22})^\top
\in\Theta\subset\mathbb R^5,
\]
and define, 
\begin{equation}
\label{eq:var_para}
\sigma^2(\theta)=\sigma^2, \quad \text{and}\quad
D(\theta)
=
\begin{pmatrix}
d_{11} & d_{12}\\
d_{21} & d_{22}
\end{pmatrix},
d_{12}=d_{21}.
\end{equation}
The parameter space $\Theta$ is restricted so that
$D(\theta)$ is positive definite and $\sigma^2(\theta)>0,\forall \theta\in \Theta$.

Let $\mathcal P$ be a subject-level pilot set, with index uniformly drawn from $1,\ldots, N$,
and the size is
\(
N_s=|\mathcal P|.
\)
Define its complement \(
\mathcal P^c=\{1,\ldots,N\}\setminus\mathcal P.\)
Let $\widehat\theta$ be the resulting covariance parameter estimated from $\mathcal P$. The estimated variance is
\[
\widehat\Sigma_i=\Sigma_i(\widehat\theta),\text{ and } 
\widehat\Sigma
=
\operatorname{blockdiag}
(\widehat\Sigma_1,\ldots,\widehat\Sigma_N).
\]
The empirical fission constructs the selection and inference responses as follows:
\[
\widehat Y^{\mathrm{sel}}=Y+\widehat W,
\qquad
\widehat Y^{\mathrm{inf}}=Y-\widehat W,
\]
where
\[
\widehat W\mid Y,\mathcal D_N,\mathcal P\sim N_n(0,\widehat\Sigma).
\]
For a fixed model $M$, let 
\[
\mathcal A_M=\{y:\mathcal S(y)=M\},
\]
defined in \eqref{eq:A_M}, consist of all possible response vectors that
would lead the selection procedure $\mathcal S$ to select exactly the
model $M$.

Define
\begin{equation}
\label{eq:AM_hat}
    A_M
=
(X_M^\top\Sigma_0^{-1}X_M)^{-1}X_M^\top\Sigma_0^{-1},
\qquad
\widehat A_M
=
(X_M^\top\widehat\Sigma^{-1}X_M)^{-1}X_M^\top\widehat\Sigma^{-1},
\end{equation}
and
\[
H_M=X_M^\top\Sigma_0^{-1}X_M,
\qquad
\widehat H_M=X_M^\top\widehat\Sigma^{-1}X_M.
\]
Then the target in \eqref{eq:beta_MN} and its estimate can be written as
\[
\beta_{M,N}^{D}=A_M\mu,
\qquad
\widehat\beta_M=\widehat A_M\widehat Y^{\mathrm{inf}},
\]
and the estimated covariance matrix for $\hat \beta_{M}$, introduced in \eqref{eq:hat_vm}, is
\begin{equation}
\label{eq:hat_VM}
\widehat V_M=2(X_M^\top \widehat\Sigma_{}^{-1}X_M)^{-1}=2\widehat H_M^{-1}=2\widehat A_M\widehat \Sigma\widehat A_M^\top .
\end{equation}

For $\theta\in \Theta_0$ where $\Theta_0$ is the fixed neighborhood of $\theta_0$, define
\[
E(\theta)
=
\Sigma_0^{-1/2}\{\Sigma(\theta)-\Sigma_0\}\Sigma_0^{-1/2}.
\]
The relative covariance perturbation is
\[
\widehat E
=
\Sigma_0^{-1/2}(\Sigma(\hat\theta)-\Sigma_0)\Sigma_0^{-1/2}.
\]

Now we define notations that are specifically used for pilot and non-pilot subjects. Without loss of generality, we order the subjects so that the pilot
subjects are indexed by
\[
\mathcal P=\{1,\ldots,N_s\},
\]
and the non-pilot subjects are indexed by
\[
\mathcal P^c=\{N_s+1,\ldots,N\}.
\]
Accordingly, define
\[
n_{\mathcal P}
=
\sum_{i=1}^{N_s}T_i,
\qquad
n_{\mathcal P^c}
=
\sum_{i=N_s+1}^{N}T_i.
\]
We partition the stacked quantities according to the pilot and
non-pilot subjects:
\[
\mu
=
\begin{pmatrix}
	\mu_{\mathcal P}\\
	\mu_{\mathcal P^c}
\end{pmatrix},
\qquad
Y
=
\begin{pmatrix}
	Y_{\mathcal P}\\
	Y_{\mathcal P^c}
\end{pmatrix},
\qquad
X_M
=
\begin{pmatrix}
	X_{M,\mathcal P}\\
	X_{M,\mathcal P^c}
\end{pmatrix},
\]
where
\[
\mu_{\mathcal P}
=
\begin{pmatrix}
	\mu_1\\
	\vdots\\
	\mu_{N_s}
\end{pmatrix},
\qquad
Y_{\mathcal P}
=
\begin{pmatrix}
	Y_1\\
	\vdots\\
	Y_{N_s}
\end{pmatrix},
\qquad
X_{M,\mathcal P}
=
\begin{pmatrix}
	X_{1,M}\\
	\vdots\\
	X_{N_s,M}
\end{pmatrix},
\]
and
\[
\mu_{\mathcal P^c}
=
\begin{pmatrix}
	\mu_{N_s+1}\\
	\vdots\\
	\mu_N
\end{pmatrix},
\qquad
Y_{\mathcal P^c}
=
\begin{pmatrix}
	Y_{N_s+1}\\
	\vdots\\
	Y_N
\end{pmatrix},
\qquad
X_{M,\mathcal P^c}
=
\begin{pmatrix}
	X_{N_s+1,M}\\
	\vdots\\
	X_{N,M}
\end{pmatrix}.
\]

For square matrices $A_1,\ldots,A_K$, we use
\[
\operatorname{blockdiag}(A_1,\ldots,A_K)
\]
to denote the block-diagonal matrix with $A_1,\ldots,A_K$ as its
diagonal blocks and zero matrices as its off-diagonal blocks.
Accordingly, the covariance matrix can be decomposed as
\[
\Sigma(\theta)
=
\operatorname{blockdiag}
\left(
\Sigma_{\mathcal P}(\theta),
\Sigma_{\mathcal P^c}(\theta)
\right),
\]
where
\[
\begin{aligned}
	\Sigma_{\mathcal P}(\theta)
	&=
	\operatorname{blockdiag}
	\left(
	\Sigma_1(\theta),\ldots,\Sigma_{N_s}(\theta)
	\right),\\
	\Sigma_{\mathcal P^c}(\theta)
	&=
	\operatorname{blockdiag}
	\left(
	\Sigma_{N_s+1}(\theta),\ldots,\Sigma_N(\theta)
	\right).
\end{aligned}
\]

As $\Sigma(\theta)$ is a block diagonal matrix, we also have
\begin{equation}
\label{eq:hat_E_dec}
E(\theta)=\mathrm{blockdiag}(E_{\mathcal P}(\theta),E_{\mathcal P^c}(\theta)),
\end{equation}
where
\[
\begin{aligned}
    E_{\mathcal P}(\theta)&=\Sigma^{-1/2}_{\mathcal P}(\theta_0)\{\Sigma^{}_{\mathcal P}(\theta)-\Sigma^{}_{\mathcal P}(\theta_0)\}\Sigma^{-1/2}_{\mathcal P}(\theta_0),\\
    E_{\mathcal P^c}(\theta)&=\Sigma^{-1/2}_{\mathcal P^c}(\theta_0)\{\Sigma^{}_{\mathcal P^c}(\theta)-\Sigma^{}_{\mathcal P^c}(\theta_0)\}\Sigma^{-1/2}_{\mathcal P^c}(\theta_0).
\end{aligned}
\]
We define $\hat \Sigma_{\mathcal P}=\Sigma_{\mathcal P}(\hat\theta), \hat \Sigma_{\mathcal P^c}=\Sigma_{\mathcal P^c}(\hat\theta),$ $\widehat E_{\mathcal P}=E_{\mathcal P}(\hat \theta)$ and $\widehat E_{\mathcal P^c}=E_{\mathcal P^c}(\hat \theta).$
Then the selection and inference responses are partitioned by:
\[
\widehat Y^{\mathrm{sel}}
=
\begin{pmatrix}
	\widehat Y_{\mathcal P}^{\mathrm{sel}}\\
	\widehat Y_{\mathcal P^c}^{\mathrm{sel}}
\end{pmatrix},
\qquad
\widehat Y^{\mathrm{inf}}
=
\begin{pmatrix}
	\widehat Y_{\mathcal P}^{\mathrm{inf}}\\
	\widehat Y_{\mathcal P^c}^{\mathrm{inf}}
\end{pmatrix},
\]
with
\[
\widehat Y_{\mathcal P}^{\mathrm{sel}}
=
Y_{\mathcal P}+\widehat W_{\mathcal P},
\qquad
\widehat Y_{\mathcal P}^{\mathrm{inf}}
=
Y_{\mathcal P}-\widehat W_{\mathcal P},
\]
and
\[
\widehat Y_{\mathcal P^c}^{\mathrm{sel}}
=
Y_{\mathcal P^c}+\widehat W_{\mathcal P^c},
\qquad
\widehat Y_{\mathcal P^c}^{\mathrm{inf}}
=
Y_{\mathcal P^c}-\widehat W_{\mathcal P^c}.
\]

Let $\mathcal G_{\mathcal P}$ 
denote the sigma-field generated by
$\mathcal D_N$, $\mathcal P$, $Y_{\mathcal P}$, and
$\widehat W_{\mathcal P}$. Conditional on $\mathcal G_{\mathcal P}$, $\widehat\theta$ and $\widehat\Sigma$ are fixed, while the
observations in the non-pilot set remain random.

We write $\mathcal L(U\mid\mathcal G)$ for the conditional law of a random
element $U$ given $\mathcal G$.
We use the following norm notation throughout this section.
Let $\|Z\|_s$ be the spectral norm of a real-valued matrix $Z$ which is defined as
\[\|Z\|_s=\sqrt{\lambda_{\max}(Z^TZ)},\]
where $\lambda_{\max}(\cdot)$ is the maximum eigenvalue of a matrix.
For a vector $z$, define the $L_2$ norm as 
$$
\|z\|_2=\sqrt{z^Tz}.
$$
For a function $z(x)$ define the $L_2$ norm as
$$
\|z(x)\|_2=\left\{\int_{-\infty}^\infty z^2(x)dx\right\}^{1/2}.
$$
If $X\sim P$, the $L_{P,k}$ norm of $z(x)$ is 
$$
\|z(X)\|_{P,k}=\{\mathbb{E}_P(|z(X)|^k)\}^{1/k}.
$$

Let $(S_1,\mathcal G_1,Q_1)$ and $(S_2,\mathcal G_2,Q_2)$ be two
probability spaces. Their product sigma-field is
\[
\mathcal G_1\otimes\mathcal G_2
=
\sigma\!\left(
\left\{
B_1\times B_2:
B_1\in\mathcal G_1,\;
B_2\in\mathcal G_2
\right\}
\right),
\]
that is, the smallest sigma-field on $S_1\times S_2$ containing all
measurable rectangles $B_1\times B_2$.

The product measure $Q_1\otimes Q_2$ is the unique probability measure
on $(S_1\times S_2,\mathcal G_1\otimes\mathcal G_2)$ satisfying
\begin{equation}
\label{eq:product_prob}
    (Q_1\otimes Q_2)(B_1\times B_2)
=
Q_1(B_1)Q_2(B_2),
\qquad
B_1\in\mathcal G_1,\quad B_2\in\mathcal G_2.
\end{equation}
Thus, $Q_1\otimes Q_2$ represents the joint distribution of independent
random elements with marginal distributions $Q_1$ and $Q_2$.
Next, let $T:S_1\to S_2$ be a measurable map, meaning that
\[
T^{-1}(B)
=
\{x\in S_1:T(x)\in B\}
\in\mathcal G_1
\qquad
\text{for every }B\in\mathcal G_2.
\]
The pushforward measure $Q_1$ under $T$, denoted by $Q_1\circ T^{-1}$
is the probability measure on $(S_2,\mathcal G_2)$ defined by
\begin{equation}
    Q_1\circ T^{-1}(B)
=
Q_1\{T^{-1}(B)\}
=
Q_1\{x\in S_1:T(x)\in B\},
\qquad B\in\mathcal G_2.
\label{eq:pushforward}
\end{equation}

\subsection{Empirical-Process Preliminaries}
\label{sec:empirical_process_preliminaries}

Let $O_1,\ldots,O_N$ be independent observations from a probability
measure $P$. For a measurable real-valued function $f$, define
\[
\mathbb P_Nf=\frac1N\sum_{i=1}^Nf(O_i),
\text{ and }
\mathbb G_Nf=\sqrt N(\mathbb P_N-P)f.
\]
For a class $\mathcal F$, write
\[
\|\mathbb G_N\|_{\mathcal F}
=
\sup_{f\in\mathcal F}|\mathbb G_Nf|.
\]
A measurable function $F$ is an envelope of $\mathcal F$ if
$|f|\le F$ for every $f\in\mathcal F$.

For a metric space $(\mathcal T,d)$, let
\[
N(\varepsilon,\mathcal T,d)
\]
denote the covering number, that is, the minimum number of
$d$-balls of radius $\varepsilon$ needed to cover $\mathcal T$.
For the Euclidean ball
\[
B_q(0,r)=\{x\in\mathbb R^q:\|x\|_2\le r\},
\]
the covering number is bounded by \citep[Lemma 2.5]{van2000empirical}
\begin{equation}
\label{eq:ball_bracket}
N\left(
\varepsilon r,
B_q(0,r),
\|\cdot\|_2
\right)
\lesssim
\left(1+\frac{4}{\varepsilon}\right)^q.
\end{equation}

For measurable real-valued functions $l$ and $u$ with $l\le u$
pointwise, define the bracket
\[
[l,u]
=
\{f:l\le f\le u\}.
\]
The bracket is called an $\varepsilon$-bracket in $L_2(P)$ if
\(
\|u-l\|_{P,2}\le \varepsilon.
\)
%where $\|g\|_{P,2}=\{\mathbb E_P(|g|^p)\}^{1/2}$
The bracketing number
\[
N_{[]}(\varepsilon,\mathcal F,L_2(P))
\]
is the minimum number of $\varepsilon$-brackets whose union contains
$\mathcal F$.
Given an envelope $F$, define the normalized bracketing entropy integral
\[
J_{[]}(\delta,\mathcal F,L_2(P))
=
\int_0^\delta
\sqrt{
1+
\log N_{[]}\!\left(
\varepsilon\|F\|_{P,2},
\mathcal F,
L_2(P)
\right)
}
\,d\varepsilon.
\]

We need the following bracketing maximal inequality \citep[Theorem~2.14.2]{van1996weak}.
\begin{corollary}
\label{thm:bracketing_maximal}
Let $\mathcal F$ be a class of real-valued measurable
functions with measurable envelope $F$. If
\[
J_{[]}(1,\mathcal F,L_2(P))
\|F\|_{P,2}
<
\infty,
\]
then
\[
\left\|
\|\mathbb G_N\|_{\mathcal F}^{*}
\right\|_{P,1}
\lesssim
J_{[]}(1,\mathcal F,L_2(P))
\|F\|_{P,2}.
\]
Here $\|\mathbb G_N\|_{\mathcal F}^{*}$ denotes the minimal measurable
majorant of $\|\mathbb G_N\|_{\mathcal F}$.
\end{corollary}

\subsection{Auxiliary Results}
\label{sec:pilot_auxiliary_results}
We provide some auxiliary results required in the main proof. %Lemma \ref{lem:full_to_submodel} - Lemma \ref{lemma:h_bound} establish regularity and perturbation bounds for the working design, the estimated covariance matrix, and the covariance-weighted projection scores. {red}{This is grammatically odd, consider revising.} These results are used to control the stochastic remainder terms arising from the pilot covariance estimation. We then collect several standard Gaussian and matrix results used to convert the resulting mean and covariancebounds into a total-variation approximation.

\begin{lemma}
\label{lem:full_to_submodel}

Let $X_{i,\mathrm{full}}$ denote the full working design matrix, and suppose
that, for some constants $0<c_1\le c_2<\infty$,
\[
c_1 I
\preceq
\mathbb E\left[
X_{i,\mathrm{full}}^\top
\Sigma_i(\theta_0)^{-1}
X_{i,\mathrm{full}}
\right]
\preceq
c_2 I.
\]
Then, for every candidate submodel $M$,
\[
c_1 I_{d_M}
\preceq
\mathbb E\left[
X_{i,M}^\top
\Sigma_i(\theta_0)^{-1}
X_{i,M}
\right]
\preceq
c_2 I_{d_M}.
\]
Moreover, suppose that, for some constants $C_X,C_\mu<\infty$,
\[
\mathbb E\left[
\left\|
\Sigma_i(\theta_0)^{-1/2}
X_{i,\mathrm{full}}
\right\|_s^4
\right]
\le C_X
\]
and
\[
\mathbb E\left[
\left\|
X_{i,\mathrm{full}}^\top
\Sigma_i(\theta_0)^{-1}
\mu_i
\right\|_2^2
\right]
\le C_\mu.
\]
For
\[
r_{Mi}^P
=
\mu_i-X_{i,M}\beta_M^P,
\]
there exists a constant $C_h<\infty$, independent of $M$, such that
\[
\sup_M
\mathbb E\left[
\left\|
X_{i,M}^\top
\Sigma_i(\theta_0)^{-1}
r_{Mi}^P
\right\|_2^2
\right]
\le C_h,
\]
where the supremum is taken over the candidate submodels under consideration.
\end{lemma}

\begin{proof}
For each candidate model $M$, let $R_M$ denote the column-selection
matrix such that
\[
X_{i,M}=X_{i,\mathrm{full}}R_M,
\qquad
R_M^\top R_M=I_{d_M}.
\]
For any $v\in\mathbb R^{d_M}$,
\[
\begin{aligned}
&v^\top
\mathbb E\left[
X_{i,M}^\top
\Sigma_i(\theta_0)^{-1}
X_{i,M}
\right]v=
(R_Mv)^\top
\mathbb E\left[
X_{i,\mathrm{full}}^\top
\Sigma_i(\theta_0)^{-1}
X_{i,\mathrm{full}}
\right]
(R_Mv).
\end{aligned}
\]
Since $\|R_Mv\|_2=\|v\|_2$, the assumed full-design bounds imply
\[
c_1\|v\|_2^2
\le
v^\top
\mathbb E\left[
X_{i,M}^\top
\Sigma_i(\theta_0)^{-1}
X_{i,M}
\right]v
\le
c_2\|v\|_2^2.
\]
This proves the first claim.

We next establish the moment bound. Write
\[
H_M^P
=
\mathbb E\left[
X_{i,M}^\top
\Sigma_i(\theta_0)^{-1}
X_{i,M}
\right].
\]
By the population normal equations,
\[
\beta_M^P
=
(H_M^P)^{-1}
\mathbb E\left[
X_{i,M}^\top
\Sigma_i(\theta_0)^{-1}
\mu_i
\right].
\]
The first part of the lemma gives
\[
\|(H_M^P)^{-1}\|_s\le c_1^{-1}.
\]
Moreover, since $X_{i,M}=X_{i,\mathrm{full}}R_M$,
\[
\left\|
X_{i,M}^\top
\Sigma_i(\theta_0)^{-1}
\mu_i
\right\|_2
\le
\left\|
X_{i,\mathrm{full}}^\top
\Sigma_i(\theta_0)^{-1}
\mu_i
\right\|_2.
\]
Therefore, by Jensen's inequality,
\[
\begin{aligned}
\|\beta_M^P\|_2
&\le
c_1^{-1}
\left\|
\mathbb E\left[
X_{i,M}^\top
\Sigma_i(\theta_0)^{-1}
\mu_i
\right]
\right\|_2\\
&\le
c_1^{-1}
\left\{
\mathbb E\left[
\left\|
X_{i,\mathrm{full}}^\top
\Sigma_i(\theta_0)^{-1}
\mu_i
\right\|_2^2
\right]
\right\}^{1/2}\\
&\le
c_1^{-1}C_\mu^{1/2}.
\end{aligned}
\]

Now,
\[
\begin{aligned}
X_{i,M}^\top
\Sigma_i(\theta_0)^{-1}
r_{Mi}^P
&=
X_{i,M}^\top
\Sigma_i(\theta_0)^{-1}
\mu_i-
X_{i,M}^\top
\Sigma_i(\theta_0)^{-1}
X_{i,M}\beta_M^P.
\end{aligned}
\]
Thus,
\[
\begin{aligned}
\left\|
X_{i,M}^\top
\Sigma_i(\theta_0)^{-1}
r_{Mi}^P
\right\|_2^2
&\le
2
\left\|
X_{i,M}^\top
\Sigma_i(\theta_0)^{-1}
\mu_i
\right\|_2^2+
2
\left\|
X_{i,M}^\top
\Sigma_i(\theta_0)^{-1}
X_{i,M}
\right\|_s^2
\|\beta_M^P\|_2^2.
\end{aligned}
\]
Because $X_{i,M}$ is obtained by selecting columns of
$X_{i,\mathrm{full}}$,
\[
\begin{aligned}
\left\|
X_{i,M}^\top
\Sigma_i(\theta_0)^{-1}
X_{i,M}
\right\|_s
&=
\left\|
\Sigma_i(\theta_0)^{-1/2}X_{i,M}
\right\|_s^2\\
&\le
\left\|
\Sigma_i(\theta_0)^{-1/2}X_{i,\mathrm{full}}
\right\|_s^2.
\end{aligned}
\]
Therefore, taking expectations gives
\[
\begin{aligned}
&\mathbb E\left[
\left\|
X_{i,M}^\top
\Sigma_i(\theta_0)^{-1}
r_{Mi}^P
\right\|_2^2
\right]\le
2C_\mu
+
2C_X c_1^{-2}C_\mu.
\end{aligned}
\]
The right-hand side does not depend on $M$. Therefore, the result follows
with,
\[
C_h
=
2C_\mu
\left(
1+C_Xc_1^{-2}
\right).
\]
\end{proof}

\begin{lemma}
\label{lem:localized_bracketing_bound}
Let $P$ be a probability measure on a measurable space $\mathcal O$, and let
$O_1,\ldots,O_N$ be i.i.d.\ observations from $P$. Let
$O\sim P$ denote a generic observation. Consider a family of measurable
functions
\[
h_\theta:\mathcal O\to\mathbb R^d,
\qquad
\theta\in\Theta_0\subset\mathbb R^q,
\]
where $q$ and $d$ are fixed.

Suppose that, for some $\theta_0\in\Theta_0$,
\[
\mathbb E_P\{h_\theta(O)\}=0
\qquad
\text{for all }\theta\in\Theta_0,
\]
and
\[
h_{\theta_0}(o)=0
\qquad
\text{for every }o\in\mathcal O.
\]
Assume further that there exists a nonnegative measurable function
$Q:\mathcal O\to[0,\infty)$ satisfying
\[
\mathbb E_P\{Q(O)^2\}<\infty
\]
such that, for all $\theta,\theta'\in\Theta_0$ and $o\in\mathcal O$,
\begin{equation}
\label{eq:lip}
    \|h_\theta(o)-h_{\theta'}(o)\|_2
\le
Q(o)\|\theta-\theta'\|_2.
\end{equation}
For
\[
\Theta(\delta)
=
\{\theta\in\Theta_0:\|\theta-\theta_0\|_2\le\delta\},
\]
every deterministic sequence $\delta_N\downarrow0$ satisfies
\[
\sup_{\theta\in\Theta(\delta_N)}
\left\|
\sum_{i=1}^N h_\theta(O_i)
\right\|_2
=
O_p(\sqrt N\,\delta_N).
\]
\end{lemma}

\begin{proof}
Let
\[
\mathbb S^{d-1}=\{v\in\mathbb R^d:\|v\|_2=1\}
\]
and define  
\[
\mathcal F_\delta
=
\{f_{\theta,v}:f_{\theta,v}(O)=v^\top h_\theta(O),\ 
\theta\in\Theta(\delta),\ v\in\mathbb S^{d-1}\}.
\]
Then
\[
\sup_{\theta\in\Theta(\delta)}
\|\mathbb G_Nh_\theta\|_2
=\sup\limits_{\theta\in \Theta(\delta)}\sup\limits_{v\in \mathbb{S}^{d-1}} |v^T\mathbb G_N h_{\theta}|
=\sup\limits_{\theta\in \Theta(\delta),v\in \mathbb{S}^{d-1}}|\mathbb G_N v^Th_{\theta}|
=\sup_{f\in\mathcal F_\delta}|\mathbb G_Nf|.
\]

For every $o\in \mathcal O$, because $h_{\theta_0}(o)=0$, the Lipschitz continuity in 
\eqref{eq:lip} shows that
\[
\|h_\theta(o)\|_2=\|h_\theta(o)-h_{\theta_0}(o)\|_2\le Q(o)\|\theta-\theta_0\|_2\le  Q(o)\delta.
\]
Therefore, we have
$$
\sup\limits_{f\in \mathcal F_\delta} |f(o)|=\sup\limits_{v\in \mathbb S^{d-1},\theta\in \Theta(\delta)} |v^\top h_{\theta}(o)|=\sup\limits_{\theta\in \Theta(\delta)}\|h_{\theta}(o)\|_2\le Q(o)\delta.
$$
Then, \[
F_\delta(o)=\delta Q(o)
\]
is an envelope for $\mathcal F_\delta$, and
$\|F_\delta\|_{P,2}=\delta\|Q\|_{P,2}$.

Choose an $\varepsilon\delta/4$-net of $\Theta(\delta)$ and an
$\varepsilon/4$-net of $\mathbb S^{d-1}$. The Euclidean-ball bound in \eqref{eq:ball_bracket} gives at
most
\[
N_1:=\left \lceil\left(1+\frac{16}{\varepsilon}\right)^q\right\rceil
\quad\text{and}\quad
N_2:=\left \lceil\left(1+\frac{16}{\varepsilon}\right)^d\right\rceil
\]
net points, respectively. 
Let $\{\theta_1,\theta_2,\ldots, \theta_{N_1}\}$, and $\{v_1,v_2,\ldots, v_{N_2}\}$ be the corresponding net-points for $\Theta(\delta)$, and $\mathbb S^{d-1}$.  For each pair of net points
$(\theta_{\tau_1},v_{\tau_2}),\tau_{1}=1,\ldots, N_1,$ and $\tau_{2}=1,\ldots, N_2$, we form the bracket
\[
\left[
f_{\theta_{\tau_1},v_{\tau_2}}-\frac\varepsilon2F_\delta,
\ f_{\theta_{\tau_1},v_{\tau_2}}+\frac\varepsilon2F_\delta
\right].
\]
For every $(\theta,v)$ assigned to $(\theta_{\tau_1},v_{\tau_2})$, with $\|\theta-\theta_{\tau_{1}}\|\le \frac{\epsilon\delta}{4}$, and $\|v-v_{\tau_2}\|\le \frac{\epsilon}{4}$, we have for every $o\in \mathcal O$
\[
\begin{aligned}
|f_{\theta,v}(o)-f_{\theta_{\tau_1},v_{\tau_2}}(o)|
&=|v^Th_{\theta}(o)-v_{\tau_2}^Th_{\theta_{\tau_1}}(o)|
\\&\le |v^T\{h_{\theta}(o)-h_{\theta_{\tau_1}}(o)\}|+|(v-v_{\tau_2})^Th_{\theta_{\tau_1}}(o)|
\\&\le
\|v\|_2\|h_\theta(o)-h_{\theta_{\tau_1}}(o)\|_2
+
\|v-v_{\tau_2}\|_2\|h_{\theta_{\tau_1}}(o)\|_2
\\
&\le
Q(o)\{\|\theta-\theta_{\tau_1}\|_2+
\delta\|v-v_{\tau_2}\|_2\}
\\
&\le\frac{\epsilon}{2}Q(o)\delta\\&=
\frac\varepsilon2F_\delta(o).
\end{aligned}
\]
Hence 
\[
f_{\theta,v}\in \left[
f_{\theta_{\tau_1},v_{\tau_2}}-\frac\varepsilon2F_\delta,
\ f_{\theta_{\tau_1},v_{\tau_2}}+\frac\varepsilon2F_\delta
\right].
\]
Thus these brackets cover $\mathcal F_\delta$, each with $L_2(P)$ width at
most $\varepsilon\|F_\delta\|_{P,2}$, and
\[
N_{[]}\left(
\varepsilon\|F_\delta\|_{P,2},
\mathcal F_\delta,
L_2(P)
\right)
\le
\left(1+\frac{16}\varepsilon\right)^{q+d}.
\]
Consequently,
\[
\begin{aligned}
J_{[]}(1,\mathcal F_\delta,L_2(P))
&\le
\int_0^1
\sqrt{1+(q+d)\log(1+16/\varepsilon)}\,d\varepsilon
\\&<
\int_0^1
2\sqrt{(q+d)\log(1+16/\varepsilon)}\,d\varepsilon
\\&<
\int_0^1
2\sqrt{(q+d)16/\varepsilon}\,d\varepsilon
\\&=8\sqrt{q+d}\int _0^1\epsilon^{-1/2}d\epsilon\\&=16\sqrt{q+d}\\&
<\infty,
\end{aligned}
\]
uniformly in $\delta$. Corollary~\ref{thm:bracketing_maximal} therefore gives
\[
\left\|
\|\mathbb G_N\|_{\mathcal F_\delta}^{*}
\right\|_{P,1}
\lesssim
\delta.
\]
By Markov's inequality,
\[
\sup_{\theta\in\Theta(\delta_N)}
\|\mathbb G_Nh_\theta\|_2
=
O_p(\delta_N).
\]
Since $\mathbb{E}_Ph_\theta=0$, multiplying by $\sqrt N$ proves the result.
\end{proof}

\begin{lemma}
\label{lem:relative_covariance_lipschitz}
There exists a constant $C<\infty$, independent of $N$, such that
\[
\|E(\theta_1)-E(\theta_2)\|_s
\le
C\|\theta_1-\theta_2\|_2,
\qquad
\theta_1,\theta_2\in\Theta_0.
\]
Consequently,
\[
\|\widehat E\|_s
=
O_p(N_s^{-1/2})
=
o_p(1).
\]
\end{lemma}

\begin{proof}
Because $\Sigma(\theta)$ is block diagonal,
\[
\|E(\theta_1)-E(\theta_2)\|_s
=
\max_{1\le i\le N}
\|E_i(\theta_1)-E_i(\theta_2)\|_s,
\]
where
\[
E_i(\theta)
=
\Sigma_{0i}^{-1/2}
\{\Sigma_i(\theta)-\Sigma_{0i}\}
\Sigma_{0i}^{-1/2}.
\]
Recall that the parametrization of $D(\theta)$ and $\sigma^2(\theta)$ is defined in \eqref{eq:var_para}.
Writing
\[
\Delta D=D(\theta_1)-D(\theta_2),
\qquad
\Delta\sigma^2=\sigma^2(\theta_1)-\sigma^2(\theta_2),
\]
we have
\[
\begin{aligned}
&\|E_i(\theta_1)-E_i(\theta_2)\|_s
\\= &\|\Sigma_{0i}^{-1/2}\{(Z_i\Delta DZ_i^T)+\Delta\sigma^2 I_{T_i}\}\Sigma_{0i}^{-1/2}\|_s\\\le &
\|\Sigma_{0i}^{-1/2}Z_i\|_s^2\|\Delta D\|_s
+
|\Delta\sigma^2|\,\|\Sigma_{0i}^{-1}\|_s.
\end{aligned}
\]
Let $B_i=Z_iD_0^{1/2}$. Since
$\Sigma_{0i}=B_iB_i^\top+\sigma_0^2I_{T_i}$, the singular value decomposition of $B_i$ gives
\[
B_i^\top(B_iB_i^\top+\sigma_0^2I)^{-1}B_i\preceq I_2.
\]
Equivalently,
\[
Z_i^\top\Sigma_{0i}^{-1}Z_i\preceq D_0^{-1}.
\]
Also $\Sigma_{0i}\succeq\sigma_0^2I_{T_i}$. Hence
\[
\|\Sigma_{0i}^{-1/2}Z_i\|_s^2\le\|D_0^{-1}\|_s,
\qquad
\|\Sigma_{0i}^{-1}\|_s\le\sigma_0^{-2}.
\]
Finally, because
\[
\|\Delta D\|_s\le\|\Delta D\|_F
\le\|\theta_1-\theta_2\|_2,
\qquad
|\Delta\sigma^2|\le\|\theta_1-\theta_2\|_2,
\]
We have 
\[
\begin{aligned}
&\|E_i(\theta_1)-E_i(\theta_2)\|_s
\\\le &
\|\Sigma_{0i}^{-1/2}Z_i\|_s^2\|\Delta D\|_s
+
|\Delta\sigma^2|\,\|\Sigma_{0i}^{-1}\|_s
\\\le & (\|D_0^{-1}\|_s+\sigma_{0}^{-2}) \|\theta_1-\theta_2\|_2.
\end{aligned}
\]
Since
$E(\theta_0)=0$, we have \(\|\widehat E\|_s\le (\|D_0^{-1}\|_s+\sigma_{0}^{-2})\|\hat \theta-\theta_0\|_2\). 
By assumption 2, \(\|\hat \theta-\theta_0\|_2=O_p(N_s^{-1/2})\). Hence, $\|\widehat E\|_s=O_p(N_s^{-1/2})$.
\end{proof}

\begin{lemma}
\label{lem:empirical_information_order}

For a fixed candidate model $M$, define
\[
H_M
=
X_M^\top\Sigma_0^{-1}X_M,
\qquad
\widehat H_M
=
X_M^\top\widehat\Sigma^{-1}X_M,
\]
and
\[
\widehat H_{M,\mathcal P}
=
X_{M,\mathcal P}^\top
\widehat\Sigma_{\mathcal P}^{-1}
X_{M,\mathcal P}.
\]
Recall that
\[
\widehat V_M=2\widehat H_M^{-1}.
\]

Under Assumptions~1, 2, and 4,
\[
\|H_M\|_s=O_p(N),
\qquad
\|H_M^{-1}\|_s=O_p(N^{-1}),
\]
\[
\|\widehat H_M\|_s=O_p(N),
\qquad
\|\widehat H_M^{-1}\|_s=O_p(N^{-1}),
\]
and
\[
\|\widehat V_M^{-1}\|_s=O_p(N),
\qquad
\|\widehat V_M^{-1/2}\widehat H_M^{-1}\|_s
=
O_p(N^{-1/2}).
\]
For the pilot block,
\[
\|\widehat H_{M,\mathcal P}\|_s=O_p(N_s)\]
and
\[
\begin{aligned}
\left\|
X_{M,\mathcal P}^\top
\Sigma_{0,\mathcal P}^{-1/2}
\right\|_s
&=
O_p(\sqrt{N_s}),
\\
\left\|
\Sigma_{0,\mathcal P}^{-1/2}
(Y_{\mathcal P}-\mu_{\mathcal P})
\right\|_2
&=
O_p(\sqrt{N_s}),
\\
\left\|
X_{M,\mathcal P}^\top
\Sigma_{0,\mathcal P}^{-1}
(Y_{\mathcal P}-\mu_{\mathcal P})
\right\|_2
&=
O_p(\sqrt{N_s}).
\end{aligned}
\]
\end{lemma}
\begin{proof}

By Assumption~4 and Lemma~\ref{lem:full_to_submodel},
\[
\frac{1}{N}H_M
=
\frac{1}{N}
\sum_{i=1}^N
X_{i,M}^\top\Sigma_{0i}^{-1}X_{i,M}
\overset{p}{\to}
\mathbb E
\left[
X_{i,M}^\top\Sigma_{0i}^{-1}X_{i,M}
\right].
\]
Hence
\[
\|H_M\|_s=O_p(N),
\qquad
\|H_M^{-1}\|_s=O_p(N^{-1}).
\]
Note
\[
\widehat H_M=X_{M}^\top \widehat \Sigma^{-1}X_M=X_M^\top\Sigma_0^{-1/2}(I+\widehat E)^{-1}\Sigma_0^{-1/2}X_M,
\]
and on the event $\|\widehat E\|_s<1/2$, with probability tending to $1$, we have
\[
\frac{1}{2}I\preceq(I+\widehat E)^{-1}\preceq 2I.
\]
Therefore, 
\[
\frac{1}{2}H_M
\preceq
\widehat H_M
\preceq
2H_M,
\]
and
\[
\|\widehat H_M\|_s=O_p(N),
\qquad
\|\widehat H_M^{-1}\|_s=O_p(N^{-1}).
\]

For the pilot block, define
\[
H_{M,\mathcal P}
=
X_{M,\mathcal P}^\top
\Sigma_{0,\mathcal P}^{-1}
X_{M,\mathcal P}.
\]
Because $\mathcal P$ is a simple random sample without replacement of
$N_s$ subjects, each subject is included in $\mathcal P$ with probability
$N_s/N$. Conditional on $\mathcal D_N$, we therefore have
\[
\begin{aligned}
\mathbb E\left\{
\operatorname{tr}(H_{M,\mathcal P})
\mid \mathcal D_N
\right\}
=
\frac{N_s}{N}\operatorname{tr}(H_M)\le\frac{N_s}{N}d_M\|H_M\|_s
=
O_p(N_s),
\end{aligned}
\]
where the last equality follows from
$\|H_M\|_s=O_p(N)$ and the fixed dimension $d_M$.
Hence, by Markov's inequality,
\[
\|H_{M,\mathcal P}\|_s
\le
\operatorname{tr}(H_{M,\mathcal P})=O_p(N_s).
\]

Moreover,
\[
\widehat H_{M,\mathcal P}
=
X_{M,\mathcal P}^\top
\Sigma_{0,\mathcal P}^{-1/2}
(I+\widehat E_{\mathcal P})^{-1}
\Sigma_{0,\mathcal P}^{-1/2}
X_{M,\mathcal P},
\]
where
\[
\|\widehat E_{\mathcal P}\|_s
\le
\|\widehat E\|_s.
\]
Thus, on the event $\|\widehat E\|_s<1/2$,
\[
\widehat H_{M,\mathcal P}
\preceq
2
H_{M,\mathcal P},
\]
which gives
\[
\|\widehat H_{M,\mathcal P}\|_s
=
O_p(N_s).
\]
Since
\[
\widehat V_M^{-1}
=
\frac{1}{2}\widehat H_M,
\]
we obtain
\[
\|\widehat V_M^{-1}\|_s
=
O_p(N).
\]
We also have
\[
\widehat V_M^{-1/2}\widehat H_M^{-1}
=
2^{-1/2}\widehat H_M^{-1/2},
\]
so that
\[
\|\widehat V_M^{-1/2}\widehat H_M^{-1}\|_s
=
O_p(N^{-1/2}).
\]
Next,
\[
\begin{aligned}
\left\|
X_{M,\mathcal P}^\top
\Sigma_{0,\mathcal P}^{-1/2}
\right\|_s^2
&=
\left\|
X_{M,\mathcal P}^\top
\Sigma_{0,\mathcal P}^{-1}
X_{M,\mathcal P}
\right\|_s=O_p(N_s).
\end{aligned}
\]
Therefore
\[
\left\|
X_{M,\mathcal P}^\top
\Sigma_{0,\mathcal P}^{-1}
\right\|_s
=
O_p(\sqrt{N_s}).
\]
Conditional on $\mathcal D_N$ and $\mathcal P$,
\[
\Sigma_{0,\mathcal P}^{-1/2}
(Y_{\mathcal P}-\mu_{\mathcal P})\mid \mathcal D_N,\mathcal P
\sim
N(0,I_{n_{\mathcal P}}),
\]
where
\[
n_{\mathcal P}
=
\sum_{i\in\mathcal P}T_i.
\]
Thus,
\[
\begin{aligned}
&\mathbb E
\left[
\left\|
\Sigma_{0,\mathcal P}^{-1/2}
(Y_{\mathcal P}-\mu_{\mathcal P})
\right\|_2^2
\mid
\mathcal D_N,\mathcal P
\right]
=
\operatorname{tr}(I_{n_{\mathcal P}})
=
O_p(N_s),
\end{aligned}
\]
and Markov's inequality gives
\[
\left\|
\Sigma_{0,\mathcal P}^{-1/2}
(Y_{\mathcal P}-\mu_{\mathcal P})
\right\|_2
=
O_p(\sqrt{N_s}).
\]
Finally, conditional on $\mathcal D_N$ and $\mathcal P$,
\[
X_{M,\mathcal P}^\top
\Sigma_{0,\mathcal P}^{-1}
(Y_{\mathcal P}-\mu_{\mathcal P})\mid \mathcal D_N,\mathcal P\sim N(0,X_{M,\mathcal P}^\top
\Sigma_{0,\mathcal P}^{-1}X_{M,\mathcal P}).
\]
Thus,
\[
\begin{aligned}
&\mathbb E
\left[
\left\|
X_{M,\mathcal P}^\top
\Sigma_{0,\mathcal P}^{-1}
(Y_{\mathcal P}-\mu_{\mathcal P})
\right\|_2^2
\mid
\mathcal D_N,\mathcal P
\right]
=
\operatorname{tr}(H_{M,\mathcal P})
=
O_p(N_s),
\end{aligned}
\]
and Markov's inequality gives
\[
\left\|
X_{M,\mathcal P}^\top
\Sigma_{0,\mathcal P}^{-1}
(Y_{\mathcal P}-\mu_{\mathcal P})
\right\|_2
=
O_p(\sqrt{N_s}).
\]
\end{proof}

\begin{lemma}
\label{lem:inv_dif}
    Let $A,B$ be two invertible matrices. Then we have 
    \[
    A^{-1}-B^{-1}=A^{-1}(B-A)B^{-1}=B^{-1}(B-A)A^{-1}.
    \]
\end{lemma}
\begin{proof}
A direct calculation shows
    \[
    \begin{aligned}
    A^{-1}(B-A)B^{-1}&=A^{-1}BB^{-1}-A^{-1}AB^{-1}=A^{-1}-B^{-1},\qquad\text{and}\\
    B^{-1}(B-A)A^{-1}&=BB^{-1}A^{-1}-B^{-1}A^{-1}A=A^{-1}-B^{-1}.
    \end{aligned}
    \]
\end{proof}

\begin{lemma}
\label{lemma:norm_change}
Let
\[
A\in\mathbb R^{m\times d},\qquad
B\in\mathbb R^{d\times d},\qquad
C\in\mathbb R^{d\times q}.
\]
Suppose that
\[
A^\top A=aI_d
\]
for some $a\ge 0$. Then
\[
\|ABC\|_s
\le
\|B\|_s\|AC\|_s.
\]
\end{lemma}

\begin{proof}
Since $A^\top A=aI_d$,
\[
\begin{aligned}
\|ABC\|_s^2
&=\|C^\top B^\top A^\top ABC\|_s
\\
&=
a\|C^\top B^\top BC\|_s\\
&\le
a\|B\|_s^2\|C\|^2_s.
\end{aligned}
\]
We also have
\[
\|AC\|_s^2=\|C^\top A^\top AC\|_s=a\|C^\top C\|_s=a\|C\|_s^2.
\]
Therefore, 
\[
\|ABC\|_s^2\le a\|B\|_s^2\|C\|^2_s=\|B\|_s^2\|AC\|_s^2.
\]

\end{proof}

\begin{lemma}
\label{lemma:h_bound}
    Suppose that $M_{\mathrm I}=M_{\mathrm S}=J$ and 
$$
\Sigma_i(\theta) = Z_iD(\theta)Z_i^\top+\sigma^2(\theta)I_{T_i}, \qquad \theta\in\Theta_0
$$
where $\Theta_0$ is a sufficiently small fixed neighborhood of \(\theta_0\) contained in the covariance-parameter space. For 
$$
h_i(\theta) = X_{i,M}^\top \left\{ \Sigma_i^{-1}(\theta)-\Sigma_{0i}^{-1} \right\} r_{Mi}^{P},
$$
there exists a constant $L_0<\infty$, independent of $i$ and $M$, such that, for all $\theta_1,\theta_2\in\Theta_0$, 
$$
\left\| h_i(\theta_1)-h_i(\theta_2) \right\|_2 \le L_0 \left\| X_{i,M}^\top \Sigma_{0i}^{-1} r_{Mi}^{P} \right\|_2\left\| \theta_1-\theta_2 \right\|_2.
$$
\end{lemma}
\begin{proof}
    Here we take $\Theta_0$ sufficiently small such that, $\forall \theta\in  \Theta_0$, we have
    \[
    \frac{1}{2}D(\theta_0)\preceq D(\theta)\preceq 2 D(\theta_0), \frac{1}{2}\sigma^2(\theta_0)\le \sigma^2(\theta)\le 2\sigma^2(\theta_0).
    \]
    Let
    \begin{equation}
        \label{eq:L_i}
        L_i(\theta)=D(\theta)Z_i^TZ_i+\sigma^2(\theta)I_2
    \end{equation}
    Then we have
    \[
    \Sigma_i(\theta)Z_i=Z_iL_i(\theta), 
    \]
    therefore,
    \[
    \Sigma_i^{-1}(\theta)Z_i=Z_iL_i^{-1}(\theta).
    \]
    Besides, we also have
    \[
    \begin{aligned}
    \Sigma_i^{-1}(\theta)Z_i&=Z_iL_i^{-1}(\theta)
    \\&=\Sigma_{i}^{-1}(\theta_0)\Sigma_{i}(\theta_0)Z_i L^{-1}_i(\theta)
    \\&=\Sigma_i^{-1}(\theta_0)Z_iL_i(\theta_0)L^{-1}_i(\theta).
    \end{aligned}
    \]
    Hence, we have
    \[
    Z_i^\top \Sigma_i^{-1}(\theta)Z_i=Z_i^T\Sigma_i^{-1}(\theta_0)Z_iL_i(\theta_0)L^{-1}_i(\theta).
    \]
    Since the left-hand side is symmetric, we have
    \[
    Z_i^\top \Sigma_i^{-1}(\theta)Z_i=\{L_i(\theta_0)L^{-1}_i(\theta)\}^\top Z_i^T\Sigma_i^{-1}(\theta_0)Z_i.
    \]

    Recall that under \eqref{eq:C_M}, we have\[
C_M(X_i)
=
\begin{pmatrix}
1 & X_{i,M_I}^{\top}  & 0& 0^\top\\
0 & 0^\top & 1 & X_{i,M_S}^{\top}
\end{pmatrix},\]
    and
   $X_{iM}=Z_iC_{iM}$. As we have $r_{Mi}^P=Z_i u_{Mi}$, with
    $u_{Mi}=m(X_i)-C_{iM}\beta^{P}_M.$ Therefore, we find that
    \begin{align}
        h_i(\theta_1)-h_i(\theta_2)&=X_{i,M}^\top \{\Sigma_i^{-1}(\theta_1)-\Sigma_i^{-1}(\theta_2)\}r_{Mi}^P\nonumber
        \\&=C_{iM}^TZ_i^T\Sigma_i^{-1}(\theta_1)Z_iu_{Mi}-C_{iM}^TZ_i^T\Sigma_i^{-1}(\theta_2)Z_iu_{Mi}\nonumber
        \\&=C_{iM}^\top \{L_i(\theta_0)L^{-1}_i(\theta_1)\}^\top Z_i^\top\Sigma_i^{-1}(\theta_0)Z_iu_{M_i}- C_{iM}^\top \{L_i(\theta_0)L^{-1}_i(\theta_2)\}^\top Z_i^\top \Sigma_i^{-1}(\theta_0)Z_iu_{M_i} \nonumber
        \\&=C_{iM}^\top \{\{L_i(\theta_0)L^{-1}_i(\theta_1)\}^\top-\{L_i(\theta_0)L^{-1}_i(\theta_2)\}^\top\}Z_i^T\Sigma_i^{-1}(\theta_0)Z_iu_{M_i}.\nonumber        
    \end{align}
    Under $M_I=M_S$, $C_{i,M}C_{i,M}^\top= (1+\|x_{iM}\|_2^2)I$, by Lemma \ref{lemma:norm_change} and \eqref{eq:delta_h}, we have
    \begin{align}
    \|h_i(\theta_1)-h_i(\theta_2)\|_s&\le \|\{L_i(\theta_0)L^{-1}_i(\theta_1)\}^\top-\{L_i(\theta_0)L^{-1}_i(\theta_2)\}^\top\|_s\|C_{iM}^\top Z_i^T\Sigma_i^{-1}(\theta_0)Z_iu_{M_i}\|_2 \nonumber
    \\&=\|L_i(\theta_0)L^{-1}_i(\theta_1)-L_i(\theta_0)L^{-1}_i(\theta_2)\|_s\left\| X_{i,M}^\top \Sigma_{0i}^{-1} r_{Mi}^{P} \right\|_2.
    \label{eq:delta_h}
    \end{align}

    Let $G_i(\theta)=L_i(\theta)L_i^{-1}(\theta_0)$, so $G_i^{-1}(\theta)=L_i(\theta_0)L_i^{-1}(\theta)$. Now we only need to bound $\|G_i^{-1}(\theta_1)-G_i^{-1}(\theta_2)\|_s$.
    
    Note that $L_i(\theta)=D(\theta)Z_i^TZ_i+\sigma^2(\theta)I_2$ in \eqref{eq:L_i}, the difference between $G_i(\theta_1)$ and $G_i(\theta_2)$ is
	\begin{align}
		\|G_i(\theta_1)-G_i(\theta_2)\|_s
		&=
		\|\left[
		\{D(\theta_1)-D(\theta_2)\}Z_i^TZ_i
		+
		\{\sigma^2(\theta_1)-\sigma^2(\theta_2)\}I_2
		\right]
		L_i^{-1}(\theta_0)\|_s \nonumber
        \\&\le \|D(\theta_1)-D(\theta_2)\|_s\|Z_i^TZ_iL_i^{-1}(\theta_0)\|_s+|\sigma^2(\theta_1)-\sigma^2(\theta_2)|\|L_i^{-1}(\theta_0)\|_s.
        \label{eq:delta_G1}
	\end{align}
	Now
	\[
    	L_i(\theta_0)
    	=
    	D_0Z_i^TZ_i+\sigma_0^2I_2
    	=
    	D_0^{1/2}
    	\left(
    	D_0^{1/2}Z_i^TZ_iD_0^{1/2}
    	+\sigma_0^2I_2
    	\right)
    	D_0^{-1/2}.
	\]
	Since $Z_i^\top Z_i\succeq0$,
	\begin{equation}
    \label{eq:L_inv_norm}
	\|L_i(\theta_0)^{-1}\|_s
	\le
	\frac{
		\|D_0^{1/2}\|_s
		\|D_0^{-1/2}\|_s
	}{
		\sigma_0^2
	}
	<\infty.
	\end{equation}
    Besides, as $L(\theta_0)=D_0Z_i^TZ_i+\sigma^2_0I_2$, we have 
    \[
    Z_i^TZ_iL_i^{-1}(\theta_0)=D_0^{-1}(I-\sigma^2_0 L_i^{-1}(\theta_0)),
    \]
    and
    \begin{equation}
    \label{eq:zzl_inv}
    \|D_0^{-1}(I-\sigma^2_0 L_i^{-1}(\theta_0))\|_s\le  \|D_0^{-1}\|_s(1+\sigma_0^2\|L_i^{-1}(\theta_0)\|_s)\le \|D_0^{-1}\|_s(1+\|D_0^{1/2}\|_s\|D_0^{-1/2}\|_s)<\infty.
    \end{equation}
    Therefore, as 
    \[\max\{\|D(\theta_1)-D(\theta_2)\|_s,|\sigma^2(\theta_1)-\sigma^2(\theta_2)|\}\le \|\theta_1-\theta_2\|_2, \]
     taking $C_0=\frac{
		\|D_0^{1/2}\|_s
		\|D_0^{-1/2}\|_s
	}{
		\sigma_0^2
	}+\|D_0^{-1}\|_s(1+\|D_0^{1/2}\|_s\|D_0^{-1/2}\|_s)$,  by \eqref{eq:delta_h},\eqref{eq:L_inv_norm}, and \eqref{eq:zzl_inv}, we have
    \begin{equation}
    \label{eq:delta_G2}
         \|G_i(\theta_1)-G_i(\theta_2)\|_s\le C_0\|\theta_1-\theta_2\|_2.   
    \end{equation}
Choose $\Theta_0$ sufficiently small that
    \[
    \sup_{\theta\in\Theta_0}\|\theta-\theta_0\|_2
    \le \frac{1}{2C_0}.
    \]
    Since $G_i(\theta_0)=I_2$, we have
    \[
    \sup_{\theta\in\Theta_0}
    \|G_i(\theta)-I_2\|_s
    \le \frac12.
    \]
    For every $v\in\mathbb R^2$,
    \[
    \|G_i(\theta)v\|_2
    \ge
    \|v\|_2-\|(G_i(\theta)-I_2)v\|_2
    \ge \frac12\|v\|_2.
    \]
    Consequently,
    \[
    \sup_{\theta\in\Theta_0}
    \|G_i(\theta)^{-1}\|_s\le2.
    \]
    Therefore, combining with \eqref{eq:delta_G2}, we obtain that
    \begin{align}
    &\|G_i^{-1}(\theta_1)-G_i^{-1}(\theta_2)\|_s \nonumber
    \\=&\|G_i^{-1}(\theta_1)\{G_i(\theta_2)-G_i(\theta_1)\}G_i^{-1}(\theta_2)\|_s \nonumber
    \\\le &\|G_i^{-1}(\theta_1)\|_s\|G_i^{-1}(\theta_2)\|_s\|\{G_i(\theta_2)-G_i(\theta_1)\}\|_s \nonumber
    \\\le& 4 C_0\|\theta_1-\theta_2\|_2.
    \label{eq:G_inv_dif}
    \end{align}

Finally, combining\eqref{eq:G_inv_dif} and \eqref{eq:delta_h}, we have
$$
\left\| h_i(\theta_1)-h_i(\theta_2) \right\|_2 \le 4C_0 \left\| X_{i,M}^\top \Sigma_{0i}^{-1} r_{Mi}^{P} \right\|_2\left\| \theta_1-\theta_2 \right\|_2.
$$
\end{proof}

\begin{remark}
If we consider a stronger assumption in assumption 4, and require that there exists a constant $C_h<\infty$ such that
	\[
	\sup_M
	\mathbb E\left[
	\left\|
	X_{i,M}^{\top}
	\Sigma_i(\theta_0)^{-1/2}\right\|_s^2\left\|\Sigma_i(\theta_0)^{-1/2}
	r_{Mi}^{P}
	\right\|_2^2
	\right]
	\le C_h,\]
    we only need to prove a weaker result here, with
$$
\left\| h_i(\theta_1)-h_i(\theta_2) \right\|_2 \le L_0 \left\| X_{i,M}^\top \Sigma_{0i}^{-1/2}\right\|_s\left\|\Sigma_{0i}^{-1/2} r_{Mi}^{P} \right\|_2\left\| \theta_1-\theta_2 \right\|_2.
$$
The proof can be simplified. Because we have
\[
\begin{aligned}
    \left\| h_i(\theta_1)-h_i(\theta_2) \right\|_2&\le \|X_{i,M}^\top \{\Sigma_i^{-1}(\theta_1)-\Sigma_i^{-1}(\theta_2)\}r_{Mi}^P\|_s
    \\&\le \|X_{i,M}^\top \Sigma_{0}^{-1/2}\|_s\|\Sigma_0^{1/2}\{\Sigma_i^{-1}(\theta_1)-\Sigma_i^{-1}(\theta_2)\}\Sigma_0^{1/2}\|_s\|\Sigma_0^{-1/2} r_{Mi}^P\|_s.
\end{aligned}
\]
We take $\Theta$ sufficiently small, such that $\|E_i(\theta)\|_s\le 1/2$, and $\|(I+E_i(\theta))^{-1}\|_s\le 2$. Therefore,
\[
\begin{aligned}
&\|\Sigma_0^{1/2}\{\Sigma_i^{-1}(\theta_1)-\Sigma_i^{-1}(\theta_2)\}\Sigma_0^{1/2}\|_s
\\=&\|(I+E_i(\theta_1))^{-1}-(I+E_i(\theta_2))^{-1}\|_s
\\=&\|(I+E_i(\theta_1))^{-1}\{I+E_i(\theta_2)-I-E_i(\theta_1)\}(I+E_i(\theta_2))^{-1}\|_s
\\\le &\|(I+E_i(\theta_1))^{-1}\|_s\|E_i(\theta_2)-E_i(\theta_1)\|_s\|(I+E_i(\theta_2))^{-1}\|_s
\\\lesssim&\|\theta_1-\theta_2\|_2,
\end{aligned}
\]
where the second equality is based on Lemma \ref{lem:inv_dif} and the last inequality is based on Lemma \ref{lem:relative_covariance_lipschitz}.
Therefore, there exists $L_0$ such that
\[
\left\| h_i(\theta_1)-h_i(\theta_2) \right\|_2 \le L_0 \left\| X_{i,M}^\top \Sigma_{0i}^{-1/2}\right\|_s\left\|\Sigma_{0i}^{-1/2} r_{Mi}^{P} \right\|_2\left\| \theta_1-\theta_2 \right\|_2.
\]
\end{remark}

\begin{lemma}[Pinsker's inequality]
\label{lemma:Pinsker_inequality}
Let $P$ and $Q$ be probability measures with
$P\ll Q$. Let $d_{TV}(P,Q)$ be the total variation between $P$ and $Q$, and $d_{KL}(P,Q)$ is the KL divergence between $P$ and $Q$.
Then
\[
d_{\mathrm{TV}}(P,Q)
\le
\sqrt{\frac{1}{2}d_{\mathrm{KL}}(P,Q)}.
\]
\end{lemma}

\begin{proof}
This result can be found in \citet[Lemma~2.5]{tsybakov2009introduction}.
\end{proof}

\begin{lemma}[KL divergence between multivariate normal distributions]
\label{lemma:KL_mult}
Let
\[
P=N_d(\mu_1,\Sigma_1),
\qquad
Q=N_d(\mu_0,\Sigma_0),
\]
where $\Sigma_0$ and $\Sigma_1$ are positive definite. Then
\[
\begin{aligned}
d_{KL}(P,Q)
=
\frac{1}{2}
\Big[
&\operatorname{tr}(\Sigma_0^{-1}\Sigma_1)
+
(\mu_0-\mu_1)^\top
\Sigma_0^{-1}
(\mu_0-\mu_1)-d
+
\log\frac{\det(\Sigma_0)}{\det(\Sigma_1)}
\Big].
\end{aligned}
\]
In particular, if $Q=N_d(0,I_d)$, then
\[
d_{\mathrm{KL}}
\left\{
N_d(\mu,\Sigma)\,,N_d(0,I_d)
\right\}
=
\frac{1}{2}
\left\{
\operatorname{tr}(\Sigma)-d-\log\det(\Sigma)
+\mu^\top\mu
\right\}.
\]
\end{lemma}

\begin{proof}
This result can be found in \citet[Theorem 1.8.2]{ihara1993information}.
\end{proof}

\begin{lemma}[Schur complement determinant formula]
\label{lemma:schur}
Let
\[
M=
\begin{pmatrix}
A & B\\
C & D
\end{pmatrix},
\]
where $D$ is invertible. Then
\[
\det(M)
=
\det(D)
\det\left(A-BD^{-1}C\right).
\]
In particular,
\[
\det
\begin{pmatrix}
A & B\\
B^\top & I
\end{pmatrix}
=
\det\left(A-BB^\top\right).
\]
\end{lemma}

\begin{proof}
This result can be found in \citet[Theorem 1.1]{zhang2006schur}.
\end{proof}

\begin{lemma}
\label{lemma:KL_cal}
Let $A\in\mathbb R^{d\times d}$ be symmetric. Then
\[
|\operatorname{tr}(A)|\le d\|A\|_s.
\]
If $\|A\|_s\le 1/2$, then $I_d+A$ is positive definite and
\[
|\log\det(I_d+A)|\le 2d\|A\|_s.
\]
\end{lemma}

\begin{proof}
Let $\lambda_1,\ldots,\lambda_d$ be the eigenvalues of $A$.
The trace bound follows from $|\lambda_j|\le\|A\|_s$.
If $\|A\|_s\le1/2$, then
$|\log(1+\lambda_j)|\le2|\lambda_j|$, and hence
\[
|\log\det(I_d+A)|
\le
\sum_{j=1}^d|\log(1+\lambda_j)|
\le
2d\|A\|_s.
\]
\end{proof}

% \begin{lemma}[Conditional contraction of total variation]
% \label{lem:conditional_tv_contraction}

% Let $(Z,S)$ be
% a random element taking values in
% $\Omega_1\times\Omega_2$, with $\mathcal G$ a $\sigma$-algebra.
% Let $(\Omega_1,\mathcal G_1,Q_1)$ and $(\Omega_2,\mathcal G_2,Q_2)$ be two probability space. For $j=1,2$, let $T_j:\Omega_{j}\to\Omega_{j}'$ be a measurable function. Let $Z=(Z_1,Z_2)$ with $Z_1\in \Omega_1, Z_2\in \Omega_1$. Let $Q_1\circ T_1^{-1}$ and $Q_2\circ T_1^{-2}$
% be the pushforward measure defined in \eqref{eq:pushforward}. Therefore,we have  
% \[
% \begin{aligned}
% d_{\mathrm{TV}}
% \left[
% \mathcal L\{T_1(Z),T_2(S)\mid\mathcal G\},
% (Q_Z\circ T_1^{-1})\otimes(Q_S\circ T_2^{-1})
% \right]
% \le
% d_{\mathrm{TV}}
% \left[
% \mathcal L(Z,S\mid\mathcal G),
% Q_Z\otimes Q_S
% \right],
% \end{aligned}
% \]
% where $\otimes$ is the product measure defined in \eqref{eq:product_prob}.
% .
% \end{lemma}
\begin{lemma}
\label{lem:conditional_tv_contraction}
Let $(\Omega_1,\mathcal G_1)$ and $(\Omega_2,\mathcal G_2)$ be
measurable spaces, and let $(Z_1,Z_2)$ be a random element taking
values in $\Omega_1\times\Omega_2$. Let $\mathcal G$ be a sigma-algebra
on the underlying probability space.

For $j=1,2$, let $Q_j$ be a reference probability distribution on
$(\Omega_j,\mathcal G_j)$, and let
\[
T_j:(\Omega_j,\mathcal G_j)\to(\Omega_j',\mathcal G_j')
\]
be a measurable map. The distributions $Q_1,Q_2$ and the maps
$T_1,T_2$ are
fixed conditional on $\mathcal G$.
Write $Q_j\circ T_j^{-1},j=1,2$ for the pushforward measure defined in
\eqref{eq:pushforward}, and use $\otimes$ for the product measure
defined in \eqref{eq:product_prob}. Assuming that the conditional
distributions below exist, we have, 
\[
\begin{aligned}
&d_{\mathrm{TV}}
\left[
\mathcal L\{T_1(Z_1),T_2(Z_2)\mid\mathcal G\},
(Q_1\circ T_1^{-1})\otimes(Q_2\circ T_2^{-1})
\right]
\le
d_{\mathrm{TV}}
\left[
\mathcal L(Z_1,Z_2\mid\mathcal G),
Q_1\otimes Q_2
\right].
\end{aligned}
\]
\end{lemma}

\begin{proof}
Conditional on $\mathcal G$, define
$T(z_1,z_2)=\{T_1(z_1),T_2(z_2)\}$. For any measurable set $B$ in the range of
$T$,
\[
\begin{aligned}
&\sup\limits_{B}\left|
\mathbb P\{T(Z_1,Z_2)\in B\mid\mathcal G\}
-(Q_1\otimes Q_2)\{T^{-1}(B)\}
\right|
\\=& \sup\limits_{B}|\mathbb P\{(Z_1,Z_2)\in T^{-1}(B)\mid\mathcal G\}
-(Q_1\otimes Q_2)\{T^{-1}(B)\}|
\\
\le&
d_{\mathrm{TV}}
\left[
\mathcal L(Z_1,Z_2\mid\mathcal G),Q_1\otimes Q_2
\right].
\end{aligned}
\]

Therefore, we obtain
\[
\begin{aligned}
&d_{\mathrm{TV}}
\left[
\mathcal L\{T_1(Z_1),T_2(Z_2)\mid\mathcal G\},
(Q_1\circ T_1^{-1})\otimes(Q_2\circ T_2^{-1})
\right]
\le
d_{\mathrm{TV}}
\left[
\mathcal L(Z_1,Z_2\mid\mathcal G),
Q_1\otimes Q_2
\right],
\end{aligned}
\]
\end{proof}
\begin{lemma}
    Let $A,B$ be matrices, then $\|A^\top\|_s=\|A\|_s$,
    $\|AB\|_s\le \|A\|_s\|B\|_s$, and $\|A+B\|_s\le \|A\|_s+\|B\|_s$.
    Furthermore, we have $\|\mathrm{blockdiag}(A,B)\|_s=\max(\|A\|_s,\|B\|_s)$
\end{lemma}
\begin{proof}
By the definition of the spectral norm, the singular value decomposition of $A$ gives
\[
\|A\|_s=\sqrt{\lambda_{\max}(A^\top A)}=\sqrt{\lambda_{\max}(A A^\top)}
=
\|A^\top \|_s.
\]
For conformable matrices $A$ and $B$,
\[
\|AB\|_s
=
\sup_{\|x\|_2=1}\|ABx\|_2
\le \sup\limits_{\|x\|_2=1}\{\|A\frac{Bx}{\|Bx\|_2}\|_2\|Bx\|_2\}\le \sup\limits_{\|x\|_2=1}\sup\limits_{\|y\|_2=1}\|Ay\|_2\|Bx\|_2
=
\|A\|_s\|B\|_s.
\]
Similarly, by the triangle inequality,
\[
\|A+B\|_s
=
\sup_{\|x\|_2=1}\|(A+B)x\|_2
\le \sup_{\|x\|_2=1}\|Ax\|_2+\sup_{\|x\|_2=1}\|Bx\|_2=
\|A\|_s+\|B\|_s.
\]
Finally, let
\[
C=\operatorname{blockdiag}(A,B).
\]
For any conformable vector $x=(x_1^\top,x_2^\top)^\top$,
\[
\|Cx\|_2^2
=
\|Ax_1\|_2^2+\|Bx_2\|_2^2
\le
\max\{\|A\|_s^2,\|B\|_s^2\}
\left(\|x_1\|_2^2+\|x_2\|_2^2\right).
\]
Hence,
\[
\|C\|_s
\le
\max\{\|A\|_s,\|B\|_s\}.
\]
The reverse inequality follows by taking either $x_2=0$ or $x_1=0$.
Therefore,
\[
\|\operatorname{blockdiag}(A,B)\|_s
=
\max\{\|A\|_s,\|B\|_s\}.
\]
\end{proof}

\subsection{Proof of Theorem \ref{thm:beta_est}}
%We first give proof of Theorem 
\begin{proof}
	Let
	\[
	C_i=C_M(X_i)
	=
	\begin{pmatrix}
		1 & X_{i,M_I}^\top & 0 & 0^\top\\
		0 & 0^\top & 1 & X_{i,M_S}^\top
	\end{pmatrix},
	\]
	and write
	\[
	m_i=m(X_i)
	=
	\begin{pmatrix}
		m_1(X_i)\\
		m_2(X_i)
	\end{pmatrix}.
	\]
	Recall that
	\[
	A_i=Z_i^\top\Sigma_i^{-1}Z_i.
	\]
	By the definition of the population projection,
	\[
	\beta_M^P
	=
	\left\{
	E(C_i^\top A_i C_i)
	\right\}^{-1}
	E(C_i^\top A_i m_i).
	\]
	Since $C_i$ and $m_i$ are functions of $X_i$, and
	\[
	E(A_i\mid X_i)
	=
	A_0
	=
	\begin{pmatrix}
		a&b\\
		b&c
	\end{pmatrix},
	\]
    therefore,
	\[
	\beta_M^P
	=
	\left\{
	E(C_i^\top A_0 C_i)
	\right\}^{-1}
	E(C_i^\top A_0m_i).
	\]
	
	We first calculate the matrix in the left-hand side. Direct multiplication gives
	\[
	C_i^\top A_0 C_i
	=
	\begin{pmatrix}
		a
		&
		aX_{i,M_I}^\top
		&
		b
		&
		bX_{i,M_S}^\top
		\\
		aX_{i,M_I}
		&
		aX_{i,M_I}X_{i,M_I}^\top
		&
		bX_{i,M_I}
		&
		bX_{i,M_I}X_{i,M_S}^\top
		\\
		b
		&
		bX_{i,M_I}^\top
		&
		c
		&
		cX_{i,M_S}^\top
		\\
		bX_{i,M_S}
		&
		bX_{i,M_S}X_{i,M_I}^\top
		&
		cX_{i,M_S}
		&
		cX_{i,M_S}X_{i,M_S}^\top
	\end{pmatrix}.
	\]
	Because $E(X_i)=0$,
	\[
	E(C_i^\top A_0 C_i)
	=
	\begin{pmatrix}
		a & 0^\top & b & 0^\top\\
		0 & a\Gamma_{II} & 0 & b\Gamma_{IS}\\
		b & 0^\top & c & 0^\top\\
		0 & b\Gamma_{SI} & 0 & c\Gamma_{SS}
	\end{pmatrix}.
	\]
	
	Next,
	\[
	A_0m_i
	=
	\begin{pmatrix}
		a m_1(X_i)+b m_2(X_i)\\
		b m_1(X_i)+c m_2(X_i)
	\end{pmatrix},
	\]
	and hence
	\[
	C_i^\top A_0m_i
	=
	\begin{pmatrix}
		a m_1(X_i)+b m_2(X_i)\\
		X_{i,M_I}
		\{a m_1(X_i)+b m_2(X_i)\}\\
		b m_1(X_i)+c m_2(X_i)\\
		X_{i,M_S}
		\{b m_1(X_i)+c m_2(X_i)\}
	\end{pmatrix}.
	\]
	Therefore,
	\[
	E(C_i^\top A_0m_i)
	=
	\begin{pmatrix}
		aE\{m_1(X)\}+bE\{m_2(X)\}\\
		a\Cov\{X_{M_I},m_1(X)\}
		+b\Cov\{X_{M_I},m_2(X)\}\\
		bE\{m_1(X)\}+cE\{m_2(X)\}\\
		b\Cov\{X_{M_S},m_1(X)\}
		+c\Cov\{X_{M_S},m_2(X)\}
	\end{pmatrix},
	\]
	where we have again used $E(X_i)=0$.
	
	Writing
	\[
	\beta_M^P
	=
	\begin{pmatrix}
		\alpha_I^P\\
		\beta_I^P\\
		\alpha_S^P\\
		\beta_S^P
	\end{pmatrix},
	\]
	the equation
	\[
	E(C_i^\top A_0C_i)\beta_M^P
	=
	E(C_i^\top A_0m_i)
	\]
	implies, for the two intercept coordinates,
	\[
	\begin{pmatrix}
		a&b\\
		b&c
	\end{pmatrix}
	\begin{pmatrix}
		\alpha_I^P\\
		\alpha_S^P
	\end{pmatrix}
	=
	\begin{pmatrix}
		a&b\\
		b&c
	\end{pmatrix}
	\begin{pmatrix}
		E\{m_1(X)\}\\
		E\{m_2(X)\}
	\end{pmatrix}.
	\]
	Since $A_0$ is positive definite,
	\[
	\alpha_I^P=E\{m_1(X)\},
	\qquad
	\alpha_S^P=E\{m_2(X)\}.
	\]
	
	The remaining coordinates satisfy
	\[
	\begin{pmatrix}
		a\Gamma_{II} & b\Gamma_{IS}\\
		b\Gamma_{SI} & c\Gamma_{SS}
	\end{pmatrix}
	\begin{pmatrix}
		\beta_I^P\\
		\beta_S^P
	\end{pmatrix}
	=
	\begin{pmatrix}
		a\Cov\{X_{M_I},m_1(X)\}
		+b\Cov\{X_{M_I},m_2(X)\}\\
		b\Cov\{X_{M_S},m_1(X)\}
		+c\Cov\{X_{M_S},m_2(X)\}
	\end{pmatrix},
	\]
	which gives the stated expression.
	
	Finally, suppose that
	\[
	M_I=M_S=J.
	\]
	Then
	\[
	\Gamma_{II}
	=
	\Gamma_{IS}
	=
	\Gamma_{SI}
	=
	\Gamma_{SS}
	=
	\Gamma_{JJ}.
	\]
	Let
	\[
	q_1=\Cov\{X_J,m_1(X)\},
	\qquad
	q_2=\Cov\{X_J,m_2(X)\}.
	\]
	The coefficient equations reduce to
	\[
	a\Gamma_{JJ}\beta_I^P
	+
	b\Gamma_{JJ}\beta_S^P
	=
	aq_1+bq_2,
	\]
	and
	\[
	b\Gamma_{JJ}\beta_I^P
	+
	c\Gamma_{JJ}\beta_S^P
	=
	bq_1+cq_2.
	\]
    Equivalently,
    \[
    (A_0 \otimes I_{|J|}) \begin{pmatrix}
        \Gamma_{JJ}\beta_I^P-q_1\\ \Gamma_{JJ}\beta_S^P-q_2
    \end{pmatrix}=0
    \]
    where $\otimes$ is the Kronecker product, and
    \[
    A_0 \otimes I_{|J|}=\begin{pmatrix}aI_{|J|}&bI_{|J|}\\b I_{|J|}&cI_{|J|}\end{pmatrix}
    \]
    
	Since $A_0$ is invertible, $A_0 \otimes I_{|J|}$ is also invertible. Therefore,
	\[
	\Gamma_{JJ}\beta_I^P=q_1,
	\qquad
	\Gamma_{JJ}\beta_S^P=q_2.
	\]
	Therefore,
	\[
	\beta_I^P
	=
	\Gamma_{JJ}^{-1}
	\Cov\{X_J,m_1(X)\},
	\qquad
	\beta_S^P
	=
	\Gamma_{JJ}^{-1}
	\Cov\{X_J,m_2(X)\}.
	\]
	This completes the proof.
\end{proof}

\subsection{Proof of Theorem~\ref{thm:pilot_assisted_validity}}
\label{sec:proof_pilot_assisted_validity}

\begin{proof}
Recall that 
\[
A_M
=
(X_M^\top\Sigma_0^{-1}X_M)^{-1}X_M^\top\Sigma_0^{-1},
\qquad
\widehat A_M
=
(X_M^\top\widehat\Sigma^{-1}X_M)^{-1}X_M^\top\widehat\Sigma^{-1}.
\]

We partition the columns of $\widehat A_M$ according to the pilot and non-pilot
observations, and denote
\[
\widehat A_M=(\widehat A_{M,\mathcal P},\widehat A_{M,\mathcal P^c}).
\]
Specifically, we have
\begin{equation}
\label{eq:A_P}
\widehat A_{M,\mathcal P}= (X_M^\top\widehat \Sigma^{-1}X_M)^{-1}X_{M,\mathcal P}^\top\widehat \Sigma^{-1}_{\mathcal P}\text{ and }\widehat A_{M,\mathcal P^c}= (X_M^\top\widehat \Sigma^{-1}X_M)^{-1}X_{M,\mathcal P^c}^\top\widehat \Sigma^{-1}_{\mathcal P^c}.
\end{equation}
Then
\begin{equation}
\label{eq:main_decomposition}
\begin{aligned}
\widehat V_M^{-1/2}(\widehat\beta_M-\beta_{M,N}^{D})
={}&
\widehat V_M^{-1/2}\widehat A_{M,\mathcal P^c}
(\widehat Y_{\mathcal P^c}^{\mathrm{inf}}-\mu_{\mathcal P^c})
\\
&+
\widehat V_M^{-1/2}\widehat A_{M,\mathcal P}
(\widehat Y_{\mathcal P}^{\mathrm{inf}}-\mu_{\mathcal P})
\\
&+
\widehat V_M^{-1/2}(\widehat A_M-A_M)\mu.
\end{aligned}
\end{equation}

Now we define the standardized post-selection statistic and the three terms in the right-hand side of \eqref{eq:main_decomposition} as follows:
\begin{align}
    T_{M,N}&=\widehat V_M^{-1/2}(\widehat\beta_M-\beta_{M,N}^{D})\nonumber
    \\
    T_{M,N}^{(1)}&=\widehat V_M^{-1/2}\widehat A_{M,\mathcal P^c}
(\widehat Y_{\mathcal P^c}^{\mathrm{inf}}-\mu_{\mathcal P^c})\nonumber
\\
    T_{M,N}^{(2)}&=\widehat V_M^{-1/2}\widehat A_{M,\mathcal P}
(\widehat Y_{\mathcal P}^{\mathrm{inf}}-\mu_{\mathcal P})\label{eq:T_2}
\\
    T_{M,N}^{(3)}&=\widehat V_M^{-1/2}(\widehat A_M-A_M)\mu. \label{eq:T_3}
\end{align}

To establish Theorem~\ref{thm:pilot_assisted_validity}, it is sufficient to show the following
asymptotic factorization:
\[
d_{\mathrm{TV}}\left[
\mathcal L\left(
T_{M,N},\widehat Y^{\mathrm{sel}}
\mid\mathcal D_N
\right),
N_d(0,I_d)\otimes
\mathcal L\left(
\widehat Y^{\mathrm{sel}}
\mid\mathcal D_N
\right)
\right]
=o_p(1).
\]
Indeed, this implies the desired conditional approximation after
restricting the selection coordinate to $\mathcal A_M$ and using
Assumption~5.

\paragraph{Step 1: Conditional on the pilot events}
Recall that
\(
\mathcal G_{\mathcal P}
\) is the sigma field generated by \(\mathcal D_N,\mathcal P,Y_{\mathcal P},\widehat W_{\mathcal P}\).
Now, the quantities
$\widehat\theta$, $\widehat\Sigma$, $\widehat A_M$, and
$\widehat V_M$ are fixed. Hence,
$T_{M,N}^{(2)}$ and $T_{M,N}^{(3)}$ are also fixed.

The only randomness $Y_{\mathcal P^c}$ and $\widehat W$ are conditionally independent, 
\[
Y_{\mathcal P^c}\mid\mathcal G_{\mathcal P}
\sim
N(\mu_{\mathcal P^c},\Sigma_{0,\mathcal P^c}),
\qquad
\widehat W_{\mathcal P^c}\mid\mathcal G_{\mathcal P}
\sim
N(0,\widehat\Sigma_{\mathcal P^c}).
\]

As $\widehat W_{\mathcal P^c}\mid\mathcal G_{\mathcal P}\sim N(0,\widehat \Sigma_{\mathcal P^c})$, we have 
$$
\begin{pmatrix}
\hat Y^{sel}_{\mathcal P^c}
\\\hat Y^{inf}_{\mathcal P^c}
\end{pmatrix} \bigm| \mathcal G_p\sim N(\begin{pmatrix}\mu_{\mathcal P^c}\\\mu_{\mathcal P^c}\end{pmatrix}\begin{pmatrix}\Sigma_{\mathcal P^c}^{+}&\Sigma_{\mathcal P^c}^{-}\\\Sigma_{\mathcal P^c}^{-}&\Sigma_{\mathcal P^c}^{+}\end{pmatrix}),
$$
where
\[
\Sigma_{\mathcal P^c}^{+}=\Sigma_{0,\mathcal P^c}+\widehat{\Sigma}_{\mathcal P^c},\text{ and }
\Sigma_{\mathcal P^c}^{-}=\Sigma_{0,\mathcal P^c}-\widehat\Sigma_{\mathcal P^c}.
\]

Hence,
\[
\begin{pmatrix}
    \widehat V_M^{-1/2}
\widehat A_{M,\mathcal P^c}
(\widehat Y_{\mathcal P^c}^{\mathrm{inf}}-\mu_{\mathcal P^c})
    \\(\Sigma_{\mathcal P^c}^+)^{-1/2}(\hat Y^{sel}_{\mathcal P^c}-\mu_{\mathcal P^c})
\end{pmatrix}\sim N(0,\Gamma_K),
\]
where 
\[
\Gamma_K=\begin{pmatrix}\widehat V_M^{-1/2}\widehat A_{M,\mathcal P^c}\Sigma^+_{\mathcal P^c}\widehat A_{M,\mathcal P^c}^\top \widehat V_M^{-1/2}&&\widehat V_M^{-1/2}
\widehat A_{M,\mathcal P^c}\Sigma_{\mathcal P^c}^{-}(\Sigma_{\mathcal P^c}^{+})^{-1/2}\\&&\\(\Sigma_{\mathcal P^c}^{+})^{-1/2}\Sigma_{\mathcal P^c}^{-}\widehat A_{M,\mathcal P^c}^\top \widehat V_M^{-1/2}
&&I_{n_{\mathcal{P}^c}}\end{pmatrix},
\]
Therefore, we have
\begin{equation}
\label{eq:Y_s_Y_i}
\begin{pmatrix}
    T_{M,N}
    \\(\Sigma_{\mathcal P^c}^+)^{-1/2}(\hat Y^{sel}_{\mathcal P^c}-\mu_{\mathcal P^c})
\end{pmatrix}\mid\mathcal G_{\mathcal P}\sim N(\mu_K,\Gamma_K),
\end{equation}
where \[\mu_K=\begin{pmatrix}
    T_{M,N}^{(2)}+T_{M,N}^{(3)}\\0
\end{pmatrix}.\]
Let 
\begin{align}
\Gamma_{K,11}&=\widehat V_M^{-1/2}\widehat A_{M,\mathcal P^c}\Sigma^+_{\mathcal P^c}\widehat A_{M,\mathcal P^c}^\top \widehat V_M^{-1/2} \in \mathbb R^{d_M\times d_M},\nonumber
\\\Gamma_{K,12}&=\widehat V_M^{-1/2}
\widehat A_{M,\mathcal P^c}\Sigma_{\mathcal P^c}^{-}(\Sigma_{\mathcal P^c}^{+})^{-1/2}\in \mathbb R^{d_M\times n_{\mathcal P^c}},\qquad \text{and }\nonumber
\\\Gamma_{K,21}&=(\Sigma_{\mathcal P^c}^{+})^{-1/2}\Sigma_{\mathcal P^c}^{-}\widehat A_{M,\mathcal P^c}^\top \widehat V_M^{-1/2}\in \mathbb R^{ n_{\mathcal P^c}\times d_M},
\label{eq:Gamma_K}
\end{align}
where $d_M=2+|M_I|+|M_s|$, defined in \eqref{eq:d_M}, is the dimension of $\beta_{M,N}^D$.

The Pinsker's inequality (Lemma \ref{lemma:Pinsker_inequality}), KL  divergence between multivariate normal distribution (Lemma \ref{lemma:KL_mult}) and Schur's formula (Lemma \ref{lemma:schur}) show that the total variation between  $N(\mu_K,\Gamma_{K})$ and $N(0,I_{d_M+n_{\mathcal P^c}})$ is
\begin{align}
&\Delta_N(\mathcal D_N,\mathcal P,Y_{\mathcal P},\widehat W_{\mathcal P})\\\nonumber:=&d_{TV}\{N(\mu_K,\Gamma_K), N(0,I_{d_M+n_{\mathcal P^c}})\} \notag
\\\le &\sqrt{\frac{1}{2}d_{KL}\{N(\mu_K,\Gamma_K),N(0,I_{d_M+n_{\mathcal P^c}})\}} \notag
\\=&
\frac12\sqrt{\left\{
\operatorname{tr}(\Gamma_K)
-(d_M+n_{\mathcal P^c})
-\log\det(\Gamma_K)+\mu_K^\top \mu_K
\right\}} \notag\\
=&\frac12\sqrt{-\log\det\{(\Gamma_{K,11}-\Gamma_{K,12}\Gamma_{K,12}^\top\}+\mathrm{tr}(\Gamma_{K,11}-I_{d_M})+(\|T_{M,N}^{(2)}\|_2+\|T_{M,N}^{(3)}\|_2)^2}\nonumber
\\
\le &
\frac12\sqrt{-\log\det\{(\Gamma_{K,11}-I_{d_M})-\Gamma_{K,12}\Gamma_{K,12}^\top+I_{d_M}\}+\mathrm{tr}(\Gamma_{K,11}-I_{d_M})+2\|T_{M,N}^{(2)}\|_2^2+2\|T_{M,N}^{(3)}\|_2^2}.
\label{eq:KL_control}
\end{align}
For the last inequality, we use
\[
\mu_K^\top\mu_K=\|T_{M,N}^{(2)}+T_{M,N}^{(3)}\|_2^2 \le 2\|T_{M,N}^{(2)}\|_2^2+2\|T_{M,N}^{(3)}\|_2^2.
\]

\paragraph{Step 2: Stochastic bounds for terms in \eqref{eq:KL_control}.}
We now establish stochastic bounds for the $\Gamma_K$,$T_{M,N}^{(2)}$, and $T_{M,N}^{(3)}$ appearing in \eqref{eq:KL_control}. Unlike in Step~1, we no
longer condition on $\mathcal G_{\mathcal P}$ when deriving these
bounds. Thus, $\Gamma_K$,$T_{M,N}^{(2)}$, and $T_{M,N}^{(3)}$ , which were fixed under the
conditional Gaussian representation given $\mathcal G_{\mathcal P}$,
are now regarded as random.

We will control
$\Gamma_{K,11}-I_{d_M}$, $\Gamma_{K,12}$, $T_{M,N}^{(2)}$, and
$T_{M,N}^{(3)}$ respectively. These bounds imply that the corresponding
KL divergence, and hence the total variation distance, converges to
zero in probability.

\paragraph{Step 2.1: Variance terms {\unboldmath$\Gamma_{K}$}}

Recall that the estimated covariance matrix \eqref{eq:hat_VM} can be separated by
\[
\begin{aligned}
\widehat V_M&=2\widehat A_M\widehat \Sigma\widehat A_M^\top\\&=2\begin{pmatrix}\widehat A_{M,\mathcal P},\widehat A_{M,\mathcal P^c}\end{pmatrix}\begin{pmatrix}\widehat \Sigma_{\mathcal P}&0\\0&\widehat\Sigma_{\mathcal P^c}\end{pmatrix}\begin{pmatrix}\widehat A_{M,\mathcal P}^\top\\\widehat A_{M,\mathcal P^c}^\top\end{pmatrix}\\&=2\widehat A_{M,\mathcal P}\widehat \Sigma_{\mathcal P}\widehat A_{M,\mathcal P}^\top+2\widehat A_{M,\mathcal P^c}\widehat \Sigma_{\mathcal P^c}\widehat A_{M,\mathcal P^c}^\top.
\end{aligned}
\]

The difference between $\Gamma_{K,11}$ and the identity matrix in terms of the spectral norm is

\begin{align}
    \|\Gamma_{K,11}-I\|_s&=\|\widehat V_M^{-1/2}\{\widehat A_{M,\mathcal P^c}
(\Sigma_{0,\mathcal P^c}+\widehat\Sigma_{\mathcal P^c})
\widehat A_{M,\mathcal P^c}^\top
-
\widehat V_M\}\widehat V_M^{-1/2}\|_s \nonumber
\\&=\|\widehat V_M^{-1/2}\{\widehat A_{M,\mathcal P^c}
(\Sigma_{0,\mathcal P^c}-\widehat\Sigma_{\mathcal P^c})
\widehat A_{M,\mathcal P^c}^\top
-
2\widehat A_{M,\mathcal P}\widehat\Sigma_{\mathcal P}
\widehat A_{M,\mathcal P}^\top\}\widehat V_M^{-1/2}\|_s \nonumber
\\&\le\|\widehat V_M^{-1/2}\{\widehat A_{M,\mathcal P^c}
(\Sigma_{0,\mathcal P^c}-\widehat\Sigma_{\mathcal P^c})
\widehat A_{M,\mathcal P^c}^\top\}\widehat V_M^{-1/2}\|_s
+\|\widehat V_M^{-1/2}\{2\widehat A_{M,\mathcal P}\widehat\Sigma_{\mathcal P}
\widehat A_{M,\mathcal P}^\top\}\widehat V_M^{-1/2}\|_s. 
\label{eq:Gamma11}
\end{align}

For the first term of \eqref{eq:Gamma11}, we have

\begin{align}
    &\|\widehat V_M^{-1/2}\{\widehat A_{M,\mathcal P^c}
(\Sigma_{0,\mathcal P^c}-\widehat\Sigma_{\mathcal P^c})
\widehat A_{M,\mathcal P^c}^\top\}\widehat V_M^{-1/2}\|_s \nonumber
\\\le &\|\widehat V_M^{-1/2}\widehat A_{M,\mathcal P^c}
\widehat \Sigma_{\mathcal P^c}^{1/2}\|_s
\left\|
\widehat\Sigma_{\mathcal P^c}^{-1/2}
(\Sigma_{0,\mathcal P^c}-\widehat\Sigma_{\mathcal P^c})
\widehat\Sigma_{\mathcal P^c}^{-1/2}
\right\|_s\|\widehat \Sigma_{\mathcal P^c}^{1/2}\widehat A_{M,\mathcal P^c} \widehat V_M^{-1/2}\|_s \nonumber
\\= &\|\widehat V_M^{-1/2}\widehat A_{M,\mathcal P^c}
\widehat \Sigma_{\mathcal P^c}^{1/2}\|_s^2
\left\|
\widehat\Sigma_{\mathcal P^c}^{-1/2}
(\Sigma_{0,\mathcal P^c}-\widehat\Sigma_{\mathcal P^c})
\widehat\Sigma_{\mathcal P^c}^{-1/2}
\right\|_s. \label{eq:proof_2_1_0}
\end{align}

Note that 
\[
\widehat A_{M,\mathcal P^c}\widehat\Sigma_{\mathcal P^c}\widehat A_{M,\mathcal P^c}^\top\preceq \widehat A_M\widehat \Sigma\widehat A_M^\top=\frac{1}{2}\widehat V_M.
\]
As $\|A\|_s^2\le \|AA^T\|_s$. Therefore, we have \begin{equation}
\label{eq:proof_2_11}
\|\widehat V_M^{-1/2}\widehat A_{M,\mathcal P^c}
\widehat \Sigma_{\mathcal P^c}^{1/2}\|_s^2=\|\widehat V_M^{-1/2}\widehat A_{M,\mathcal P^c}
\widehat \Sigma_{\mathcal P^c}\widehat A_{M,\mathcal P^c}^\top\widehat V_M^{-1/2}\|_s\le \frac{1}{2}.
\end{equation}

Let $\widehat E_{\mathcal P^c}$ be the non-pilot principal block of $\widehat E$, defined in \eqref{eq:hat_E_dec}. 
Therefore, we have
\[
\Sigma_{0,\mathcal P^c}-\widehat \Sigma_{\mathcal P^c}=-\Sigma_{0,\mathcal P^c}^{1/2}\widehat E_{\mathcal P^c}\Sigma_{0,\mathcal P^c}^{1/2}\text{ and } \widehat \Sigma_{\mathcal P^c}=\Sigma_{0,\mathcal P^c}^{1/2}(I+\widehat E_{\mathcal P^c})^{}\Sigma_{0,\mathcal P^c}^{1/2}.
\]
As $\widehat E=\mathrm{blockdiag}(\hat E_{\mathcal P},\widehat E_{\mathcal P^c})$, we have
$$
\|\widehat E_{\mathcal P^c}\|_s\le \max\{\|\widehat E_{\mathcal P}\|_s,\|\widehat E_{\mathcal P^c}\|_s\}
=\|\widehat E\|_s.
$$
Recall that by Lemma \ref{lem:relative_covariance_lipschitz}, we have $\|\hat E\|_s=O_p(N_s^{-1/2})=o_p(1)$.
When $\|\hat E\|_s\le 1/2$, with probability tending to $1$,
\begin{align}
&\left\|
\widehat\Sigma_{\mathcal P^c}^{-1/2}
(\Sigma_{0,\mathcal P^c}-\widehat\Sigma_{\mathcal P^c})
\widehat\Sigma_{\mathcal P^c}^{-1/2}
\right\|_s \nonumber
\\\le&\|
\widehat\Sigma_{\mathcal P^c}^{-1/2}\Sigma_{0,\mathcal P^c}^{1/2}\|_s\|\Sigma_{0,\mathcal P^c}^{-1/2}
(\Sigma_{0,\mathcal P^c}-\widehat\Sigma_{\mathcal P^c})\Sigma_{0,\mathcal P^c}^{-1/2}\|_s\|\Sigma_{0,\mathcal P^c}^{1/2}
\widehat\Sigma_{\mathcal P^c}^{-1/2}
\|_s \nonumber
\\=&\|\widehat E_{\mathcal P^c}\|_s\|(I+\widehat E_{\mathcal P^c})^{-1}\|_s \nonumber
\\\le &
\frac{\|\widehat E_{\mathcal P^c}\|_s}{1-\|\widehat E_{\mathcal P^c}\|_s}\nonumber
\\\le& \frac{\|\widehat E\|_s}{1-\|\widehat E\|_s} \nonumber
\\\le& 2\|\widehat E\|_s.
\label{eq:proof_2_12}
\end{align}

Combine \eqref{eq:proof_2_1_0}, \eqref{eq:proof_2_11} and \eqref{eq:proof_2_12} we have
\begin{equation}
\label{eq:proof1.1}
\|\widehat V_M^{-1/2}\{\widehat A_{M,\mathcal P^c}
(\Sigma_{0,\mathcal P^c}-\widehat\Sigma_{\mathcal P^c})
\widehat A_{M,\mathcal P^c}^\top\}\widehat V_M^{-1/2}\|_s\le \|\widehat E\|_s=O_p(N_s^{-1/2}).
\end{equation}

According to
Lemma~\ref{lem:empirical_information_order},
we have $\|\widehat V_M\|_s=O_p(N^{-1})$, and $\|\widehat H_{M,\mathcal P}^{}\|_s=O_p(N_s)$. Recall that, by \eqref{eq:A_P}, \[\widehat A_{M.\mathcal P}=(X_{M}^\top\widehat \Sigma^{-1}X_M)^{-1}X_{M,\mathcal P}^\top\widehat\Sigma^{-1}_{\mathcal P}=\frac{1}{2}\widehat V_MX_{M,\mathcal P}^\top\widehat\Sigma^{-1}_{\mathcal P}.\] So, for the second term of \eqref{eq:Gamma11}, 
\begin{align}
\|\widehat V_M^{-1/2}
\left(2\widehat A_{M,\mathcal P}\widehat\Sigma_{\mathcal P}
\widehat A_{M,\mathcal P}^\top\right)
\widehat V_M^{-1/2}\|_s
&=\|\frac{1}{2}\widehat V_M^{-1/2}\widehat V_MX_{M,\mathcal P}^\top\widehat\Sigma^{-1}_{\mathcal P}\widehat\Sigma^{}_{\mathcal P}\widehat\Sigma^{-1}_{\mathcal P}X_{M,\mathcal P}\widehat V_M\widehat V_M^{-1/2}\|\\&
=\frac{1}{2}\|\widehat V_M^{1/2}\widehat H_{M,\mathcal P}\widehat V_M^{1/2}\|_s \nonumber
\\&\le \|\widehat V_M^{1/2}\|_s^2\|\widehat H_{M,\mathcal P}^{}\|_s \nonumber
\\&\le \|\widehat V_M^{}\|_s\|\widehat H^{}_{M,\mathcal P}\|_s \nonumber
\\&=O_p(N_s/N).
\label{eq:proof1.2}
\end{align}

Therefore, by \eqref{eq:Gamma11},  \eqref{eq:proof1.1} and \eqref{eq:proof1.2}, the spectrum norm of $\Gamma_{K,11}-I$ is bounded by
\[
\|\Gamma_{K,11}-I\|_s=O_p(N_s/N+N^{-1/2}_s)=o_p(1).
\]

For the cross product term $\Gamma_{K,12}$ in \eqref{eq:Gamma_K}, when $\|\widehat E\|_s\le 1/2$, with probability tending to 1, we have
\[
\begin{aligned}
\|\Gamma_{K,12}\|_s&=\|\widehat V_M^{-1/2}
\widehat A_{M,\mathcal P^c}(\widehat \Sigma_{\mathcal P^c})^{1/2}(\widehat \Sigma_{\mathcal P^c})^{-1/2}
\Sigma_{\mathcal P^c}^{-}
(\Sigma_{\mathcal P^c}^{+})^{-1/2}\|_s
\\&\le \|\widehat V_M^{-1/2}
\widehat A_{M,\mathcal P^c}(\widehat \Sigma_{\mathcal P^c})^{1/2}\|_s\| (\widehat \Sigma_{\mathcal P^c})^{-1/2}
\Sigma_{\mathcal P^c}^{-}
(\Sigma_{\mathcal P^c}^{+})^{-1/2}\|_s
\\&\le  \frac{1}{\sqrt{2}}\| (\widehat \Sigma_{\mathcal P^c})^{-1/2}
\Sigma_{\mathcal P^c}^{-}
(\Sigma_{\mathcal P^c}^{+})^{-1/2}\|_s \qquad \text{(by \eqref{eq:proof_2_11})}
\\&\le  \frac{1}{\sqrt{2}}\| (\widehat \Sigma_{\mathcal P^c})^{-1/2}\Sigma_{0,\mathcal P^c}^{1/2}\|_s\|\Sigma_{0,\mathcal P^c}^{-1/2}
\Sigma_{\mathcal P^c}^{-}\Sigma_{0,\mathcal P^c}^{-1/2}\|_s\|\Sigma_{0,\mathcal P^c}^{1/2}
(\Sigma_{\mathcal P^c}^{+})^{-1/2}\|_s
\\&=\frac{1}{\sqrt{2}}\| \Sigma_{0,\mathcal P^c}^{1/2}\widehat \Sigma_{\mathcal P^c}^{-1}\Sigma_{0,\mathcal P^c}^{1/2}\|_s^{1/2}\|\Sigma_{0,\mathcal P^c}^{-1/2}
\Sigma_{\mathcal P^c}^{-}\Sigma_{0,\mathcal P^c}^{-1/2}\|_s\|\Sigma_{0,\mathcal P^c}^{1/2}
(\Sigma_{\mathcal P^c}^{+})^{-1}\Sigma_{0,\mathcal P^c}^{1/2}\|_s^{1/2}
\\&=\frac{1}{\sqrt{2}}\|(I+\widehat E_{\mathcal P^c})^{-1}\|_s^{1/2}\|\widehat E_{\mathcal P^c}\|_s\|(2I+\widehat E_{\mathcal P^c})^{-1}\|_s^{1/2}
\\&\le \frac{\|\widehat E_{\mathcal P^c}\|_s }{\sqrt{2(1-\|\widehat E_{\mathcal P^c}\|_s)(2-\|\widehat E_{\mathcal P^c}\|_s)}}
\\&\le \frac{\|\widehat E_{}\|_s }{\sqrt{2(1-\|\widehat E_{}\|_s)(2-\|\widehat E_{}\|_s)}}
\\&\le \|\widehat E\|_s=O_p(N_s^{-1/2})=o_p(1).
\end{aligned}
\]
The preceding bounds imply
\[
\left\|
\Gamma_{K,11}-I_d
-
\Gamma_{K,12}\Gamma_{K,12}^{\top}
\right\|_s\le\|\Gamma_{K,11}-I_d\|_s+\|\Gamma_{K,12}\|_s^2
=o_p(1).
\]
Hence, with probability tending to one,
\[
\left\|
\Gamma_{K,11}-I_d
-
\Gamma_{K,12}\Gamma_{K,12}^{\top}
\right\|_s
\le \frac12.
\]
Therefore, by Lemma \ref{lemma:KL_cal}, we have 
\begin{equation}
\label{eq:Gamma_control_tr}
    \mathrm{tr}(\Gamma_{K,11}-I)|\le d\|\Gamma_{K,11}-I\|_s=o_p(1),
\end{equation}|
and
\begin{align}
&|\log\det\{(\Gamma_{K,11}-I)-\Gamma_{K,12}\Gamma_{K,12}^\top+I\}|\nonumber \\\le& 2d\|\Gamma_{K,11}-I-\Gamma_{K,12}\Gamma_{K,12}^T\|_s\nonumber\\\le& 2d(\|\Gamma_{K,11}-I\|_s+\|\Gamma_{K,12}\|_s^2)=o_p(1).
\label{eq:Gamma_control}
\end{align}

\paragraph{Step 2.2: Bound of {\unboldmath$T^{(2)}_{M,N}$}.}

Recall from \eqref{eq:T_2} that
\[
T_{M,N}^{(2)}
=
\widehat V_M^{-1/2}
\widehat A_{M,\mathcal P}
(\widehat Y_{\mathcal P}^{\mathrm{inf}}-\mu_{\mathcal P}).
\]
Since
\[
\widehat A_{M,\mathcal P}
=
\widehat H_M^{-1}
X_{M,\mathcal P}^{\top}
\widehat\Sigma_{\mathcal P}^{-1},
\]
we can write
\begin{equation}
    \label{eq:T2}
T_{M,N}^{(2)}
=
\widehat V_M^{-1/2}\widehat H_M^{-1}
\left\{
X_{M,\mathcal P}^{\top}
\widehat\Sigma_{\mathcal P}^{-1}
(\widehat Y_{\mathcal P}^{\mathrm{inf}}-\mu_{\mathcal P})
\right\}.
\end{equation}
By Lemma~\ref{lem:empirical_information_order},
\[
\left\|
\widehat V_M^{-1/2}\widehat H_M^{-1}
\right\|_s
=
O_p(N^{-1/2}).
\]
It therefore remains to control $X_{M,\mathcal P}^{\top}
\widehat\Sigma_{\mathcal P}^{-1}
(\widehat Y_{\mathcal P}^{\mathrm{inf}}-\mu_{\mathcal P})$.
As 
\(
\widehat Y_{\mathcal P}^{\mathrm{inf}}-\mu_{\mathcal P}
=
(Y_{\mathcal P}-\mu_{\mathcal P})-\widehat W_{\mathcal P},
\)
we have 
\begin{equation}
\label{eq:proof_2_2_1}
    X_{M,\mathcal P}^{\top}
	\widehat\Sigma_{\mathcal P}^{-1}
	(\widehat Y_{\mathcal P}^{\mathrm{inf}}-\mu_{\mathcal P})=
	-
	X_{M,\mathcal P}^{\top}
	\widehat\Sigma_{\mathcal P}^{-1}
	\widehat W_{\mathcal P}
	+
	X_{M,\mathcal P}^{\top}
	\widehat\Sigma_{\mathcal P}^{-1}
	(Y_{\mathcal P}-\mu_{\mathcal P}).
\end{equation}

We control the two terms in \eqref{eq:proof_2_2_1} separately. First, conditional on
$(Y_{\mathcal P},\mathcal D_N,\mathcal P)$,
the pilot covariance estimate $\widehat\Sigma_{\mathcal P}$ is fixed,
whereas
\[
\widehat W_{\mathcal P}
\mid
Y_{\mathcal P},\mathcal D_N,\mathcal P
\sim
N(0,\widehat\Sigma_{\mathcal P}).
\]
Therefore,
\[
\begin{aligned}
	&
	\mathbb E\left[
	\left\|
	X_{M,\mathcal P}^{\top}
	\widehat\Sigma_{\mathcal P}^{-1}
	\widehat W_{\mathcal P}
	\right\|_2^2
	\middle|
	Y_{\mathcal P},\mathcal D_N,\mathcal P
	\right]
	\\
	&=
	\mathbb E\left[
	\widehat W_{\mathcal P}^{\top}
	\widehat\Sigma_{\mathcal P}^{-1}
	X_{M,\mathcal P}X_{M,\mathcal P}^{\top}
	\widehat\Sigma_{\mathcal P}^{-1}
	\widehat W_{\mathcal P}
	\middle|
	Y_{\mathcal P},\mathcal D_N,\mathcal P
	\right]
	\\
	&=
	\operatorname{tr}\left[
	\widehat\Sigma_{\mathcal P}^{-1}
	X_{M,\mathcal P}X_{M,\mathcal P}^{\top}
	\widehat\Sigma_{\mathcal P}^{-1}
	\mathbb E\left(
	\widehat W_{\mathcal P}
	\widehat W_{\mathcal P}^{\top}
	\middle|
	Y_{\mathcal P},\mathcal D_N,\mathcal P
	\right)
	\right]
	\\
	&=
	\operatorname{tr}\left(
	X_{M,\mathcal P}^{\top}
	\widehat\Sigma_{\mathcal P}^{-1}
	X_{M,\mathcal P}
	\right)
	\\
	&=
	\operatorname{tr}(\widehat H_{M,\mathcal P})
	\\
	&\le
	d_M\|\widehat H_{M,\mathcal P}\|_s.
\end{aligned}
\]
Since, by Lemma~\ref{lem:empirical_information_order}, 
$\|\widehat H_{M,\mathcal P}\|_s=O_p(N_s)$,
 the Markov inequality gives
\begin{equation}
\label{eq:proof_2_2_2}
  \left\|
X_{M,\mathcal P}^{\top}
\widehat\Sigma_{\mathcal P}^{-1}
\widehat W_{\mathcal P}
\right\|_2
=
O_p(\sqrt{N_s}).  
\end{equation}

We next consider the second term in \eqref{eq:proof_2_2_1}.
Decompose
\[
\begin{aligned}
	X_{M,\mathcal P}^{\top}
	\widehat\Sigma_{\mathcal P}^{-1}
	(Y_{\mathcal P}-\mu_{\mathcal P})=
	X_{M,\mathcal P}^{\top}
	\Sigma_{0,\mathcal P}^{-1}
	(Y_{\mathcal P}-\mu_{\mathcal P})+
	X_{M,\mathcal P}^{\top}
	(\widehat\Sigma_{\mathcal P}^{-1}
	-\Sigma_{0,\mathcal P}^{-1})
	(Y_{\mathcal P}-\mu_{\mathcal P}).
\end{aligned}
\]

We also find that 
\begin{align} &
\|X_{M,\mathcal P}^\top\widehat\Sigma_{\mathcal P}^{-1} (Y_{\mathcal P}-\mu_{\mathcal P})\|_s \nonumber
\\=&\| X_{M,\mathcal P}^\top\Sigma_{0,\mathcal P}^{-1} (Y_{\mathcal P}-\mu_{\mathcal P}) + X_{M,\mathcal P}^\top (\widehat\Sigma_{\mathcal P}^{-1}-\Sigma_{0,\mathcal P}^{-1}) (Y_{\mathcal P}-\mu_{\mathcal P})\|_s \nonumber
\\\le &\| X_{M,\mathcal P}^\top\Sigma_{0,\mathcal P}^{-1} (Y_{\mathcal P}-\mu_{\mathcal P})\|_s+\|X_{M,\mathcal P}^\top (\widehat\Sigma_{\mathcal P}^{-1}-\Sigma_{0,\mathcal P}^{-1}) (Y_{\mathcal P}-\mu_{\mathcal P})\|_s \nonumber
\\\le &\| X_{M,\mathcal P}^\top\Sigma_{0,\mathcal P}^{-1} (Y_{\mathcal P}-\mu_{\mathcal P})\|_s
+\|X^\top_{M,\mathcal P}\Sigma_{0,\mathcal P}^{-1/2}\|_s \|(I+\widehat E_{\mathcal P})^{-1}-I\|_s \|\Sigma_{0,\mathcal P}^{-1/2}(Y_{\mathcal P}-\mu_{\mathcal P})\|_2 . 
\label{eq:proof_2_2_3}
\end{align}
By Lemma \ref{lem:empirical_information_order}, we have 
\[ \begin{aligned} \|X_{M,\mathcal P}^\top\Sigma_{0,\mathcal P}^{-1/2}\|_s&=O_p(\sqrt{N_s}),
\\ \|\Sigma_{0,\mathcal P}^{-1/2}(Y_{\mathcal P}-\mu_{\mathcal P})\|_2&=O_p(\sqrt{N_s}),
\\ \|X_{M,\mathcal P}^\top \Sigma_{0,\mathcal P}^{-1} (Y_{\mathcal P}-\mu_{\mathcal P})\|_2&=O_p(\sqrt{N_s}). \end{aligned} \] 
By Lemma \ref{lem:relative_covariance_lipschitz}, \(\|\hat E_{\mathcal P}\|_s\le \|\hat E_{}\|_s=O_p(N^{-1/2}_s)\). Then, when $\|\hat E_{}\|_s\le \frac{1}{2}$, with probability tending to $1$, we have \[ \begin{aligned} &\|(I+\widehat E_{\mathcal P})^{-1}-I\|_s \\=& \|I(I+\widehat E_{\mathcal P}-I)(I+\widehat E_{\mathcal P})^{-1}\|_s \\\le &\|\hat E_{\mathcal P}\|_s\|(I+\widehat E_{\mathcal P})^{-1}\|_s\\\le &\frac{\|\hat E_{\mathcal P}\|_s}{1-\|\widehat E_{\mathcal P}^{}\|_s}\\\le&2\|\hat E_{\mathcal P}\|_s=O_p(N_s^{-1/2}) \end{aligned} \] 
Here the first equality use the fact that $A^{-1}-B^{-1}=A^{-1}(B-A)B^{-1}$, when $A,B$ are invertible matrix.
Therefore, 
\begin{equation}
    \label{eq:proof_2_2_4}
    \|X_{M,\mathcal P}^\top\widehat\Sigma_{\mathcal P}^{-1} (Y_{\mathcal P}-\mu_{\mathcal P})\|_s=O_p(N_s^{1/2})
\end{equation} 

By \eqref{eq:proof_2_2_1}, \eqref{eq:proof_2_2_2}, \eqref{eq:proof_2_2_3}, and \eqref{eq:proof_2_2_4}, we have
\begin{equation}
    \label{eq:proof_2_2_5}
\left\|X_{M,\mathcal P}^\top\widehat\Sigma_{\mathcal P}^{-1} (\widehat Y^{inf}_{\mathcal P}-\mu_{\mathcal P})\right\|_s\le \|X_{M,\mathcal P}^{\top}
	\widehat\Sigma_{\mathcal P}^{-1}
	\widehat W_{\mathcal P}\|_s
	+
	\|X_{M,\mathcal P}^{\top}
	\widehat\Sigma_{\mathcal P}^{-1}
	(Y_{\mathcal P}-\mu_{\mathcal P})\|_s=O_p(N^{1/2}_s)
\end{equation}

Finally, by \eqref{eq:T2} and \eqref{eq:proof_2_2_5}, we have

\begin{equation} \label{eq:pilot_negligible} \|T_{M,N}^{(2)}\|_2\le \left\| \widehat V_M^{-1/2}\widehat H_M^{-1} \right\|_s\left\|X_{M,\mathcal P}^\top\widehat\Sigma_{\mathcal P}^{-1} (\widehat Y^{inf}_{\mathcal P}-\mu_{\mathcal P})\right\|_s = O_p\left(\sqrt{\frac{N_s}{N}}\right) = o_p(1). \end{equation}

\paragraph{Step 2.3: Bound of {\unboldmath$T_{M,N}^{(3)}$}.}
Recall that, in \eqref{eq:T_3}, we have 
\[
T_{M,N}^{(3)}=\widehat V_M^{-1/2}(\widehat A_M-A_M)\mu,
\]
and in \eqref{eq:AM_hat}, we have
\[
\begin{aligned}
A_M&=(X_M^\top\Sigma_0^{-1}X_M)^{-1}X_M^\top\Sigma_{0}^{-1} \\
\widehat A_M&=(X_M^\top\widehat\Sigma^{-1}X_M)^{-1}X_M^\top\widehat\Sigma^{-1}. 
\end{aligned}
\]
Then we have
\[
\beta_{M,N}^D=(X_M^\top \Sigma_0^{-1}X_M)^{-1}X_M^T\Sigma_0^{-1}\mu=A_M\mu.
\]

Let
\[
r_M=\mu-X_M\beta_{M,N}^{D}=\mu-X_M(X_M^\top \Sigma_0^{-1}X_M)^{-1}X_M^T\Sigma_0^{-1}\mu.
\]
Then, we have
$X_M^\top\Sigma_0^{-1}r_M=0$, and therefore $A_Mr_M=0$. Since
$A_MX_M=\widehat A_MX_M=I_d$,

\begin{align}
(\widehat A_M-A_M)\mu 
&=(\widehat A_M-A_M)(r_M+X_M\beta_{M,N}^D) \nonumber
\\&=\widehat A_M r_M-A_Mr_M+(\widehat A_M X_M-A_M X_M)\beta_{M,N}^D \nonumber
\\&=\widehat A_M r_M \nonumber
\\&=\widehat H_M^{-1}X_M^\top
\widehat\Sigma^{-1}r_M-0 \nonumber
\\&=\widehat H_M^{-1}X_M^\top
(\widehat\Sigma^{-1}-\Sigma_0^{-1})r_M.
\label{eq:dif_A_mu}
\end{align}

By Lemma \ref{lem:empirical_information_order}, $\|\hat H^{-1}_M\|=O_p(N^{-1}),$ and we only need to bound $X_M^\top(\widehat\Sigma^{-1}-\Sigma_0^{-1})r_M$.

Let $r_M^P=\mu-X_M\beta_M^P$. Then
\[
r_M=r_M^P+X_M(\beta_M^P-\beta_{M,N}^{D}),
\]
so
\begin{equation}
\label{eq:proof_2_3_1}
X_M^\top(\widehat\Sigma^{-1}-\Sigma_0^{-1})r_M
={}
\underbrace{X_M^\top(\widehat\Sigma^{-1}-\Sigma_0^{-1})r_M^P}_{H_1}
+
\underbrace{X_M^\top(\widehat\Sigma^{-1}-\Sigma_0^{-1})X_M
(\beta_M^P-\beta_{M,N}^{D})}_{H_2}.
\end{equation}

\paragraph{Step 2.3.1: Bound of {\unboldmath $H_1$}.}
For $\theta\in\Theta_0$, define
\[
h_i(\theta)
=
X_{i,M}^\top
\{\Sigma_i(\theta)^{-1}-\Sigma_{0,i}^{-1}\}
r_{Mi}^P.
\]
Clearly, $h_i(\theta_0)=0$, and
\[
H_1=\sum_{i=1}^N h_i(\widehat\theta).
\]

We next show that
\[
\mathbb E\{h_i(\theta)\}=0,
\qquad \forall \theta\in\Theta_0.
\]
To see this, for each fixed $\theta\in\Theta_0$, consider the
covariance-weighted population projection
\[
\beta_M^P(\theta)
=
\left[
\mathbb E\left\{
X_{i,M}^\top\Sigma_i(\theta)^{-1}X_{i,M}
\right\}
\right]^{-1}
\mathbb E\left\{
X_{i,M}^\top\Sigma_i(\theta)^{-1}\mu_i
\right\}.
\]
By Assumption~3, $M_I=M_S=J$.
Therefore, Theorem~\ref{thm:beta_est}, applied for each fixed
$\theta\in\Theta_0$, gives
\[
\beta_M^P(\theta)
=
\beta_M^P(\theta_0)
=
\beta_M^P.
\]
In particular, under the common-support condition, the population
projection does not depend on the covariance parameter $\theta$.

For every fixed $\theta\in\Theta_0$, we have
\[
\mathbb E\left[
X_{i,M}^\top\Sigma_i(\theta)^{-1}
\{\mu_i-X_{i,M}\beta_M^P(\theta)\}
\right]
=0.
\]
Since $\beta_M^P(\theta)=\beta_M^P$, this becomes
\[
\mathbb E\left[
X_{i,M}^\top\Sigma_i(\theta)^{-1}r_{Mi}^P
\right]
=0.
\]
Take $\theta=\theta_0$,
\[
\mathbb E\left[
X_{i,M}^\top\Sigma_{0,i}^{-1}r_{Mi}^P
\right]
=0.
\]
Consequently,
\[
\begin{aligned}
\mathbb E\{h_i(\theta)\}
&=
\mathbb E\left[
X_{i,M}^\top\Sigma_i(\theta)^{-1}r_{Mi}^P
\right]
-
\mathbb E\left[
X_{i,M}^\top\Sigma_{0,i}^{-1}r_{Mi}^P
\right]
=0,
\qquad \theta\in\Theta_0.
\end{aligned}
\]

By Lemma \ref{lemma:h_bound}, there exists $L_0>0$ such that
\[
\begin{aligned}
\|h_i(\theta_1)-h_i(\theta_2)\|_2\le 
L_0
\|X_{i,M}^\top\Sigma_{0i}^{-1}r_{Mi}^P\|_2
\|\theta_1-\theta_2\|_2.
\end{aligned}
\] 
By Assumption~4, the Lipschitz factor
\[
Q_i
=
L_0
\|X_{i,M}^\top\Sigma_{0i}^{-1}r_{Mi}^P\|_2
\]
satisfies $\mathbb E(Q_i^2)<\infty$.
Lemma~\ref{lem:localized_bracketing_bound}, with $q=5$ and $d=d_M$, now
gives, for every deterministic $\delta_N\downarrow0$,
\[
\sup_{\|\theta-\theta_0\|_2\le\delta_N}
\left\|\sum_{i=1}^Nh_i(\theta)\right\|_2
=
O_p(\sqrt N\,\delta_N).
\]
By Assumption~2, for every $\eta>0$, there exists a sufficiently large
constant $K_0<\infty$ such that
\[
\mathbb P\left(
\|\widehat\theta-\theta_0\|_2
\le K_0 N_s^{-1/2}
\right)\ge1-\eta
\]
for all sufficiently large $N$.
Applying Lemma~\ref{lem:localized_bracketing_bound} with the deterministic sequence
$\delta_N=K_0 N_s^{-1/2}$ gives
\begin{equation}
\|H_1\|_2
=
O_p(\sqrt N\,N_s^{-1/2}).
\label{eq:H1-bound}
\end{equation}

\paragraph{Step 2.3.2: Bound of {\unboldmath $H_2$.}}
Recall that 
\begin{equation}
\label{eq:H2}
   H_2= X_M^\top(\widehat\Sigma^{-1}-\Sigma_0^{-1})X_M
(\beta_M^P-\beta_{M,N}^{D})
\end{equation}
Note that we have
\[
\begin{aligned}
    &\beta_{M,N}^D-\beta_{M}^P
    \\=&(X^T_M\Sigma^{-1}_{0}X_M)^{-1}X^T_M\Sigma_0^{-1}\mu -\beta_M^P
    \\=&(X^T_M\Sigma^{-1}_{0}X_M)^{-1}X^T_M\Sigma_0^{-1}(r_M^P+X_M\beta_M^P) -\beta_M^P
    \\=&(X^T_M\Sigma^{-1}_{0}X_M)^{-1}X^T_M\Sigma_0^{-1}r_M^P
    \\=&
    H_M^{-1}\sum_{i=1}^N X_{i,M}^\top\Sigma_{0i}^{-1}r_{Mi}^P.
\end{aligned}
\]
Since the participants are independent, with
$$
\mathbb E[X_{i,M}^\top \Sigma_{0i}^{-1}r_{Mi}^P]=0,
$$
and, by Assumption 4,
$$
\mathbb E[\|X_{i,M}^\top \Sigma_{0i}^{-1}r_{Mi}^P\|^2_2]<C_h,
$$
the Central Limit Theorem gives $\|\sum\limits_{i=1}^NX_{i,M}^\top \Sigma_{0i}^{-1}r_{Mi}^P\|_2=O_p(N^{1/2}).$ Since by Lemma \ref{lem:empirical_information_order}, $\|H_M^{-1}\|_s=O_p(N^{-1})$, we have
\begin{equation}
\label{eq:proof_2_3_2}
    \|\beta_{M,N}^{D}-\beta_M^P\|_2=
O_p(N^{-1/2}).
\end{equation}

For \(X_M^\top(\widehat\Sigma^{-1}-\Sigma_0^{-1})X_M\),  we find
\[
\widehat\Sigma^{-1}-\Sigma_0^{-1}
=
\Sigma_0^{-1/2}\{(I_n+\widehat E)^{-1}-I_n\}\Sigma_0^{-1/2},
\]
and, on $\|\widehat E\|_s<1/2$,
\[
\|(I_n+\widehat E)^{-1}-I_n\|_s\le2\|\widehat E\|_s.
\]
By Lemma \ref{lem:empirical_information_order}, 
\(\|H_M\|_s=O_p(N)\).
Therefore,
\begin{align}
&\left\|
X_M^\top(\widehat\Sigma^{-1}-\Sigma_0^{-1})X_M
\right\|_s \nonumber
\\\le&\|X_M^\top\Sigma_{0}^{-1/2}\|_s\|(I_n+\widehat E)^{-1}-I_n\|_s\|\Sigma_{0}^{-1/2}X_M\|_s\nonumber
\\\le &
2\|\widehat E\|_s\|X_M^\top\Sigma_{0}^{-1}X_M\|_s\nonumber
\\= &
2\|\widehat E\|_s\|H_M\|_s
=
O_p(NN_s^{-1/2}).
\label{eq:proof_2_3_3}
\end{align}
Hence, by \eqref{eq:proof_2_3_2} and \eqref{eq:proof_2_3_3}
\begin{equation}
\label{eq:H2_bound}
\|H_2\|_2\le \| X_M^\top(\widehat\Sigma^{-1}-\Sigma_0^{-1})X_M\|_s\|
(\beta_M^P-\beta_{M,N}^{D})\|_s
=
O_p(\sqrt N\,N_s^{-1/2}).
\end{equation}
Combining \eqref{eq:dif_A_mu}, \eqref{eq:proof_2_3_1}, \eqref{eq:H1-bound}, \eqref{eq:H2_bound}, and
$\|\widehat H_M^{-1}\|_s=O_p(N^{-1})$ gives
\[
\|(\widehat A_M-A_M)\mu\|_2\le \|\hat H_M^{-1}\|_s(\|H_1\|_s+\|H_2\|_s)
=
O_p\left(\frac1{\sqrt{NN_s}}\right)
=
o_p(N^{-1/2}).
\]
By Lemma \ref{lem:empirical_information_order} , $\|\hat V_M^{-1/2}\|=O_p(N^{1/2})$. Consequently,
\begin{equation}
\label{eq:center_negligible}
\|T_{M,N}^{(3)}\|_2=\left\|
\widehat V_M^{-1/2}(\widehat A_M-A_M)\mu
\right\|_2
=
o_p(1).
\end{equation}

\paragraph{Step 3: Conclusion.}
By \eqref{eq:KL_control},\eqref{eq:Gamma_control},\eqref{eq:pilot_negligible}, and \eqref{eq:center_negligible} we have
\[
\Delta_N(\mathcal D_N,\mathcal P,Y_{\mathcal P},\widehat W_{\mathcal P})=o_p(1).
\]

For any \(z\in\mathbb R^d\), given $\mathcal G_{\mathcal P}$, we have 
\[
\begin{aligned}
&\left|
\mathbb P\left\{
T_{M,N}\le z,
\widehat Y^{\mathrm{sel}}\in\mathcal A_M
\mid\mathcal G_{\mathcal P}
\right\}
-
\Phi_d(z)
\mathbb P\left\{
\widehat Y^{\mathrm{sel}}\in\mathcal A_M
\mid\mathcal G_{\mathcal P}
\right\}
\right|
\\\le&d_{TV}[\mathcal L(T_{M,N},\widehat Y^{sel}\mid\mathcal G_{\mathcal P}), N(0,I_d)\otimes \mathcal L(
\widehat Y^{\mathrm{sel}}\mid\mathcal G_{\mathcal P})] 
\\\le &d_{TV}[\mathcal L(T_{M,N},(\Sigma_{\mathcal P^c}^+)^{-1/2}(\hat Y^{sel}_{\mathcal P^c}-\mu_{\mathcal P^c})\mid\mathcal G_{\mathcal P}), N(0,I_d)\otimes N(0,I_{n_{\mathcal P^c}})] 
\\=&d_{TV}[N(\mu_K,\Gamma_{K}), N(0,I_d)\otimes N(0,I_{n_{\mathcal P^c}})] 
\\=&\Delta_N(\mathcal D_N,\mathcal P,Y_{\mathcal P},\widehat W_{\mathcal P})=o_p(1),
\end{aligned}
\]
where the first inequality is based on the definition of total variation. The second inequality is based on Lemma \ref{lem:conditional_tv_contraction}, and the equation in the third line is based on the distribution of \((T_{M,N},(\Sigma_{\mathcal P^c}^+)^{-1/2}(\hat Y^{sel}_{\mathcal P^c}-\mu_{\mathcal P^c})\mid\mathcal G_{\mathcal P})\) in \eqref{eq:Y_s_Y_i}.

Taking conditional expectations given $\mathcal D_N$ and then the supremum
over $z$ gives
\[
\begin{aligned}
&\sup_{z\in\mathbb R^d}
\left|
\mathbb P\left\{
T_{M,N}\le z,
\widehat Y^{\mathrm{sel}}\in\mathcal A_M
\mid\mathcal D_N
\right\}
-
\Phi_d(z)
\mathbb P\left\{
\widehat Y^{\mathrm{sel}}\in\mathcal A_M
\mid\mathcal D_N
\right\}
\right|
\\&=\sup_{z\in\mathbb R^d}
\left|
\mathbb E\left[\mathbb P\left\{
T_{M,N}\le z,
\widehat Y^{\mathrm{sel}}\in\mathcal A_M
\mid\mathcal G_{\mathcal P}
\right\}
-
\Phi_d(z)
\mathbb P\left\{
\widehat Y^{\mathrm{sel}}\in\mathcal A_M
\mid\mathcal G_{\mathcal P}
\right\}\mid \mathcal D_N\right]
\right|
\\&\le\sup_{z\in\mathbb R^d}
\mathbb E\left[\left|
\mathbb P\left\{
T_{M,N}\le z,
\widehat Y^{\mathrm{sel}}\in\mathcal A_M
\mid\mathcal G_{\mathcal P}
\right\}
-
\Phi_d(z)
\mathbb P\left\{
\widehat Y^{\mathrm{sel}}\in\mathcal A_M
\mid\mathcal G_{\mathcal P}
\right\}
\right| \mid\mathcal D_N\right]
\\
&\le
\sup_{z\in\mathbb R^d}\mathbb E(\Delta_N(\mathcal D_N,\mathcal P,Y_{\mathcal P},\widehat W_{\mathcal P})\mid\mathcal D_N).
\\&=
\mathbb E(\Delta_N(\mathcal D_N,\mathcal P,Y_{\mathcal P},\widehat W_{\mathcal P})\mid\mathcal D_N).
\end{aligned}
\]
As \(\Delta_N(\mathcal D_N,\mathcal P,Y_{\mathcal P},\widehat W_{\mathcal P})\overset p \to 0\), and \(0\le \Delta_N(\mathcal D_N,\mathcal P,Y_{\mathcal P},\widehat W_{\mathcal P})\le 1\), we have \[\mathbb E(\Delta_N(\mathcal D_N,\mathcal P,Y_{\mathcal P},\widehat W_{\mathcal P}))\to 0.\]
The Markov inequality then gives that \[\mathbb E(\Delta_N(\mathcal D_N,\mathcal P,Y_{\mathcal P},\widehat W_{\mathcal P})\mid\mathcal D_N)=o_p(1).\]

Assumption~5 gives \(\mathbb P\left\{
\widehat Y^{\mathrm{sel}}\in\mathcal A_M
\mid\mathcal D_N
\right\}>c_4\). Therefore, we have 
\[
\begin{aligned}
&\sup_{z\in\mathbb R^d}
\left|
\mathbb P\left\{
T_{M,N}\le z
\mid
\widehat Y^{\mathrm{sel}}\in\mathcal A_M,\mathcal D_N
\right\}
-
\Phi_d(z)
\right|
\\\le& \frac{1}{\mathbb P\left\{
\widehat Y^{\mathrm{sel}}\in\mathcal A_M
\mid\mathcal D_N
\right\}}\sup_{z\in\mathbb R^d}
\left|
\mathbb P\left\{
T_{M,N}\le z,
\widehat Y^{\mathrm{sel}}\in\mathcal A_M
\mid\mathcal D_N
\right\}
-
\Phi_d(z)
\mathbb P\left\{
\widehat Y^{\mathrm{sel}}\in\mathcal A_M
\mid\mathcal D_N
\right\}
\right|
\\\le &\frac{1}{c_4}E(\Delta_N(\mathcal D_N,\mathcal P,Y_{\mathcal P},\widehat W_{\mathcal P})\mid\mathcal D_N)
\\=&o_p(1),
\end{aligned}
\]
which completes the proof.
\end{proof}

\section{Variance Calculations for Signal Calibration}
\label{sec:signal_calibration}

Let $X=(X_1,\ldots,X_p)^\top$ follow the equicorrelated Gaussian
distribution specified in the simulation study, with
\[
\operatorname{Var}(X_j)=1,
\qquad
\operatorname{Cov}(X_j,X_k)=\rho
\quad (j\ne k).
\]
For the linear and nonlinear cases in Section~\ref{subsec:mean_stru}, $f_1$ and $f_2$
differ only by a relabeling of the covariates. Since the distribution of
$X$ is invariant under coordinate permutations, the two functions have
the same marginal distribution. It therefore suffices to calculate
\[
V_m(\rho)
=
\operatorname{Var}\{f_1(X)\}
=
\operatorname{Var}\{f_2(X)\}.
\]

For the linear case,  $f_1(X)=\sum_{k=1}^{10}X_k$, we obtain
\begin{equation}
V_m(\rho)
=
\sum_{k=1}^{10}\operatorname{Var}(X_k)
+
2\sum_{1\le j<k\le 10}\operatorname{Cov}(X_j,X_k)
=
10+90\rho.
\label{eq:Vm_linear}
\end{equation}

For the nonlinear case, define the cubic, linear, and interaction terms as
\[
U_{\mathrm{cub}}
=
\frac{1}{\sqrt{15}}\sum_{k=1}^{4}X_k^3,
\qquad
U_{\mathrm{lin}}
=
\sum_{k=5}^{8}X_k,
\qquad
U_{\mathrm{int}}
=
X_5X_9+X_6X_{10},
\]
so that
\[
f_1(X)
=
U_{\mathrm{cub}}+U_{\mathrm{lin}}+U_{\mathrm{int}}.
\]
The centered Gaussian distribution gives
\[
\mathbb E(U_{\mathrm{cub}})
=
\mathbb E(U_{\mathrm{lin}})
=
0,
\qquad
\mathbb E(U_{\mathrm{int}})=2\rho.
\] 
By Lemma \ref{lem:gaussian_moments}, we have
\[
\operatorname{Cov}(U_{\mathrm{cub}},U_{\mathrm{int}})
=
\operatorname{Cov}(U_{\mathrm{lin}},U_{\mathrm{int}})
=
0.
\]
By Lemma \ref{lem:gaussian_moments}, we also calculate that for
$j\ne k$,
\begin{align*}
\mathbb E(X_j^6)&=15,
&
\mathbb E(X_j^3X_k)&=3\rho,
\\
\mathbb E(X_j^3X_k^3)&=9\rho+6\rho^3,
&
\mathbb E(X_j^2X_k^2)&=1+2\rho^2.
\end{align*}
Consequently,
\begin{align*}
\operatorname{Var}(U_{\mathrm{cub}})
&=
\frac{1}{15}
\left\{
4\cdot15
+
2\binom{4}{2}(9\rho+6\rho^3)
\right\}
=
4+\frac{36}{5}\rho+\frac{24}{5}\rho^3,
\\
\operatorname{Var}(U_{\mathrm{lin}})
&=
4+2\binom{4}{2}\rho
=
4+12\rho,
\\
\operatorname{Cov}(U_{\mathrm{cub}},U_{\mathrm{lin}})
&=
\frac{1}{\sqrt{15}}
\sum_{j=1}^{4}\sum_{k=5}^{8}\mathbb E(X_j^3X_k)
=
\frac{48}{\sqrt{15}}\rho.
\end{align*}
For interaction terms, the mean is $\rho$ and the variance is
\[
\operatorname{Var}(X_jX_k)
=
\mathbb E(X_j^2X_k^2)-\{\mathbb E(X_jX_k)\}^2
=
1+\rho^2
\quad (j\ne k),
\]
and the covariance term is
\[
\operatorname{Cov}(X_5X_9,X_6X_{10})
=\mathbb E(X_5X_9X_6X_{10})-\mathbb E(X_5X_9)\mathbb E(X_6X_{10})=
3\rho^2-\rho^2
=
2\rho^2.
\]
It follows that
\[
\operatorname{Var}(U_{\mathrm{int}})
=
2(1+\rho^2)+2(2\rho^2)
=
2+6\rho^2.
\]
Combining these expressions yields
\begin{align}
V_m(\rho)
&=
\operatorname{Var}(U_{\mathrm{cub}})
+
\operatorname{Var}(U_{\mathrm{lin}})
+
\operatorname{Var}(U_{\mathrm{int}})
+
2\operatorname{Cov}(U_{\mathrm{cub}},U_{\mathrm{lin}})
\notag\\
&=
10+
\left(\frac{96}{5}+\frac{96}{\sqrt{15}}\right)\rho
+
6\rho^2
+
\frac{24}{5}\rho^3.
\label{eq:Vm_nonlinear}
\end{align}
The total variance for both linear and nonlinear case satisfies $V_m(0)=10$.
Substituting \eqref{eq:Vm_linear} or \eqref{eq:Vm_nonlinear} into
\eqref{eq:simulation_calibration} gives the corresponding values of
$\sigma_u^2$ and $\beta_v$.

\begin{lemma}[Isserlis's theorem]
\label{lem:gaussian_moments}
Let $X=(X_1,\ldots,X_p)^\top$ be a centered Gaussian random vector with
covariance matrix $\Sigma$. For any indices
$j_1,\ldots,j_{2m}$, 
\[
\mathbb E\left(
\prod_{\ell=1}^{2m}X_{j_\ell}
\right)
=
\sum_{\pi\in\mathcal P_{2m}}
\prod_{(a,b)\in\pi}
\Cov(X_{j_a},X_{j_b}),
\]
where $\mathcal P_{2m}$ denotes the collection of all pairings of
$\{1,\ldots,2m\}$. All odd-order moments are zero.
\end{lemma}

\begin{proof}
This result can be found in \citet{isserlis1918formula}.
\end{proof}

\section{Additional Simulation Results}

\begin{figure}[ht!]
	\centering
	\includegraphics[width=\textwidth]{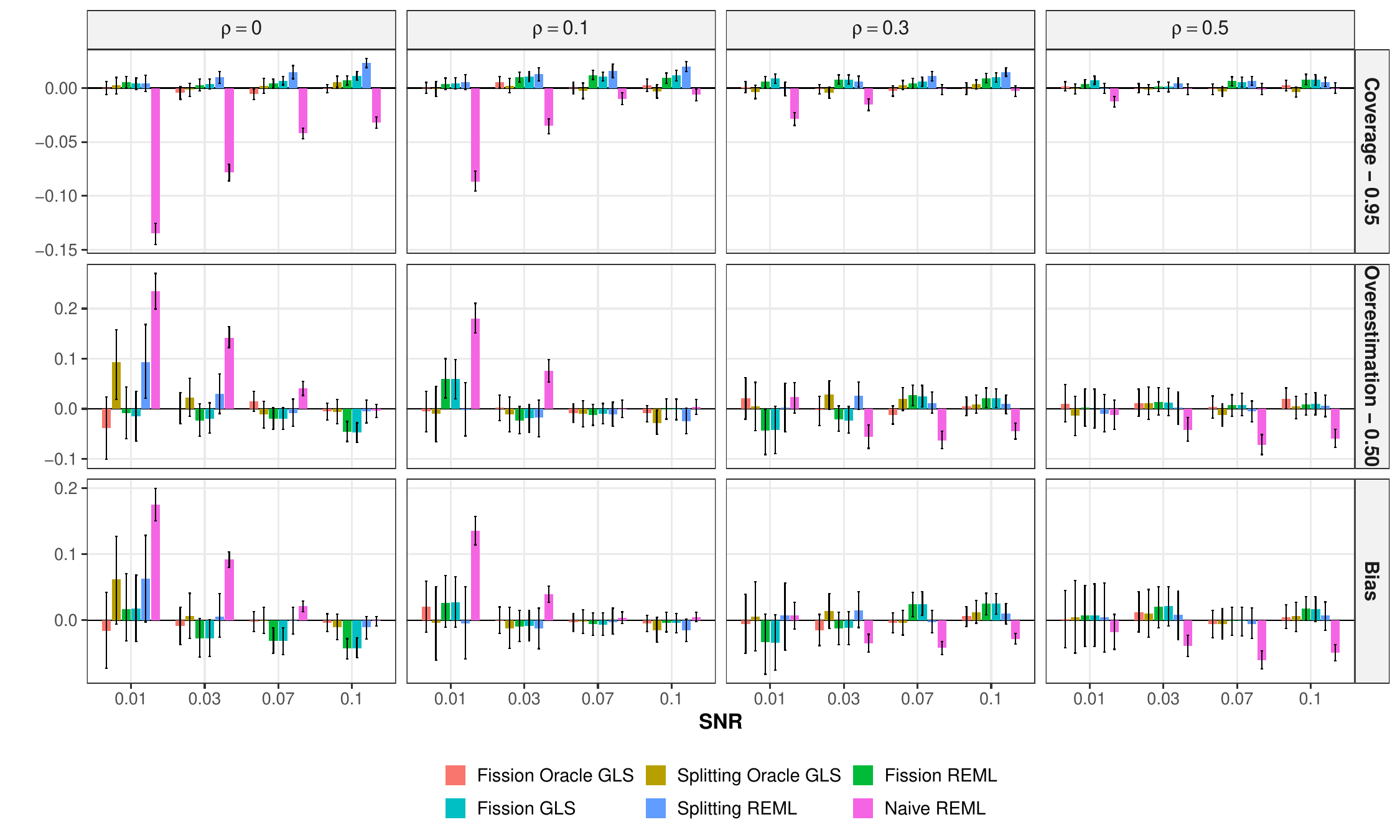}
	\includegraphics[width=\textwidth]{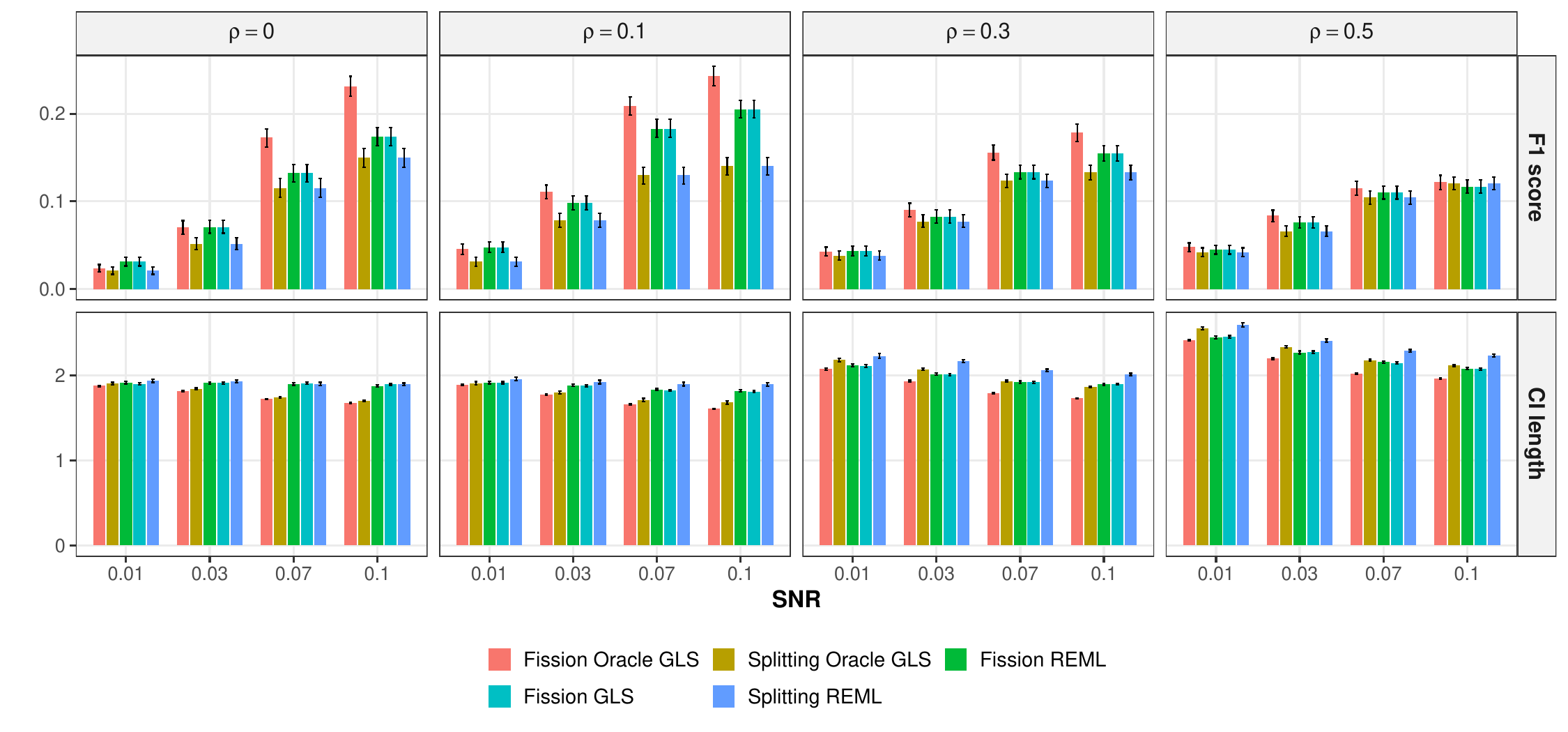}
	\caption{Inferential performance and selection accuracy of the neural-network selector under the linear mean trajectory with $N=200$ and $p=500$.}
	\label{fig:nn-linear-N200-p500}
\end{figure}

\begin{figure}[ht!]
	\centering
	\includegraphics[width=\textwidth]{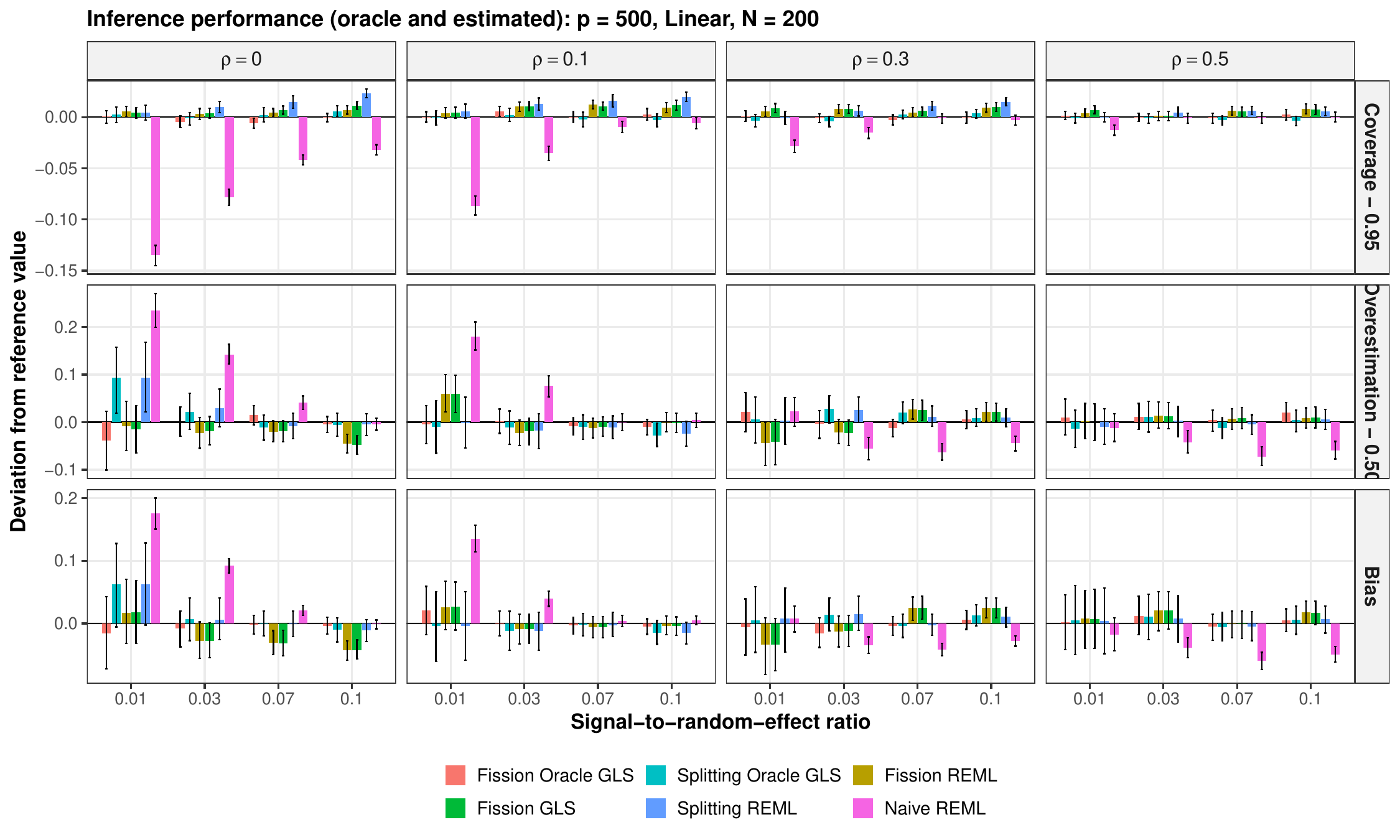}
	\includegraphics[width=\textwidth]{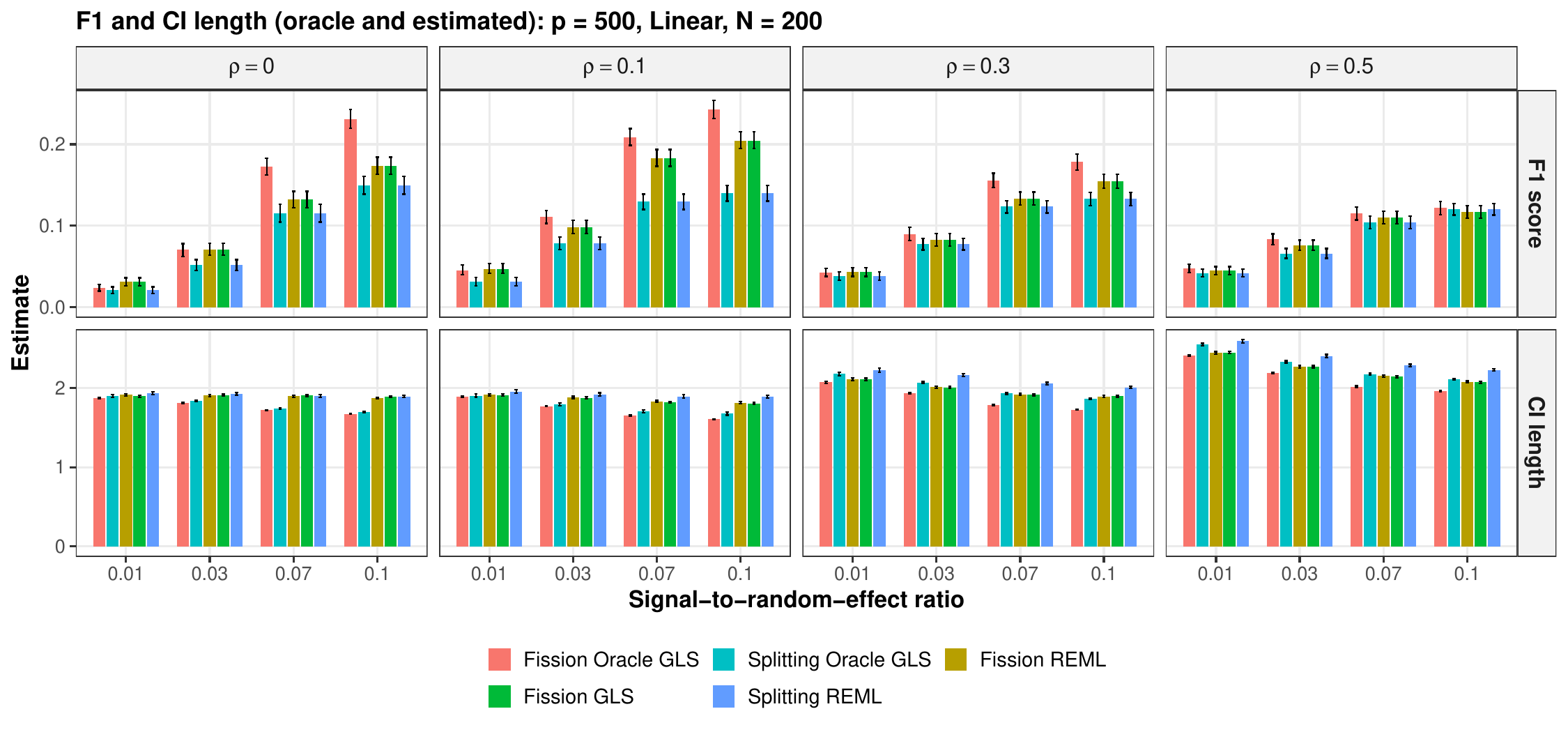}
	\caption{Inferential performance and selection accuracy of the mixed-effects-tree selector under the linear mean trajectory with $N=200$ and $p=500$.}
	\label{fig:tree-linear-N200-p500}
\end{figure}

\begin{figure}[ht!]
	\centering
	\includegraphics[width=\textwidth]{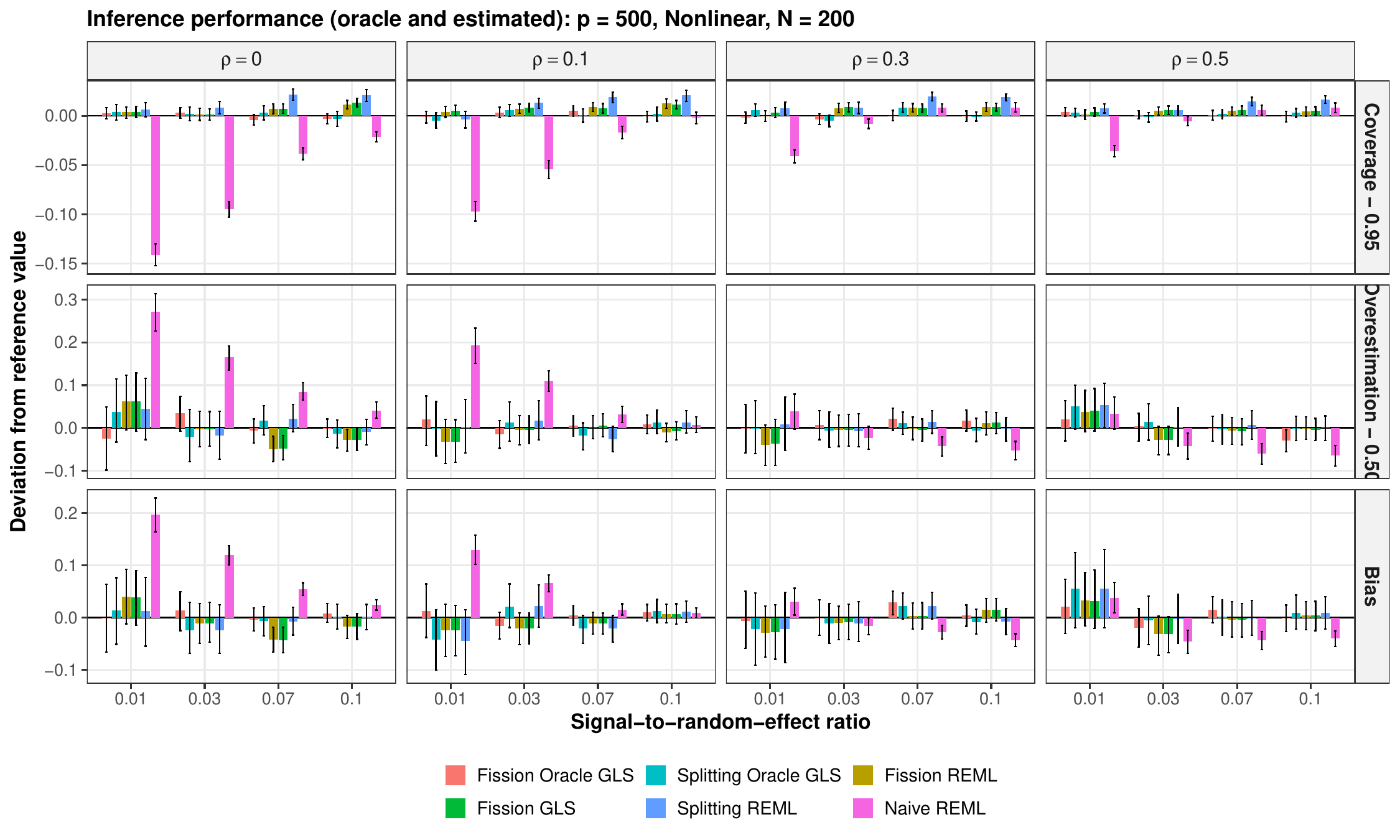}
	\includegraphics[width=\textwidth]{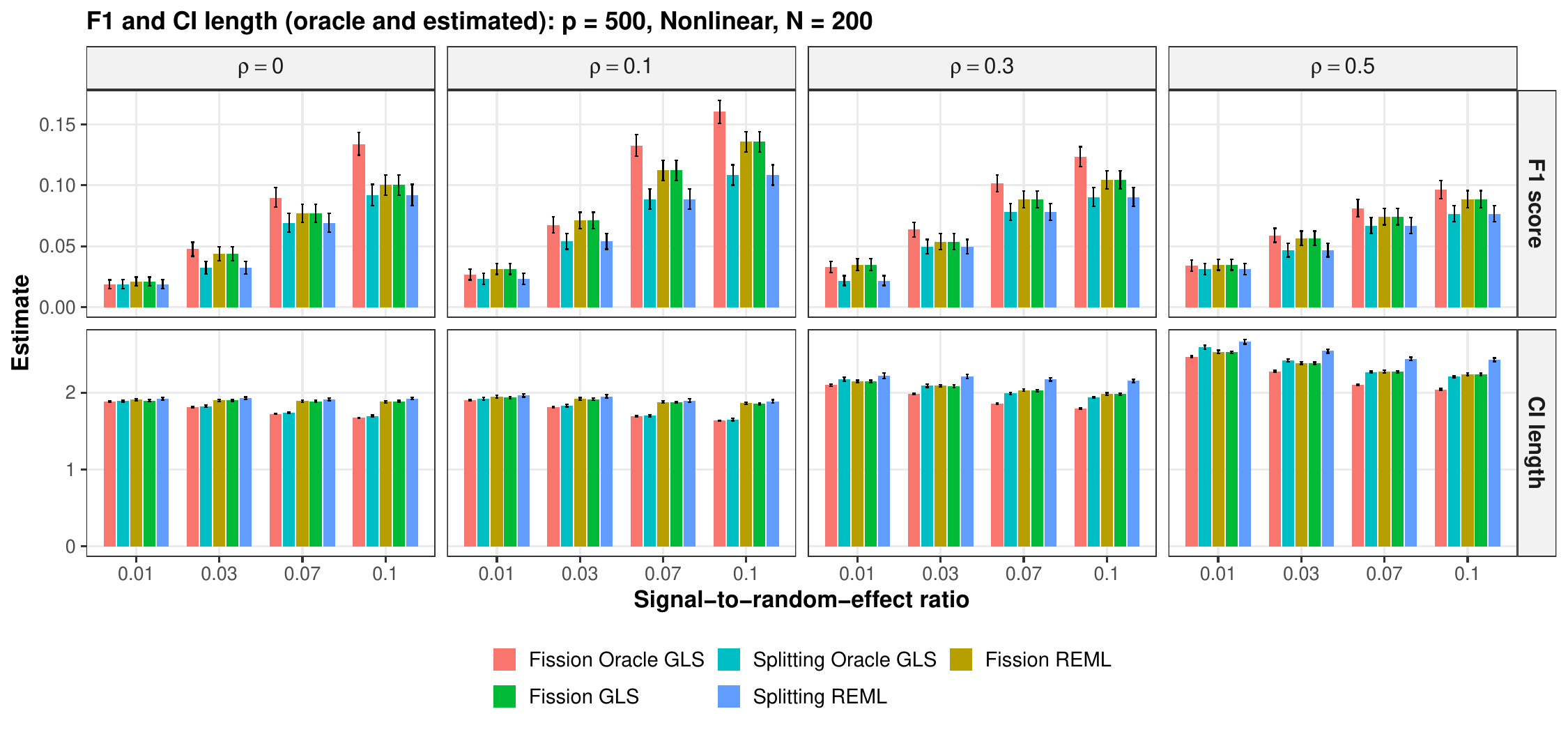}
	\caption{Inferential performance and selection accuracy of the mixed-effects-tree selector under the nonlinear mean trajectory with $N=200$ and $p=500$.}
	\label{fig:tree-nonlinear-N200-p500}
\end{figure}
\endgroup

\clearpage
\bibliographystyle{plainnat}
\bibliography{ref}
\end{document}